\documentclass[amsmath,amssymb,nofootinbib,reprint,usenames,dvipsnames]{revtex4-1}

\usepackage{graphicx}
\usepackage{dcolumn}
\usepackage{bm}
\usepackage{tree-dvips}
\usepackage[english]{babel}
\usepackage{array}
\usepackage{eqnarray,amsmath}
\usepackage{bbold}
\usepackage{esvect}
\usepackage{braket}
\usepackage{wrapfig}
\usepackage{graphicx}
\usepackage{dblfloatfix}
\usepackage{fixltx2e}
\usepackage[font=small,labelfont=bf,justification=justified,format=plain]{caption}
\usepackage{float}
\usepackage{subcaption}\usepackage{makeidx}
\usepackage{lipsum}
\usepackage{fancyhdr}
\usepackage{ragged2e}
\usepackage{amsmath}
\usepackage{amssymb}
\usepackage{subcaption}
\usepackage{amsthm}
\usepackage[thinlines]{easytable}
\usepackage{enumitem}
\usepackage{pagecolor}
\usepackage{tikz}
\usepackage{url}
\usepackage[utf8]{inputenc}
\usepackage[toc,page]{appendix}
\usepackage{sidecap}
\usepackage{calrsfs}

\usepackage[colorlinks=true,citecolor=red,linkcolor=brown,urlcolor=magenta]{hyperref}

\newtheorem{theorem}{Theorem}
\newtheorem{lemma}{Lemma}

\newtheorem{mydef}{Definition}

\newtheorem{sub-section}{}

\usetikzlibrary{shapes.misc, positioning}

\tikzset{cross/.style={cross out, draw=black, minimum size=2*(#1-\pgflinewidth), inner sep=0pt, outer sep=0pt},
cross/.default={1pt}}

\DeclareMathAlphabet{\pazocal}{OMS}{zplm}{m}{n}
\newcommand{\E}{\pazocal{E}}

\newcommand{\F}{\pazocal{F}}
\newcommand{\R}{\pazocal{R}}
\newcommand{\h}{\pazocal{H}}

\newcommand{\C}{\pazocal{C}}
\newcommand{\p}{\pazocal{P}}

\newcommand{\s}{\pazocal{S}}

\newcommand{\Q}{\pazocal{Q}}
\newcommand{\V}{\pazocal{V}}
\newcommand{\U}{\pazocal{U}}
\newcommand{\M}{\pazocal{M}}

\newcommand{\Z}{\pazocal{Z}}
\newcommand{\X}{\pazocal{X}}
\newcommand{\Y}{\pazocal{Y}}
\newcommand{\I}{\pazocal{I}}

\definecolor{yellow-green}{rgb}{0.6, 0.8, 0.2}
\definecolor{ufogreen}{rgb}{0.24, 0.82, 0.44}

\newcommand{\be}{\begin{equation}}
\newcommand{\ee}{\end{equation}}
\newcommand{\ben}{\begin{eqnarray}}
\newcommand{\een}{\end{eqnarray}}
\newcommand{\bes}{\begin{subequations}}
\newcommand{\ees}{\end{subequations}}
\newcommand{\bF}{\begin{figure}}
\newcommand{\eF}{\end{figure}}

\begin{document}

\title{Accrediting outputs of noisy intermediate-scale quantum computing devices}

\author{Samuele Ferracin}
\author{Theodoros Kapourniotis}
\author{Animesh Datta}
\affiliation{Department of Physics, University of Warwick, Coventry CV4 7AL, United Kingdom}

\begin{abstract}
We present an accreditation protocol for the outputs of noisy intermediate-scale quantum devices. By testing entire circuits rather than individual gates, our accreditation protocol can provide an upper-bound on the variation distance between noisy and noiseless probability distribution of the outputs of the target circuit of interest. Our accreditation protocol requires implementing quantum circuits no larger than the target circuit, therefore it is practical in the near term and scalable in the long term. Inspired by trap-based protocols for the verification of quantum computations, our accreditation protocol assumes that single-qubit gates have bounded probability of error. We allow for arbitrary spatial and temporal correlations in the noise affecting state preparation, measurements, single-qubit and two-qubit gates. We describe how to implement our protocol on real-world devices, and we also present a novel cryptographic protocol (which we call ``mesothetic'' protocol) inspired by our accreditation protocol.
\end{abstract}

\date{\today}

\maketitle

\section*{Introduction}
\noindent Quantum computers promise to expand our computing capabilities beyond their current horizons. Several commercial institutions \cite{Google,IBM,Rigetti} are taking steps towards building the first prototypes of quantum computers that can outperform existing supercomputers in certain tasks \cite{AA11,BMS16,GWD17,T18,DHKL18,N&al17}, the so-called ``Noisy Intermediate-Scale Quantum'' (NISQ) computing devices \cite{Preskill18}. As all their internal operations such as state preparations, gates, and measurements are by definition noisy, the outputs of computations implemented on NISQ devices are unreliable. It is thus essential to devise protocols able to accredit these outputs.

A commonly employed approach involves simulating the quantum circuit whose output we wish to accredit, the target circuit, on a classical computer. This is feasible for small circuits, as well as for circuits composed of Clifford gates \cite{G98} and few non-Clifford gates \cite{BG16,B&al17}. Classical simulations have been performed for quantum computations of up to 72 qubits, often exploiting subtle insights into the nature of specific quantum circuits involved \cite{B&al16,Villalonga&al19}. Though practical for the present, classical simulations of quantum circuits are not scalable. Worthwhile quantum computations will not be efficiently simulable on classical computers, hence we must seek for alternative methods.

Another approach employed in experiments consists of individually testing classes of gates present in the target circuit. This is typically undertaken using a family of protocols centered around randomized benchmarking and its extensions \cite{EAZ05,K&al07,DCEL09,MGE11,Erhard&al19,CGFF17}. These protocols allow extraction of the fidelity of gates or cycles of gates and can witness progresses towards fault-tolerant quantum computing \cite{SWS16}. However they rely on assumptions that may be invalid in experiments. In particular, they require the noise to be Markovian and cannot account for temporal correlations \cite{Wallman17,MPF18}. Quantum circuits are more than the sum of their gates, and the noise in the target circuit may exhibit characteristics that cannot be captured by benchmarking its individual gates independently.

This calls for protocols able to test circuits as a whole rather than individual gates. Such protocols have been devised inspired by Interactive Proof Systems \cite{GKK17}. In these protocols (which we call ``cryptographic protocols'') the outputs of the target circuit are verified through an interaction between a trusted verifier and an untrusted prover (Figure \ref{fig:cryptop}). The verifier is typically allowed to possess a noiseless {quantum} device able to prepare \cite{ABE08,FK12,BFKW13,KD17,KW17,FKD17,ABEM17,B15} or measure \cite{HM15,MF16,HKSE17,MK18,TMMMF18} single qubits, however recently a protocol for a fully classical verifier was devised that relies on the widely believed intractability of a computational problem for quantum computers \cite{UM18}. Other protocols for classical verifiers have also been devised, but they require interaction with multiple entangled and non-communicating provers \cite{RUV12,GKW15,M16,FH15,NV16}. Cryptographic protocols show that with minimal assumptions, verification of the outputs of quantum computations of arbitrary size can be done efficiently, in principle.

In practice, implementing cryptographic protocols in experiments remains challenging, especially in the near term. In experiments all the operations are noisy, as in Figure \ref{fig:physicsp}, and the verifier does not possess noiseless quantum devices. Thus, the verifiability of protocols requiring noiseless devices for the verifier is not guaranteed. Moreover, the concept of \textit{scalability}, which is of primary interest in cryptographic protocols, is not equivalent to that of \textit{practicality}, which is essential for experiments. For instance, suppose that the target circuit contains a few hundred qubits and a few hundred gates. Cryptographic protocols require implementing this circuit on a large cluster state containing thousands of qubits and entangling gates \cite{ABE08,FK12,BFKW13,KD17,KW17,FKD17,HM15} 
	or on two spatially-separated devices sharing thousands of copies of Bell states \cite{RUV12,GKW15};
	or appending several teleportation gadgets to the target circuit (one for each $T$-gate in the circuit and six for each Hadamard gate) \cite{B15}; 
		or building Feynman-Kitaev clock states, which require entangling the system with an auxiliary qubit per gate in the target circuit \cite{FH15,MF16,UM18}. These protocols are scalable, as they require a number of additional qubits, gates and measurements growing linearly with the size of the target circuit, yet they remain impractical for NISQ devices.

In this paper we present an accreditation protocol that provides an upper-bound on the variation distance between noisy and noiseless probability distribution of the outputs of a NISQ device, under the assumption that
\begin{itemize}[leftmargin=0.7cm]
\item[{N1}:] Noise in state preparation, entangling gates, and measurements is an arbitrary Completely Positive Trace Preserving (CPTP) map encompassing the whole system and the environment (Equation \ref{eq:stateend});
\item[{N2}:] Noise in single-qubit gates is a CPTP map $\F_{SE}$ of the form $\F_{SE}=(1-r)\I+r\F^{\prime}_{SE}$ with $0\leq r<1$, where $\I_{SE}$ is the identity on system and environment and $\F^{\prime}_{SE}$ is an arbitrary (potentially gate-dependent) CPTP map encompassing the whole system and the environment.
\end{itemize}
Inspired by cryptographic protocols  \cite{ABE08,FK12,BFKW13,B15,KD17,KW17,FKD17,ABEM17,HM15,MF16,HKSE17,MK18,TMMMF18,UM18,RUV12,GKW15,M16,FH15,NV16}, our accreditation protocol is trap-based, meaning that the target circuit being accredited is implemented together with a number $v$ of classically simulable circuits (the  ``trap'' circuits) able to detect all types of noise subject to conditions {N1} and {N2} above.  
	A single run of our protocol requires implementing the target circuit being accredited and $v$ trap circuits.
	It provides a {binary} outcome in which the outputs of the target circuit are either accepted as \textit{correct} (with confidence increasing linearly with $v$) or rejected as \textit{potentially incorrect}. More usefully, consider running our protocol $d$ times, each time with the same target and $v$ potentially different trap circuits. Suppose that the output of the target is accepted as correct by $N_{\textup{acc}}>0$ runs. With confidence $1-2$exp$(-2d\theta^2)$, for each of these accepted outputs our protocol ensures 
\begin{align}
\label{eq:eq1}
\frac{1}{2}\sum_{\overline{s}}\big|p_{\textup{noiseless}}(\overline{s})-p_{\textup{noisy}}(\overline{s})\big|\leq\frac{\varepsilon}{N_{\textup{acc}}/d-\theta}\textup{ ,}
\end{align} 
where $\theta\in(0,N_{\textup{acc}}/d)$ is a tunable parameter that affects both the confidence and the upper-bound, $p_{\textup{noiseless}}(\overline{s})$ and $p_{\textup{noisy}}(\overline{s})$ are the noiseless and noisy probability distributions of the outputs $\{\overline{s}\}$ of the target circuit respectively and $\varepsilon \propto 1/v$. Bounds of this type can fruitfully accredit the outputs of experimental quantum computers as well as underpin attempts at demonstrating and verifying quantum supremacy in sampling experiments \cite{AA11,BMS16,GWD17,T18,DHKL18,N&al17}.

Crucially, our accreditation protocol is both experimentally practical and scalable: all circuits implemented in our protocol are no wider (in the number of qubits) or deeper (in the number of gates) than the circuit we seek to accredit. This makes our protocol more readily implementable on NISQ devices than cryptographic protocols. Moreover, our protocol requires no noiseless quantum device, and it only relies on the assumption that the single-qubit gates suffer bounded (but potentially non-local in space and time and gate-dependent) noise|condition {N2}. This assumption is motivated by the empirical observation that single-qubit gates are the

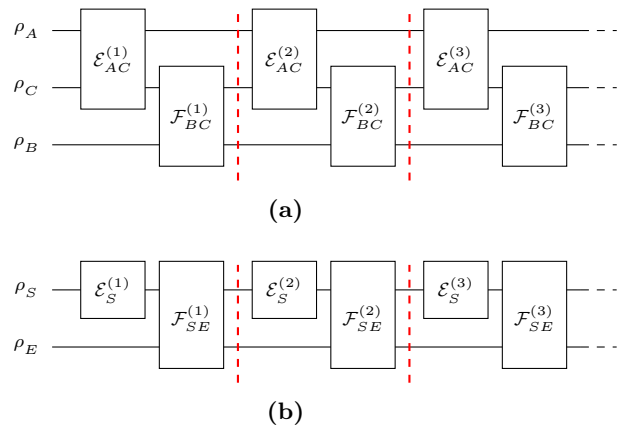
\begin{figure}[H]
\begin{subfigure}{0.5\textwidth}
\centering
\begin{tikzpicture}[scale=0.95, every node/.style={scale=0.99}]

\node at (-1.75,0.0) {\scriptsize $\rho_{_B}$};
\node at (-1.75,0.8) {\scriptsize $\rho_{_C}$};
\node at (-1.75,1.6) {\scriptsize $\rho_{_A}$};

\foreach \x in {3,...,5}
\draw (-1.4,4.0-\x*0.8) -- (6,4.0-\x*0.8);

\foreach \x in {3,...,5}
\draw [dashed] (6,4.0-\x*0.8) -- (6.5,4.0-\x*0.8);

\draw [fill=white] (1,0.0-0.3) rectangle (0.1,0.8+0.3);
\node at (-0.55+1.1,0.4) {\scriptsize $\F_{BC}^{(1)}$};

\draw [fill=white] (-1,0.8-0.3) rectangle (-0.1,1.6+0.3);
\node at (-0.55,0.4+0.8) {\scriptsize $\E_{AC}^{(1)}$};

\draw [dashed, red, thick] (1.2,.0-0.5) -- (1.2,1.9);
\draw [dashed, red, thick] (3.6,.0-0.5) -- (3.6,1.9);

\draw [fill=white] (1+2.4,0.0-0.3) rectangle (0.1+2.4,0.8+0.3);
\node at (-0.55+2.4+1.1,0.4) {\scriptsize $\F_{BC}^{(2)}$};

\draw [fill=white] (1.4,0.8-0.3) rectangle (-0.1+2.4,1.6+0.3);
\node at (-0.55+2.4,0.4+0.8) {\scriptsize $\E_{AC}^{(2)}$};

\draw [fill=white] (1+2.4+2.4,0.0-0.3) rectangle (0.1+2.4+2.4,0.8+0.3);
\node at (-0.55+2.4+2.4+1.1,0.4) {\scriptsize $\F_{BC}^{(3)}$};

\draw [fill=white] (1.4+2.4,0.8-0.3) rectangle (-0.1+2.4+2.4,1.6+0.3);
\node at (-0.55+2.4+2.4,0.4+0.8) {\scriptsize $\E_{AC}^{(3)}$};

\end{tikzpicture}
\caption{\small \textcolor{white}{ciao\\ciao}}
\label{fig:cryptop}
\end{subfigure}
~
\begin{subfigure}{0.5\textwidth}
\centering
\begin{tikzpicture}[scale=0.95, every node/.style={scale=0.99}]

\node at (-1.75,0.0) {\scriptsize $\rho_{_E}$};
\node at (-1.75,0.8) {\scriptsize $\rho_{_S}$};

\foreach \x in {4,...,5}
\draw (-1.4,4.0-\x*0.8) -- (6,4.0-\x*0.8);

\foreach \x in {4,...,5}
\draw [dashed] (6,4.0-\x*0.8) -- (6.5,4.0-\x*0.8);

\draw [fill=white] (1,0.0-0.3) rectangle (0.1,0.8+0.4);
\node at (-0.55+1.1,0.4) {\scriptsize $\F_{SE}^{(1)}$};

\draw [fill=white] (-1,0.8-0.4) rectangle (-0.1,0.8+0.4);
\node at (-0.55,0.8) {\scriptsize $\E_{S}^{(1)}$};

\draw [dashed, red, thick] (1.2,.0-0.5) -- (1.2,1.9-0.75);
\draw [dashed, red, thick] (3.6,.0-0.5) -- (3.6,1.9-0.75);

\draw [fill=white] (1+2.4,0.0-0.3) rectangle (0.1+2.4,0.8+0.4);
\node at (-0.55+2.4+1.1,0.4) {\scriptsize $\F_{SE}^{(2)}$};

\draw [fill=white] (1.4,0.8-0.4) rectangle (-0.1+2.4,0.8+0.4);
\node at (-0.55+2.4,0.8) {\scriptsize $\E_{S}^{(2)}$};

\draw [fill=white] (1+2.4+2.4,0.0-0.3) rectangle (0.1+2.4+2.4,0.8+0.4);
\node at (-0.55+2.4+2.4+1.1,0.4) {\scriptsize $\F_{SE}^{(3)}$};

\draw [fill=white] (1.4+2.4,0.8-0.4) rectangle (-0.1+2.4+2.4,0.8+0.4);
\node at (-0.55+2.4+2.4,0.8) {\scriptsize $\E_{S}^{(3)}$};

\end{tikzpicture}
\caption{\small \textcolor{white}{ciao}}
\label{fig:physicsp}
\end{subfigure}
\caption{\small \textbf{(a)} In cryptographic protocols a verifier A and a prover B apply operations on their own registers and on a shared register C. 
			\textbf{(b)} {In accreditation protocols all operations applied to the system S are noisy. Noise couples the system to an environment E.}
			}
\label{fig:protocols}
\end{figure}

\noindent most accurate operations in prominent quantum computing platforms such as trapped ions \cite{H&al14,H&al16} and superconducting qubits \cite{IBM,B&al14,SMCG16}.

In addition to its ready implementability on NISQ devices, our accreditation protocol can detect all types of noise typically considered by techniques centered around randomized benchmarking and its extensions \cite{EAZ05,K&al07,DCEL09,MGE11,Erhard&al19,CGFF17}. Moreover it can detect noise that may be missed by those techniques such as noise correlated in time. Mathematically, this amounts to allowing noisy operations to encompass both system and environment (Figure \ref{fig:physicsp}) and tracing out the environment only at the end of the protocol. This noise model is more general than the Markovian noise model considered in protocols centered around randomized benchmarking \cite{Wallman17,MPF18}. Moreover, by testing circuits rather than gates, our protocol ensures that all possible noise (subject to condition {N1} and {N2}) in state preparation, measurement and gates is detected, even noise that arises only when these components are put together to form a circuit. On the contrary, benchmarking isolated gates can sometimes yield over-estimates of their fidelities \cite{CGFF17}, and consequently of the fidelity of the resulting circuit. We note that noise of the type {N2} excludes unbounded gate-dependent errors in single-qubit gates such as systematic over- or under-rotations, as also is the case for other works \cite{EAZ05,K&al07,DCEL09,MGE11,Erhard&al19,CGFF17}.

Inspired by our accreditation protocol we also present a novel cryptographic protocol, which we call ``mesothetic verification protocol''. In the mesothetic protocol the verifier implements the single-qubit gates in all circuits while the prover undertakes all other operations.  This is distinct from prepare-and-send \cite{ABE08,FK12,BFKW13,B15,KD17,KW17,FKD17,ABEM17} or receive-and-measure \cite{HM15,MF16,HKSE17,MK18,TMMMF18} cryptographic protocols in that the verifier  intervenes during the actual implementation of the circuits, and not before or after the circuits are implemented.

Our paper is organized as follows. In Section 1 of Results we introduce the notation, in Section 2 we provide the necessary definitions, in Sections 3 and 4 we present our protocol and prove our results and in Section 5 we present the mesothetic verification protocol.

\section*{Results}
\noindent{\textbf{1. Notation: }}
We indicate unitary matrices acting on the system with capital letters such as $U,V$ and $W$, and Completely Positive Trace Preserving (CPTP) maps with calligraphic letters such as $\E,\F,\R$ and $\M$. We indicate the $2\times2$ identity matrix as $I$, the single-qubit Pauli gates as $X,Y,Z$, the controlled-$Z$ gate as $cZ$, the controlled-$X$ gate as $cX$, the Hadamard gate as $H$ and $S=\textup{diag}(1,i)$. The symbol $\circ$ denotes the composition of CPTP maps: $\circ_{p=1}^q\E_p(\rho)=\E_q\ldots \E_1(\rho)$, Tr${}_E[\textup{ }\cdot\textup{ }]$ is the trace over the environment, $D(\sigma,\tau)=\textup{Tr}|\sigma-\tau|/2$ is the trace distance between the states $\sigma$ and $\tau$. We say that a noisy implementation $\widetilde{\E}$ of $\E$ suffers bounded noise if $\widetilde{\E}$ can be written as $\widetilde{\E}=(1-r)\E+r\F$  for some CPTP map $\F$ and number $0\leq r<1$, otherwise if $r=1$ we say that the noise is unbounded \cite{AKN98}.\\

\noindent\noindent{\textbf{2. Background:}} We start by defining our notion of protocol:
\begin{mydef}
\label{def:protocol}
\textup{\textbf{[Protocol]. }}Consider a system $S$ in the state $\rho_S$. A protocol on input $\rho_S$ is a collection of CPTP maps $\{\E^{(p)}_S\}_{p=1}^q$ acting on $S$ and yielding the state $\rho_{\textup{out}}=\circ_{p=1}^q\E^{(p)}_S(\rho_{S})$.
\end{mydef}
\noindent When implemented on real devices protocols suffer the effects of noise. Modeling noise as a set $\{\F^{(p)}_{SE}\}$ of CPTP maps acting on system and environment (Figure \ref{fig:physicsp}), the state of the system at the end of a noisy protocol run is
\begin{align}
{\rho}_{\textup{out}}=\textup{Tr}_E\big[\circ_{p=1}^q\F^{(p)}_{SE}\big(\E^{(p)}_S\otimes \I_E)(\rho_{S}\otimes\rho_E)\big]\textup{ ,}\label{eq:stateend}
\end{align}
where $\rho_E$ is the state of the environment at the beginning of the protocol. We allow each map $\F^{(p)}_{SE}$ to depend arbitrarily on the corresponding operation $\E^{(p)}_{S}$. 

A trap-based accreditation protocol is defined as follows. A single run of such a protocol takes as input a classical description of the target circuit and a number $v$, implements $v+1$ circuits (the target and $v$ traps) and returns the outputs of the target circuit, together with a ``flag bit'' set to ``$\textup{acc}$'' ($``\textup{rej}$'') indicating that the output of the target must be accepted (rejected). Formally,
\begin{mydef}
\label{def:accreditation}
\textup{\textbf{[Trap-Based Accreditation Protocol]. }}Consider a protocol $\{\E^{(p)}_S\}_{p=1}^q$ with input $\rho_S$, where $\rho_S$ contains a classical description of the target circuit and the number $v$ of trap circuits. Consider also a set of CPTP maps $\{\F^{(p)}_{SE}\}_{p=1}^q$ (the noise) acting on system and environment. We say that the protocol $\{\E^{(p)}_S\}_{p=1}^q$ can accredit the outputs of the target circuit in the presence of noise $\{\F^{(p)}_{SE}\}_{p=1}^q$ if the following two properties hold:
\begin{itemize}[leftmargin=0.4cm]
\item[1)] The state of the system at the end of a single protocol run (Equation \ref{eq:stateend}) can be expressed as
\begin{eqnarray}
\label{eq:stateend2}
 {\rho}_{\textup{out}}& = &b\textup{ }\tau_{\textup{out}}^{\prime\textup{ tar}}\otimes|\textup{acc}\rangle\langle\textup{acc}|  	\\								  
			& + &  (1-b)\bigg(l\textup{ }\sigma_{\textup{out}}^{\textup{tar}}\otimes|\textup{acc}\rangle\langle\textup{acc}|+(1-l)\tau_{\textup{out}}^{\textup{tar}}\otimes|\textup{rej}\rangle\langle\textup{rej}|\bigg)\cr  \nonumber
\end{eqnarray}
where $\sigma_{\textup{out}}^{\textup{tar}}$ ($\tau_{\textup{out}}^{\prime\textup{ tar}}$) is the state of the target circuit at the end of a noiseless (noisy) protocol run, $\tau_{\textup{out}}^{\textup{tar}}$ is an arbitrary state for the target circuit, $|\textup{acc}\rangle$ is the state of the flag indicating acceptance, $|\textup{rej}\rangle=|\textup{acc}\oplus 1\rangle$, $0\leq l\leq1$, $0\leq b\leq\varepsilon$ and $\varepsilon\in[0,1]$.
\item[2)] After d protocol runs with the same target circuit and v potentially different trap circuits, if all these runs are affected by independent and identically distributed (i.i.d.) noise, then the variation distance between noisy and noiseless probability distribution of the outputs of each of the $N_{\textup{acc}}\in[0,d]$ protocol runs ending with flag bit in the state $|\textup{acc}\rangle$ is upper-bounded as in Equation \ref{eq:eq1}.
\end{itemize}
\end{mydef}
\noindent Property 1 ensures that the probability of accepting the outputs of a single protocol run when the target circuit is affected by noise (the number $b$ in Equation \ref{eq:stateend2}) is smaller than a constant $\varepsilon$. The constant $\varepsilon$ is a function of the number of trap circuits, of the protocol and of the noise model and is to be computed analytically. The quantity $1-\varepsilon$ quantifies the \textit{credibility} of the accreditation protocol.

Note that Property 1 in the above definition implies Property 2. To see this, assume Property 1 is valid for a given protocol. Suppose that this protocol is run $d$ times with i.i.d. noise (a standard assumption in trap-based cryptographic protocols \cite{KD17,HKSE17}) and suppose that $N_{\textup{acc}}>0$ protocol runs end with flag bit in the state $|\textup{acc}\rangle$. For each of these $N_{\textup{acc}}$ runs, the state of the system at the end of the protocol run is thus of the form (cfr. Equation \ref{eq:stateend2})
\begin{align}
{\rho}_{\textup{out, acc}}&=\frac{(1-b)l\textup{ }\sigma_{\textup{out}}^{\textup{tar}}+b\textup{ }\tau_{\textup{out}}^{\prime\textup{ tar}}}{(1-b)l+b}\otimes|\textup{acc}\rangle\langle\textup{acc}|
\end{align}
This yields a bound on the variation distance of the type \cite{NC00}
\begin{align}
\label{eq:bbbboundvardistt}
\frac{1}{2}\sum_{\overline{s}}&\big|p_{\textup{noiseless}}(\overline{s})-p_{\textup{noisy}}(\overline{s})\big|\cr
&\leq\textup{ }D\bigg({\sigma}^{\textup{tar}}_{\textup{out}}\textup{ , }\frac{(1-b)l\textup{ }\sigma_{\textup{out}}^{\textup{tar}}+b\textup{ }\tau_{\textup{out}}^{\prime\textup{ tar}}}{(1-b)l+b}\bigg)\cr
&\leq\frac{b}{(1-b)l+b}\textup{ }\leq\textup{}\frac{\varepsilon}{\textup{prob(acc)}}\textup{ ,}
\end{align}
where in the last inequality we used that $b\leq\varepsilon$ (Property 1) and that the quantity prob(acc)$=(1-b)l+b$ is the probability of accepting (Equation \ref{eq:stateend2}). Hoeffding's Inequality ensures that $|$prob(acc)$-N_{\textup{acc}}/d|\leq\theta$ with confidence $1-2$exp$(-2d\theta^2)$ and this yields Property 2. 

Bounding the variation distance as in Equation \ref{eq:eq1} requires knowledge of the two numbers $\varepsilon$ and $N_{\textup{acc}}$, the former obtained theoretically from the protocol and the latter experimentally from the device being tested. $\varepsilon$ is a property of the protocol, of its input and of the noise model and can be computed without running the protocol. However, different devices running the same target circuit will suffer different noise levels and this is captured by $N_{\textup{acc}}$, which depends on the experimental device being tested. It is important to note that the bound on the variation distance is valid only for the outputs of the $N_{\textup{acc}}$ protocol runs ending with flag bit in the state $|\textup{acc}\rangle$. If a protocol run ends with flag bit in the state $|\textup{rej}\rangle$, Property 1 implies no bound on the variation distance and all rejected outputs must be discarded.
 
We can now present our accreditation protocol (a formal description can be found in Box 1 in the Methods).\\

\noindent{\textbf{3. Our accreditation protocol:}}  Our accreditation protocol takes as input a classical description of the target circuit and the number $v$ of trap circuits. The target circuit (Figure \ref{fig:circuit}) must start with qubits in the state $\ket{+}$, contain only single-qubit gates and $cZ$ gates and end with a round of measurements in the Pauli-$X$ basis\footnote{This does not result in any loss of generality: every experimental architecture has its native input states, entangling gates and measurement basis, but these can always be mapped to $\ket{+}$ states, $cZ$ gates and Pauli-$X$ measurements.}. Moreover, it must be decomposed as a sequence of bands, each one containing one round of single-qubit gates and one round of $cZ$ gates. We will indicate the number of qubits with $n$ and the number of bands with $m$.

In our accreditation protocol $v+1$ circuits are implemented, one (chosen at random) being the target and the remaining $v$ being the traps. The trap circuits are obtained by replacing the single-qubit gates in the target circuit with other single-qubit gates, but input state, measurements and $cZ$ gates are the same as in the target (Figure \ref{subfig:trapa}; all single-qubit gates acting on the same qubit in the same band must be recompiled into one gate). These single-qubit gates are chosen as follows (Routine 2, Box 3 in the Methods):\\

\noindent For each band $j\in\{1,\dots,m-1\}$ and for each qubit $i\in\{1,\ldots,n\}$:

\noindent$\bullet$ If qubit $i$ is connected to another qubit $i'$ by a $cZ$ gate, a gate is chosen at random from the set $\{H_i\otimes S_{i'},S_i\otimes H_{i'}\}$ and is implemented on qubits $i$ and $i'$ in band $j$. This gate is then undone in band $j+1$. 

\noindent$\bullet$ Otherwise, if qubit $i$ is not connected to any other qubit by a $cZ$ gate, a gate is chosen at random from the set $\{H_i,S_i\}$ and is implemented on qubit $i$ in band $j$. This gate is then undone in band $j+1$.\\

\noindent Moreover, depending on the random bit $t\in\{0,1\}$, the traps may begin and end with a round of Hadamard gates. Since $(S\otimes H)cZ(S^\dagger\otimes H)=cX$, the trap circuits are a sequence of (randomly oriented) $cX$ gates acting on $\ket{+}^{\otimes n}$ (if $t=0$) or $\ket{0}^{\otimes n}$ (if $t=1$)|Figure \ref{subfig:trapb}. In the absence of noise, they always output $\overline{s}=\overline{0}$.

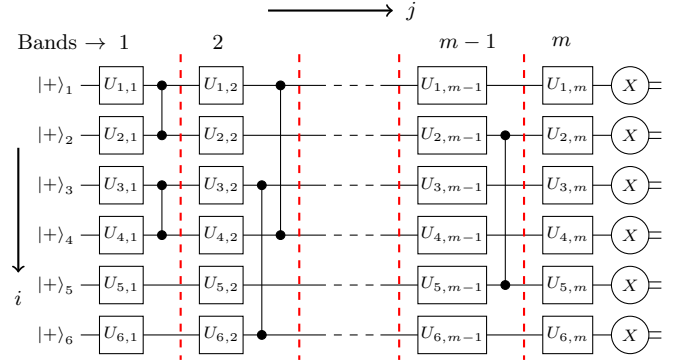
\begin{figure}
\centering
\begin{tikzpicture}[scale=0.83, every node/.style={scale=0.9}]

\draw [thick,->] (-1,3.0) -- (-1,1.0);
\draw [thick,->] (1+2,5.2) -- (3+2,5.2);
\node at (-1,0.6) {$i$};
\node at (5.3,5.2) {$j$};
\node at (-0.3,4.7) {Bands $\rightarrow$};
\node at (0.7,4.7) {1};
\node at (2.2,4.7) {2};
\node at (6.2,4.7) {$m-1$};
\node at (7.7,4.7) {$m$};

\foreach \x in {1,...,6}
\node at (-0.4,4.8-\x*0.8) {\scriptsize $\ket{+}_{\x}$};

\foreach \x in {0,...,5}
\draw (-0.0,4.0-\x*0.8) -- (3.8,4.0-\x*0.8);

\foreach \x in {0,...,5}
\draw [dashed] (3.8,4.0-\x*0.8) -- (4.8,4.0-\x*0.8);

\foreach \x in {0,...,5}
\draw (4.8,4.0-\x*0.8) -- (9.0,4.0-\x*0.8);

\foreach \x in {0,...,5}
\draw [fill=white] (0.3,4.0-\x*0.8-0.3) rectangle (1.0,4.0-\x*0.8+0.3);

\foreach \x in {1,...,6}
\node at (0.65,4.8-\x*0.8) {\scriptsize $U_{\x,1}$};

\foreach \x in {1,...,4}
\draw [fill=black] (1.3,4.8-\x*0.8) circle [radius=0.07cm];

\draw (1.3,2.4) -- (1.3,1.6);
\draw (1.3,3.2) -- (1.3,4.0);

\draw [thick,red,dashed] (1.6,4.5) -- (1.6,-0.4);

\draw [fill=white] (0.3+1*1.6,4.0-0*0.8-0.3) rectangle (1.0+1*1.6,4.0-0*0.8+0.3);
\draw [fill=white] (0.3+1*1.6,4.0-1*0.8-0.3) rectangle (1.0+1*1.6,4.0-1*0.8+0.3);
\draw [fill=white] (0.3+1*1.6,4.0-2*0.8-0.3) rectangle (1.0+1*1.6,4.0-2*0.8+0.3);
\draw [fill=white] (0.3+1*1.6,4.0-3*0.8-0.3) rectangle (1.0+1*1.6,4.0-3*0.8+0.3);
\draw [fill=white] (0.3+1*1.6,4.0-4*0.8-0.3) rectangle (1.0+1*1.6,4.0-4*0.8+0.3);
\draw [fill=white] (0.3+1*1.6,4.0-5*0.8-0.3) rectangle (1.0+1*1.6,4.0-5*0.8+0.3);

\node at (0.65+1*1.6,4.8-1*0.8) {\scriptsize $U_{1,2}$};
\node at (0.65+1*1.6,4.8-2*0.8) {\scriptsize $U_{2,2}$};
\node at (0.65+1*1.6,4.8-3*0.8) {\scriptsize $U_{3,2}$};
\node at (0.65+1*1.6,4.8-4*0.8) {\scriptsize $U_{4,2}$};
\node at (0.65+1*1.6,4.8-5*0.8) {\scriptsize $U_{5,2}$};
\node at (0.65+1*1.6,4.8-6*0.8) {\scriptsize $U_{6,2}$};

\draw [fill=black] (1.6+1*1.6,4.8-1*0.8) circle [radius=0.07cm];
\draw [fill=black] (1.3+1*1.6,4.8-3*0.8) circle [radius=0.07cm];
\draw [fill=black] (1.6+1*1.6,4.8-4*0.8) circle [radius=0.07cm];
\draw [fill=black] (1.3+1*1.6,4.8-6*0.8) circle [radius=0.07cm];

\draw (1.3+1*1.6,0.0) -- (1.3+1*1.6,2.4);
\draw (1.6+1*1.6,4.0) -- (1.6+1*1.6,1.6);

\draw [thick,red,dashed] (1.9+1*1.6,4.5) -- (1.9+1*1.6,-0.4);

\draw [thick,red,dashed] (5.1,4.5) -- (5.1,-0.4);

\draw [fill=white] (5.4,4.0-1*0.8-0.3) rectangle (6.5,4.0-1*0.8+0.3);
\draw [fill=white] (5.4,4.0-2*0.8-0.3) rectangle (6.5,4.0-2*0.8+0.3);
\draw [fill=white] (5.4,4.0-3*0.8-0.3) rectangle (6.5,4.0-3*0.8+0.3);
\draw [fill=white] (5.4,4.0-4*0.8-0.3) rectangle (6.5,4.0-4*0.8+0.3);
\draw [fill=white] (5.4,4.0-5*0.8-0.3) rectangle (6.5,4.0-5*0.8+0.3);
\draw [fill=white] (5.4,4.0-0*0.8-0.3) rectangle (6.5,4.0-0*0.8+0.3);

\node at (5.4+0.55,4.8-1*0.8) {\scriptsize $U_{1,m-1}$};
\node at (5.4+0.55,4.8-2*0.8) {\scriptsize $U_{2,m-1}$};
\node at (5.4+0.55,4.8-3*0.8) {\scriptsize $U_{3,m-1}$};
\node at (5.4+0.55,4.8-4*0.8) {\scriptsize $U_{4,m-1}$};
\node at (5.4+0.55,4.8-5*0.8) {\scriptsize $U_{5,m-1}$};
\node at (5.4+0.55,4.8-6*0.8) {\scriptsize $U_{6,m-1}$};

\draw [fill=black] (6.8,4.8-2*0.8) circle [radius=0.07cm];
\draw [fill=black] (6.8,4.8-5*0.8) circle [radius=0.07cm];

\draw (6.8,0.8) -- (6.8,3.2);

\draw [thick,red,dashed] (7.1,4.5) -- (7.1,-0.4);

\foreach \x in {0,...,5}
\draw [fill=white] (7.4,4.0-\x*0.8-0.3) rectangle (8.2,4.0-\x*0.8+0.3);

\foreach \x in {1,...,6}
\node at (7.4+0.4,4.8-\x*0.8) {\scriptsize $U_{\x,m}$};

\foreach \x in {1,...,6}
\draw (9.1,4.8-\x*0.8-0.05) -- (9.3,4.8-\x*0.8-0.05);
\foreach \x in {1,...,6}
\draw (9.1,4.8-\x*0.8+0.05) -- (9.3,4.8-\x*0.8+0.05);

\foreach \x in {0,...,5}
\draw [fill=white] (8.8,\x*0.8) circle [radius=0.3cm];
\foreach \x in {1,...,6}
\node at (8.8,4.8-\x*0.8) {\scriptsize $X$};

\end{tikzpicture}
\caption{\small A six-qubit example of target circuit.}
\label{fig:circuit}
\end{figure}

\noindent 

Our protocol requires appending a Quantum One-Time Pad (QOTP) to all single-qubit gates in all circuits (target and traps). This is described in Routine 1 in the Methods and is done as follows:\\

\noindent$\bullet$ For all bands $j=1,\ldots,m$ and qubits $i=1,\ldots,n$, a random Pauli gate is appended \textit{after} each gate $U_{i,j}$ (Figure \noindent\ref{fig:correction1}). This yields
\begin{equation}
U_{i,j}'=X_i^{\alpha'_{i,j}}Z_i^{\alpha_{i,j}}U_{i,j}\textrm{ ,}
\end{equation} 
\noindent where $\alpha_{i,j},\alpha'_{i,j}\in\{0,1\}$ are random bits.\\

\noindent$\bullet$ For all bands $j=2,\ldots,m$ and qubits $i=1,\ldots,n$, another Pauli gate is appended \textit{before} each single-qubit gate. This Pauli gate is chosen so that it undoes the QOTP coming from the previous band (Figure \ref{fig:correction2}). Choosing this Pauli gate requires using the identities
\ben
(X_1\otimes I_2) cZ&= cZ(X_1\otimes Z_2),\label{eq:czconjj1}
\\
(Z_1\otimes I_2) cZ&= cZ(Z_1\otimes I_2)\label{eq:czconjj2}
.
\een
This yields
\begin{equation}
\label{eq:correction}
{U_{i,j}''=X_i^{\alpha'_{i,j}}Z_i^{\alpha_{i,j}}U_{i,j}P_i\big(\{\alpha_{i,j-1}\}_i,\{\alpha_{i,j-1}'\}_i\big)\textup{ ,}}
\end{equation}
\noindent {where $P_i\big(\{\alpha_{i,j-1}\}_i,\{\alpha_{i,j-1}'\}_i\big)$} is a Pauli gate that depends on the QOTP in the previous band.\\

\noindent$\bullet$ A random Pauli-$X$ gate is appended \textit{before} all the gates $U_{i,1}'$ in the first band. This yields
\ben
U_{i,1}''=U_{i,1}'X_i^{\gamma_{i}}
\een
with $\gamma_i\in\{0,1\}$ chosen at random.\\

\noindent Overall, replacing each gate $U_{i,j}$ with $U_{i,j}''$ yields a new circuit that is equivalent to the the original one, apart from the un-recovered QOTP $\otimes_{i=1}^nX^{\alpha'_{i,m}}Z^{\alpha_{i,m}}$ in the last

\newpage
\onecolumngrid

\begin{figure}[H]
\centering
\begin{subfigure}{0.55\textwidth}
\begin{tikzpicture}[scale=0.75, every node/.style={scale=0.95}]

\node at (-0.3,4.7) {Bands $\rightarrow$};
\node at (1.2,4.7) {1};
\node at (3.7,4.7) {2};
\node at (7.6,4.7) {$m-1$};
\node at (9,4.7) {$m$};
\foreach \x in {1,...,6}
\node at (-0.3,4.8-\x*0.8) {\scriptsize $\ket{+}_{\x}$};

\foreach \x in {0,...,5}
\draw (0.1,4.0-\x*0.8) -- (5.5,4.0-\x*0.8);

\foreach \x in {0,...,5}
\draw [dashed] (5.5,4.0-\x*0.8) -- (6.5,4.0-\x*0.8);

\foreach \x in {0,...,5}
\draw (11,4.0-\x*0.8) -- (6.5,4.0-\x*0.8);

\foreach \x in {0,...,5}
\draw [fill=white] (0.4,4.0-\x*0.8-0.3) rectangle (1.1,4.0-\x*0.8+0.3);

\foreach \x in {1,...,6}
\node at (0.75,4.8-\x*0.8) {\scriptsize $H^t$};

\foreach \x in {0,...,5}
\draw [fill=white] (1.3,4.0-\x*0.8-0.3) rectangle (2.0,4.0-\x*0.8+0.3);

\node at (1.65,4.8-1*0.8) {\scriptsize $S$};
\node at (1.65,4.8-2*0.8) {\scriptsize $H$};
\node at (1.65,4.8-3*0.8) {\scriptsize $H$};
\node at (1.65,4.8-4*0.8) {\scriptsize $S$};
\node at (1.65,4.8-5*0.8) {\scriptsize $S$};
\node at (1.65,4.8-6*0.8) {\scriptsize $H$};

\foreach \x in {1,...,4}
\draw [fill=black] (2.3,4.8-\x*0.8) circle [radius=0.07cm];

\draw (2.3,2.4) -- (2.3,1.6);
\draw (2.3,3.2) -- (2.3,4.0);

\draw [thick,red,dashed] (2.6,4.5) -- (2.6,-0.4);

\foreach \x in {0,...,5}
\draw [fill=white] (2.8,4.0-\x*0.8-0.3) rectangle (3.6,4.0-\x*0.8+0.3);

\node at (3.2,4.8-1*0.8) {\scriptsize $S^\dagger$};
\node at (3.2,4.8-2*0.8) {\scriptsize $H$};
\node at (3.2,4.8-3*0.8) {\scriptsize $H$};
\node at (3.2,4.8-4*0.8) {\scriptsize $S^\dagger$};
\node at (3.2,4.8-5*0.8) {\scriptsize $S^\dagger$};
\node at (3.2,4.8-6*0.8) {\scriptsize $H$};

\foreach \x in {0,...,5}
\draw [fill=white] (3.8,4.0-\x*0.8-0.3) rectangle (4.6,4.0-\x*0.8+0.3);

\node at (4.2,4.8-1*0.8) {\scriptsize $H$};
\node at (4.2,4.8-2*0.8) {\scriptsize $S$};
\node at (4.2,4.8-3*0.8) {\scriptsize $S$};
\node at (4.2,4.8-4*0.8) {\scriptsize $S$};
\node at (4.2,4.8-5*0.8) {\scriptsize $S$};
\node at (4.2,4.8-6*0.8) {\scriptsize $H$};

\draw [fill=black] (5.1,4.8-1*0.8) circle [radius=0.07cm];
\draw [fill=black] (4.9,4.8-3*0.8) circle [radius=0.07cm];
\draw [fill=black] (5.1,4.8-4*0.8) circle [radius=0.07cm];
\draw [fill=black] (4.9,4.8-6*0.8) circle [radius=0.07cm];

\draw (4.9,0.0) -- (4.9,2.4);
\draw (5.1,4.0) -- (5.1,1.6);

\draw [thick,red,dashed] (5.4,4.5) -- (5.4,-0.4);

\draw [thick,red,dashed] (6.7,4.5) -- (6.7,-0.4);

\foreach \x in {0,...,5}
\draw [fill=white] (6.9,4.0-\x*0.8-0.3) rectangle (7.7,4.0-\x*0.8+0.3);

\node at (7.3,4.8-1*0.8) {\scriptsize $H$};
\node at (7.3,4.8-2*0.8) {\scriptsize $S$};
\node at (7.3,4.8-3*0.8) {\scriptsize $S$};
\node at (7.3,4.8-4*0.8) {\scriptsize $S$};
\node at (7.3,4.8-5*0.8) {\scriptsize $H$};
\node at (7.3,4.8-6*0.8) {\scriptsize $H$};

\draw [fill=black] (8,4.8-2*0.8) circle [radius=0.07cm];
\draw [fill=black] (8,4.8-5*0.8) circle [radius=0.07cm];

\draw (8,0.8) -- (8,3.2);

\draw [thick,red,dashed] (8.3,4.5) -- (8.3,-0.4);

\foreach \x in {0,...,5}
\draw [fill=white] (8.6,4.0-\x*0.8-0.3) rectangle (9.4,4.0-\x*0.8+0.3);

\node at (9,4.8-1*0.8) {\scriptsize $H$};
\node at (9,4.8-2*0.8) {\scriptsize $S$};
\node at (9,4.8-3*0.8) {\scriptsize $S$};
\node at (9,4.8-4*0.8) {\scriptsize $S$};
\node at (9,4.8-5*0.8) {\scriptsize $H$};
\node at (9,4.8-6*0.8) {\scriptsize $H$};

\foreach \x in {0,...,5}
\draw [fill=white] (9.6,4.0-\x*0.8-0.3) rectangle (10.4,4.0-\x*0.8+0.3);

\foreach \x in {0,...,5}
\node at (10,4.0-\x*0.8) {\scriptsize $H^t$};

\foreach \x in {1,...,6}
\draw (9.1+2.2,4.8-\x*0.8-0.05) -- (9.3+2.2,4.8-\x*0.8-0.05);
\foreach \x in {1,...,6}
\draw (9.1+2.2,4.8-\x*0.8+0.05) -- (9.3+2.2,4.8-\x*0.8+0.05);

\foreach \x in {0,...,5}
\draw [fill=white] (11,\x*0.8) circle [radius=0.3cm];
\foreach \x in {1,...,6}
\node at (9.8+1.2,4.8-\x*0.8) {\scriptsize $X$};

\end{tikzpicture}
\caption{}
\label{subfig:trapa}
\end{subfigure}
~
\begin{subfigure}{0.4\textwidth}
\begin{tikzpicture}[scale=0.75, every node/.style={scale=0.95}]
\node at (-0.3,4.7) {Bands $\rightarrow$};
\node at (0.9,4.7) {1};
\node at (2.5,4.7) {2};
\node at (5.1,4.7) {$m-1$};
\node at (6.4,4.7) {$m$};
\foreach \x in {1,...,6}
\node at (-0.4,4.8-\x*0.8) {\scriptsize $\ket{+}_{\x}$};

\foreach \x in {0,...,5}
\draw (-0.0,4.0-\x*0.8) -- (3.4,4.0-\x*0.8);

\foreach \x in {0,...,5}
\draw [dashed] (3.4,4.0-\x*0.8) -- (4.4,4.0-\x*0.8);

\foreach \x in {0,...,5}
\draw (4.4,4.0-\x*0.8) -- (7.3,4.0-\x*0.8);

\foreach \x in {0,...,5}
\draw [fill=white] (0.2,4.0-\x*0.8-0.3) rectangle (1.0,4.0-\x*0.8+0.3);

\foreach \x in {1,...,6}
\node at (0.6,4.8-\x*0.8) {\scriptsize $H^t$};

\draw [fill=black] (1.5,4.8-1*2*0.8+0.8) circle [radius=0.07cm];
\draw [fill=black] (1.5,4.8-2*2*0.8) circle [radius=0.07cm];
\draw [] (1.5,4.8-1*2*0.8) circle [radius=0.15cm];
\draw [] (1.5,4.8-2*2*0.8+0.8) circle [radius=0.15cm];

\draw (1.5,2.55) -- (1.5,1.6);
\draw (1.5,3.05) -- (1.5,4.);

\draw [thick,red,dashed] (1.9,4.5) -- (1.9,-0.4);

\draw [] (2.7,4.8-1*0.8) circle [radius=0.15cm];
\draw [fill=black] (2.4,4.8-3*0.8) circle [radius=0.07cm];
\draw [fill=black] (1.6+1*1.6-0.5,4.8-4*0.8) circle [radius=0.07cm];
\draw [] (1.3+1*1.6-0.5,4.8-6*0.8) circle [radius=0.15cm];

\draw (1.3+1*1.6-0.5,-0.15) -- (1.3+1*1.6-0.5,2.4);
\draw (1.6+1*1.6-0.5,4.15) -- (1.6+1*1.6-0.5,1.6);

\draw [thick,red,dashed] (3.2,4.5) -- (3.2,-0.4);

\draw [thick,red,dashed] (4.5,4.5) -- (4.5,-0.4);

\draw [fill=black] (5.1,4.8-2*0.8) circle [radius=0.07cm];
\draw [] (5.1,4.8-5*0.8) circle [radius=0.15cm];

\draw (5.1,0.8-0.15) -- (5.1,3.2);

\draw [thick,red,dashed] (5.7,4.5) -- (5.7,-0.4);

\foreach \x in {0,...,5}
\draw [fill=white] (6.0,4.0-\x*0.8-0.3) rectangle (6.8,4.0-\x*0.8+0.3);

\foreach \x in {1,...,6}
\node at (6.0+0.4,4.8-\x*0.8) {\scriptsize $H^t$};

\foreach \x in {1,...,6}
\draw (7.5,4.8-\x*0.8-0.05) -- (8.,4.8-\x*0.8-0.05);
\foreach \x in {1,...,6}
\draw (7.5,4.8-\x*0.8+0.05) -- (8.,4.8-\x*0.8+0.05);

\foreach \x in {0,...,5}
\draw [fill=white] (7.5,\x*0.8) circle [radius=0.3cm];
\foreach \x in {1,...,6}
\node at (7.5,4.8-\x*0.8) {\scriptsize $X$};

\end{tikzpicture}
\caption{}
\label{subfig:trapb}
\end{subfigure}
\caption{\small \textbf{(a)} Example of trap circuit for the target circuit in Figure \ref{fig:circuit} and \textbf{(b)} overall computation implemented through this trap circuit. All the single-qubit gates acting on the same qubit in the same band must be recompiled into one gate|for instance, in Figure \ref{subfig:trapa}, the $H^t$-gate and subsequent $S$-gate acting on qubit 1 in band 1 must be implemented as one gate $SH^t$.}
\label{fig:trap}
\end{figure}
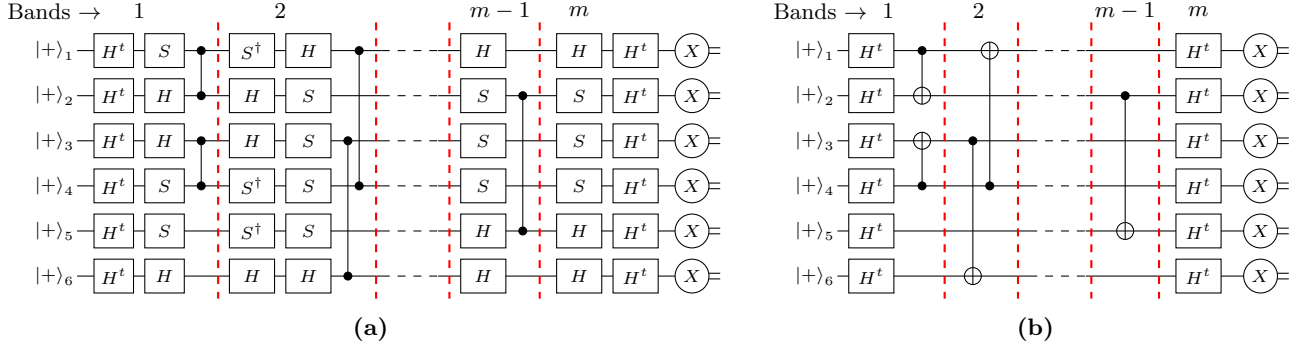

\twocolumngrid

\begin{figure}[H]
\begin{subfigure}{0.5\textwidth}
\centering
\begin{tikzpicture}[scale=1., every node/.style={scale=0.9}]


\draw [thick] (5.4,1.5-0.2) -- (5.4,1.5) -- (3.2,1.5) -- (3.2,1.5-0.2);
\node at (5.2-1.3+0.4,1.8) {\footnotesize $U_{i,j+1}^{\prime}$};
\draw [thick] (5.4,-0.3) -- (5.4,-0.5) -- (3.2,-0.5) -- (3.2,-0.3);
\node at (5.2-1.3+0.4,-0.9) {\footnotesize $U_{i+1,j+1}^{\prime}$};

\draw [thick] (0.05,1.5-0.2) -- (0.05,1.5) -- (2.1,1.5) -- (2.1,1.5-0.2);
\node at (5.2-1.3-2.6,1.8) {\footnotesize $U_{i+1,j}^{\prime}$};
\draw [thick] (0.05,-0.3) -- (0.05,-0.5) -- (2.1,-0.5) -- (2.1,-0.3);
\node at (5.2-1.3-2.6,-0.9) {\footnotesize $U_{i+1,j}^{\prime}$};
\draw [thick,red,dashed] (5.8-0.15,1.3) -- (5.8-0.15,-0.3);
\draw [thick,red,dashed] (0.0,1.3) -- (0.0,-0.3);

\draw [dashed] (-0.3,0.0) -- (0.0,0.0); 
\draw [dashed] (-0.3,1.0) -- (0.0,1.0); 
\draw (0,1.0) -- (5.2+1.2+0.2,1.0);
\draw (0,0) -- (5.2+1.2+0.2,0);
\draw [dashed] (5.2+1.2+0.6,0.0) -- (5.2+1.2+0.3,0.0); 
\draw [dashed] (5.2+1.2+0.6,1.0) -- (5.2+1.2+0.3,1.0); 

\draw (2.3,0.0) -- (2.3,1.0);
\draw [fill] (2.3,0.0) circle [radius=0.07cm];
\draw [fill] (2.3,1.0) circle [radius=0.07cm];

\draw (5.5,0.0) -- (5.5,1.0);
\draw [fill] (5.5,0.0) circle [radius=0.07cm];
\draw [fill] (5.5,1.0) circle [radius=0.07cm];

\draw [fill=white] (0.15,1-0.3) rectangle (1.25,1+0.3);
\node at (0.7,1.0) {\footnotesize $U_{i,j}$};
\draw [fill=red!20] (1.4,1-0.3) rectangle (2.0,1+0.3);
\node at (1.7,1.0) {\footnotesize $X$};
\draw [fill=white] (3.35,1-0.3) rectangle (4.45,1+0.3);
\node at (3.9,1.0) {\footnotesize $U_{i,j+1}$};
\draw [fill=red!20] (4.6,1-0.3) rectangle (5.2,1+0.3);
\node at (4.9,1.0) {\footnotesize $Z$};

\draw [fill=white] (0.15,-0.3) rectangle (1.25,0.3);
\node at (0.7,0.0) {\footnotesize $U_{i+1,j}$};
\draw [fill=red!20] (1.4,-0.3) rectangle (2.0,0.3);
\node at (1.7,0.0) {\footnotesize $XZ$};
\draw [fill=white] (3.35,-0.3) rectangle (4.45,0.3);
\node at (3.9,0.0) {\footnotesize $U_{i+1,j+1}$};
\draw [fill=red!20] (4.6,-0.3) rectangle (5.2,0.3);
\node at (4.9,0.0) {\footnotesize $X$};

\end{tikzpicture}
\caption{\small }
\label{fig:correction1}
\end{subfigure}
~
\begin{subfigure}{0.5\textwidth}
\centering
\vspace{0.3cm}
\begin{tikzpicture}[scale=1, every node/.style={scale=0.9}]


\draw [thick] (5.4,1.5-0.2) -- (5.4,1.5) -- (2.4,1.5) -- (2.4,1.5-0.2);
\node at (5.2-1.3,1.8) {\footnotesize $U_{i,j+1}^{\prime\prime}$};
\draw [thick] (5.4,-0.3) -- (5.4,-0.5) -- (2.4,-0.5) -- (2.4,-0.3);
\node at (5.2-1.3,-0.9) {\footnotesize $U_{i+1,j+1}^{\prime\prime}$};

\draw [thick] (0.05,1.5-0.2) -- (0.05,1.5) -- (2.1,1.5) -- (2.1,1.5-0.2);
\node at (5.2-1.3-2.8,1.8) {\footnotesize $U_{i,j}^{\prime}$};
\draw [thick] (0.05,-0.3) -- (0.05,-0.5) -- (2.1,-0.5) -- (2.1,-0.3);
\node at (5.2-1.3-2.8,-0.9) {\footnotesize $U_{i+1,j}^{\prime}$};

\draw [thick,red,dashed] (5.8-0.15,1.3) -- (5.8-0.15,-0.3);
\draw [thick,red,dashed] (0.0,1.3) -- (0.0,-0.3);

\draw [dashed] (-0.3,0.0) -- (0.0,0.0); 
\draw [dashed] (-0.3,1.0) -- (0.0,1.0); 
\draw (0,1.0) -- (5.2+1.2+0.2,1.0);
\draw (0,0) -- (5.2+1.2+0.2,0);
\draw [dashed] (5.2+1.2+0.6,0.0) -- (5.2+1.2+0.3,0.0); 
\draw [dashed] (5.2+1.2+0.6,1.0) -- (5.2+1.2+0.3,1.0);

\draw (2.3,0.0) -- (2.3,1.0);
\draw [fill] (2.3,0.0) circle [radius=0.07cm];
\draw [fill] (2.3,1.0) circle [radius=0.07cm];

\draw (5.5,0.0) -- (5.5,1.0);
\draw [fill] (5.5,0.0) circle [radius=0.07cm];
\draw [fill] (5.5,1.0) circle [radius=0.07cm];

\draw [fill=white] (0.15,1-0.3) rectangle (1.25,1+0.3);
\node at (0.7,1.0) {\footnotesize $U_{i,j}$};
\draw [fill=red!20] (1.4,1-0.3) rectangle (2.0,1+0.3);
\node at (1.7,1.0) {\footnotesize $X$};
\draw [fill=green!20] (2.6,1-0.3) rectangle (3.2,1+0.3);
\node at (2.9,1.0) {\footnotesize $XZ$};
\draw [fill=white] (3.35,1-0.3) rectangle (4.45,1+0.3);
\node at (3.9,1.0) {\footnotesize $U_{i,j+1}$};
\draw [fill=red!20] (4.6,1-0.3) rectangle (5.2,1+0.3);
\node at (4.9,1.0) {\footnotesize $Z$};
\draw [fill=green!20] (4.6+1.2,1-0.3) rectangle (5.2+1.2,1.3);
\node at (4.9+1.2,1.0) {\footnotesize $I$};

\draw [fill=white] (0.15,-0.3) rectangle (1.25,0.3);
\node at (0.7,0.0) {\footnotesize $U_{i+1,j}$};
\draw [fill=red!20] (1.4,-0.3) rectangle (2.0,0.3);
\node at (1.7,0.0) {\footnotesize $XZ$};
\draw [fill=green!20] (2.6,-0.3) rectangle (3.2,0.3);
\node at (2.9,0.0) {\footnotesize $X$};
\draw [fill=white] (3.35,-0.3) rectangle (4.45,0.3);
\node at (3.9,0.0) {\footnotesize $U_{i+1,j+1}$};
\draw [fill=red!20] (4.6,-0.3) rectangle (5.2,0.3);
\node at (4.9,0.0) {\footnotesize $X$};
\draw [fill=green!20] (4.6+1.2,-0.3) rectangle (5.2+1.2,0.3);
\node at (4.9+1.2,0.0) {\footnotesize $X$};

\end{tikzpicture}
\caption{\small}
\label{fig:correction2}
\end{subfigure}
\caption{\small Example of Quantum one-time pad. \textbf{(a)} The red Pauli gates apply the QOTP. \textbf{(b)} The green gates undo the QOTP coming from previous bands.}
\label{fig:corrections}
\end{figure}
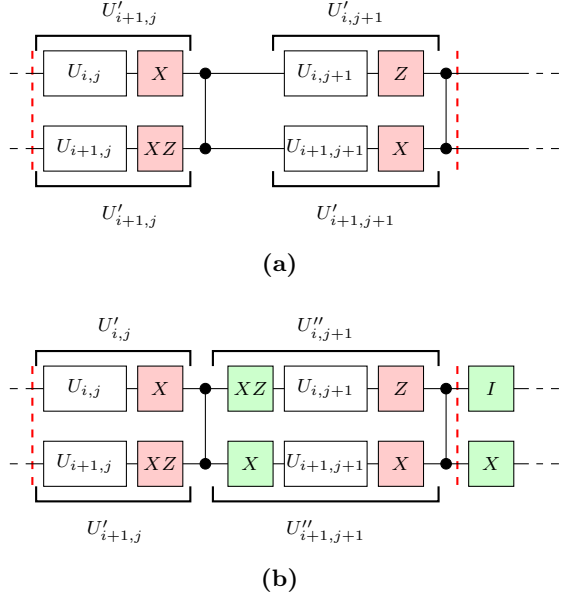

\noindent band. Since all measurements are in the Pauli-$X$ basis, the Pauli-$X$ component of this un-recovered QOTP is irrelevant, while its Pauli-$Z$ component bit-flips some of the outputs. These bit-flips can be undone by replacing each output $s_i$ with $s_i\oplus\alpha_{i,m}$ (a procedure that we call ``classical post-processing of the outputs''). This allows to recover the correct outputs.

After all the circuits have been implemented and the outputs have been post-processed, the flag bit is initialized to {$|\textup{acc}\rangle=|\textup{0}\rangle$}, then it is checked whether all the traps gave the correct output $\overline{s}=\overline{0}$. If they do, the protocol returns the output of the target together with the bit $|\textup{acc}\rangle$, otherwise it returns the output of the target together with the bit $|\textup{rej}\rangle=|1\rangle$. The output of the target is only accepted in the first case, while it is discarded in the second case.

In the absence of noise, our protocol always returns the correct output of the target circuit and always accepts it. Correctness of the target is ensured by the fact that the QOTP has no effect on the computation, as all the extra Pauli gates cancel out with each other or are countered by the classical post-processing of the outputs. Acceptance is ensured by the fact that in the absence of noise all the trap circuits always yield the correct outcome $\overline{s}=\overline{0}$.

We will now consider a noisy implementation of our protocol, explain the role played by the various tools (QOTP, trap circuits etc.) and show that with single-qubit gates suffering bounded noise, our protocol ensures that wrong outputs are rejected with high probability.\\

\noindent\noindent{\textbf{4. The credibility of our protocol: }}As per Equation \ref{eq:stateend}, we model noise as a set of CPTP maps acting on the whole system and on the environment (Figure \ref{fig:noisyprotocol}). For simplicity, let us begin with the assumption that all the rounds of single-qubit gates in our protocol are noiseless, i.e. that {for all circuits $k=1,\ldots,v+1$ and bands $j=1,\ldots,m$, a noisy implementation of the round of single-qubit gates is} (cfr. Figure \ref{fig:noisyprotocol} for notation)
\begin{align}
\label{eq:noiselesssq}
\widetilde{\U}_j^{\prime\prime(k)}=\E_j^{(k)}\big({\U}_j^{\prime\prime(k)}\otimes\I_E\big)\textup{ }\textup{ with }\textup{ }\E_j^{(k)}=\I_{SE}\textup{ ,}
\end{align}
where $\I_{SE}$ is the identity on system and environment. Under this assumption, a first simplification to the noise of type {N1} comes from the QOTP, a tool used in many works in verification \cite{GKK17} and benchmarking protocols \cite{WE16,Erhard&al19} that also plays a crucial role in our protocol. If single-qubit gates are noiseless, the QOTP allows to randomize all noise processes, even those non-local in space and time, to classically correlated Pauli errors (see Lemma \ref{lem:pauliz} in Appendix \ref{subapp:lem1}). {A similar result was previously proven in Ref. \cite{WE16} for Markovian noise, and here we show that this result holds also if the noise creates correlations in time.}

Having reduced arbitrary non-local noise to Pauli errors via the QOTP, we show (see Lemma \ref{lem:ct} in Appendix \ref{subapp:lem2}) that our trap circuits detect all Pauli errors with non-zero probability. The reasoning is as follows: Since the trap circuits contain only Clifford gates, the noise acting

\newpage
\onecolumngrid

\begin{figure}[H]
\centering
\begin{tikzpicture}[scale=0.72, every node/.style={scale=0.93}]

\node at (-0.6,3) {$\otimes_{i=1}^n|+\rangle_i\langle+|$};
\node at (14.-0.1,1.5) {$\otimes_{i=1}^n|+\rangle_i\langle+|$};

\draw (0.7,3-0.15) -- (1.0,3+0.15);
\draw (0.7+14.5,1.5-0.15) -- (1.0+14.5,1.5+0.15);

\node at (0.0,0.0) {$\rho_E$};

\draw (0.6,0*1.5) -- (7.5,0*1.5);
\draw [dashed] (7.5,0*1.5) -- (8.5,0*1.5);
\draw (8.4,0*1.5) -- (6.9+0.4+14.5,0*1.5);
\draw [dashed] (6.9+0.4+14.5,0*1.5) -- (14.5+8.1,0*1.5);

\draw (0.6,2*1.5) -- (7.5,2*1.5);
\draw [dashed] (7.5,2*1.5) -- (8.5,2*1.5);
\draw (8.4,2*1.5) -- (13,2*1.5);
\draw (13,2*1.5+0.05) -- (6.9+0.4+14.5,2*1.5+0.05);
\draw (13,2*1.5-0.05) -- (6.9+0.4+14.5,2*1.5-0.05);
\draw [dashed] (6.9+0.4+14.5,2*1.5+0.05) -- (14.5+8.1,2*1.5+0.05);
\draw [dashed] (6.9+0.4+14.5,2*1.5-0.05) -- (14.5+8.1,2*1.5-0.05);

\draw (14.5+0.6,1.5) -- (6.9+0.4+14.5,1.5);
\draw [dashed] (6.9+0.4+14.5,1.5) -- (14.5+8.1,1.5);

\draw [fill=white] (1.2,3.5) rectangle (2.2,-0.5);
\node at (1.7,1.5) {$\R^{(1)}$};

\draw [red,dashed,thick] (2.4,-0.5) -- (2.4,3.8);

\draw [fill=white] (2.55,2.5) rectangle (3.6,3.5);
\node at (3.1,3) {$\U_1^{\prime\prime(1)}$};
\draw [fill=white] (3.8,3.5) rectangle (4.8,-0.5);
\node at (4.3,1.5) {$\E_1^{(1)}$};

\draw [red,dashed,thick] (4.95,-0.5) -- (4.95,3.8);

\draw [fill=white] (4.7+0.4,3.5) rectangle (5.7+0.4,2.5);
\node at (5.2+0.4,3) {$\C\Z_{1}$};
\draw [fill=white] (5.9+0.4,3.5) rectangle (6.9+0.4,-0.5);
\node at (6.4+0.4,1.5) {$\F_1^{(1)}$};

\draw [red,dashed,thick] (6.9+0.55,-0.5) -- (6.9+0.55,3.8);

\draw [fill=white] (8.75,3.5) rectangle (9.8,2.5);
\node at (9.3,3) {$\U_m^{\prime\prime(1)}$};
\draw [fill=white] (10,3.5) rectangle (11,-0.5);
\node at (10.5,1.5) {$\E_m^{(1)}$};

\draw [red,dashed,thick] (11.15,-0.5) -- (11.15,3.8);

\draw [fill=white] (11.4,3.5) rectangle (12.4,-0.5);
\node at (11.9,1.5) {$\M^{(1)}$};

\draw [fill=white] (13,3) circle [radius=0.3cm];
\node at (13,3) {\scriptsize $X$};

\draw [fill=white] (1.2+14.5,2) rectangle (2.2+14.5,-0.5);
\node at (1.7+14.5,0.75) {$\R^{(2)}$};

\draw [red,dashed,thick] (2.2+14.65,-0.5) -- (2.2+14.65,2.3);

\draw [fill=white] (2.55+14.5,2) rectangle (3.6+14.5,1);
\node at (3.1+14.5,1.5) {$\U_1^{\prime\prime(2)}$};
\draw [fill=white] (3.8+14.5,2) rectangle (4.8+14.5,-0.5);
\node at (4.3+14.5,0.75) {$\E_1^{(2)}$};

\draw [red,dashed,thick] (4.8+14.65,-0.5) -- (4.8+14.65,2.3);

\draw [fill=white] (4.7+0.4+14.5,2) rectangle (5.7+0.4+14.5,1);
\node at (5.2+0.4+14.5,1.5) {$\C\Z_{1}$};
\draw [fill=white] (5.9+0.4+14.5,2) rectangle (6.9+0.4+14.5,-0.5);
\node at (6.4+0.4+14.5,0.75) {$\F_1^{(2)}$};

\draw [red,dashed,thick] (6.9+0.4+14.65,-0.5) -- (6.9+0.4+14.65,2.3);

\end{tikzpicture}
\caption{\small Schematic illustration of a noisy implementation of our protocol where all boxes represent CPTP maps. $\U_j^{\prime\prime(k)}$ implements the round of single-qubit gates in band $j$ of circuit $k$, $\C\Z_j$ implements the round of $cZ$ gates in band $j$. In each circuit $k=1,\ldots,v+1$: $\R^{(k)}$ is the noise in state preparation, $\M^{(k)}$ is the noise in measurements, $\E_j^{(k)}$ is the noise in the round of single-qubit gates in band $j$ and $\F_j^{(k)}$ is the noise in the round of $cZ$ gates in band $j$.}
\label{fig:noisyprotocol}
\end{figure}
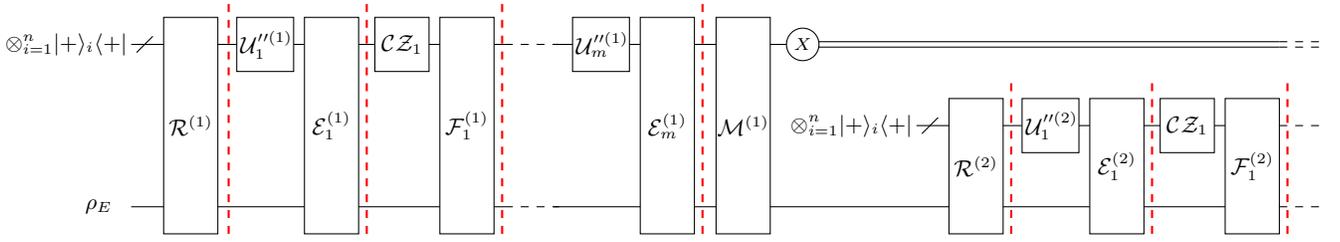

\twocolumngrid

\noindent at any point of a trap circuit can be factored to the end of the circuit. The noisy trap circuit is thus rewritten as the original one (Figure \ref{subfig:trapa}) with a Pauli-$Z$ error $P\in\{I,Z\}^{\otimes n}$ applied before the measurements. If $P\neq I^{\otimes n}$, the trap outputs a wrong output ($\overline{s}\neq\overline{0}$) and the noise is detected. However, if the errors in different parts of the circuit happen to cancel out, then $P=I^{\otimes n}$, the trap outputs $\overline{s}=\overline{0}$ and the noise is not detected. The role of $H$ and $S$-gates in our trap circuits is to ensure that this happens with suitably low probability for \textit{all} types of noise that can possibly affect the trap. These gates map Pauli errors into other Pauli errors as
\begin{align}
 & \textcolor{white}{\rightarrow}& H&\textup{ }&S\cr
X & \rightarrow& Z&&Y\cr
Y & \rightarrow& Y&&X\cr
Z & \rightarrow& X&&Z
\end{align}
where we omit unimportant prefactors (global phases do not affect outputs). Therefore, the random implementation of $H$ and $S$-gates prevents errors in state preparation and two-qubit gates from canceling trivially. Similarly, the rounds of Hadamard gates activated at random at the beginning and at the end of the trap circuits prevent measurement errors from canceling trivially with noise happening before. These arguments are used to prove the claim of Lemma \ref{lem:ct}, that states that our trap circuits can detect all possible Pauli errors with probability larger than 1/4.

The above arguments and Lemmas can be used to prove that our protocol can detect arbitrary noise in state preparation, measurement and two-qubit gates, provided that single-qubit gates are noiseless:
\begin{theorem}~\label{th:verif1}
Suppose that all single-qubit gates in our accreditation protocol are noiseless. For any number $v\geq3$ of trap circuits, our accreditation protocol can accredit the outputs of a noisy quantum computer affected by noise of the form {N1} with
\begin{equation}
\varepsilon=\frac{\kappa}{v+1}\textup{ ,}
\label{eq:eps}
\end{equation}
where $\kappa = 3(3/4)^2 \approx 1.7.$
\end{theorem}
To prove $\varepsilon={\kappa}/({v+1})$ we write the state of the system at the end of a noisy protocol run as in Equation \ref{eq:stateend2}. We do this using Lemmas \ref{lem:pauliz} and \ref{lem:ct}. The proof of Theorem \ref{th:verif1} is in Appendix \ref{app:th1}.

We now relax the assumption of noiseless single-qubit gates and generalize our results to noise of the form {N2}. We assume that all rounds of single-qubit gates suffer bounded noise, i.e. that for all circuits $k=1,\ldots,v+1$ and bands $j=1,\ldots,m$, a noisy implementation of the round of single-qubit gates is (cfr. Figure \ref{fig:noisyprotocol} for notation)
\begin{align}
\label{eq:noisebound}
\widetilde{\U}_j^{\prime\prime(k)}=\E_j^{(k)}\big({\U}_j^{\prime\prime(k)}\otimes\I_E\big)\textup{ }
\end{align}
with $\E_j^{(k)}=(1-r_j^{(k)})\I_{SE}+r_j^{(k)}\E_j^{\prime\textup{ }(k)}$ for some arbitrary CPTP map $\E_j^{\prime\textup{ }(k)}$ acting on both system and environment and for some number $0\leq r_j^{(k)}<1$. We refer to the number $r_{j}^{(k)}$ as ``error rate'' of ${\U}_j^{\prime\prime(k)}$. Since each ${\U}_j^{\prime\prime(k)}$ is chosen at random (depending on whether circuit $k$ is the target or a trap and on the QOTP) and since noise in single-qubit gates is potentially gate-dependent (condition {N2}), let us indicate with $r_{\textup{max, }j}^{(k)}$ the maximum error rate of single-qubit gates in band $j$ of circuit $k$, the maximum being taken over all possible choices of ${\U}_j^{\prime\prime(k)}$. 

We can now state Theorem \ref{th:verif2}:
\begin{theorem}~\label{th:verif2}
Our protocol with $v\geq3$ of trap circuits can accredit the outputs of a noisy quantum computer affected by noise of the form \textup{{N1}} and \textup{{N2}} with
\begin{align}
\label{eq:varepsilonth2}
\varepsilon=g\frac{\kappa}{v+1}+\textup{ }1\textup{ }-g\textrm{ ,}
\end{align}
where $\kappa = 3(3/4)^2 \approx 1.7$ and $g=\prod_{j,k}(1-r_{\textup{max, }j}^{(k)})$.
\end{theorem}
To calculate $\varepsilon$ for the protocol with noisy single-qubit gates we use that $\E_j^{(k)}=(1-r_{\textup{max, }j}^{(k)})\I_{SE}+r_{\textup{max, }j}^{(k)}\Q_j^{(k)}$, where $\Q_j^{(k)}$ is a CPTP map encompassing the system and the environment. We can then rewrite the state of the system at the end of the protocol as
\begin{align}~\label{eq:rhoprime}
{\rho}_{\textup{out}}^\star=&g{\rho}_{\textup{out}}+\big(1-g\big)\textup{}\widetilde{\rho}_{\textup{out}}
\end{align}
where $\rho_{\textup{out}}$ is the state of the system at the end of a  protocol run with noiseless single-qubit gates|which by

\onecolumngrid

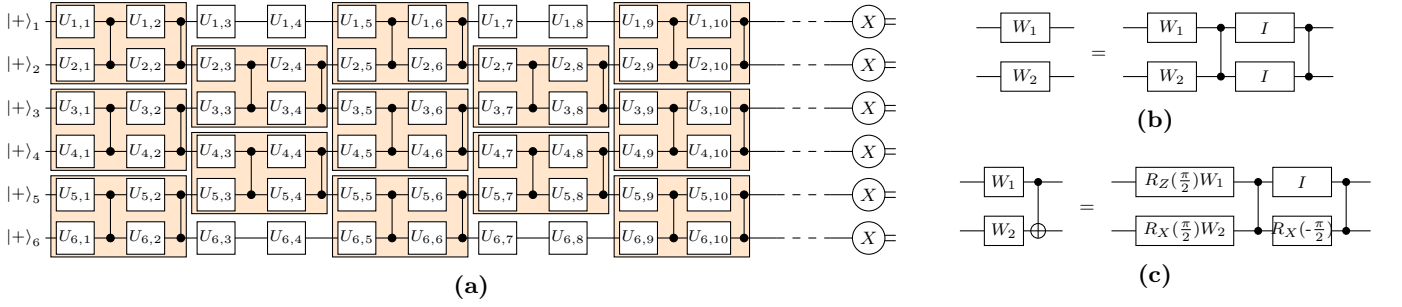
\begin{figure}[H]
\begin{subfigure}{0.7\textwidth}
\begin{tikzpicture}[scale=0.72, every node/.style={scale=0.85}]

\draw [fill=orange!20] (0.2,0*0.8-0.35) rectangle (2.7,1*0.8+0.35);
\draw [fill=orange!20] (0.2,2*0.8-0.35) rectangle (2.7,3*0.8+0.35);
\draw [fill=orange!20] (0.2,4*0.8-0.35) rectangle (2.7,5*0.8+0.35);

\draw [fill=orange!20] (0.2+1*2.6,1*0.8-0.35) rectangle (2.7+1*2.6,2*0.8+0.35);
\draw [fill=orange!20] (0.2+1*2.6,3*0.8-0.35) rectangle (2.7+1*2.6,4*0.8+0.35);

\draw [fill=orange!20] (0.2+2*2.6,0*0.8-0.35) rectangle (2.7+2*2.6,1*0.8+0.35);
\draw [fill=orange!20] (0.2+2*2.6,2*0.8-0.35) rectangle (2.7+2*2.6,3*0.8+0.35);
\draw [fill=orange!20] (0.2+2*2.6,4*0.8-0.35) rectangle (2.7+2*2.6,5*0.8+0.35);

\draw [fill=orange!20] (-0.1+3*2.7,1*0.8-0.35) rectangle (2.7+3*2.7-0.3,2*0.8+0.35);
\draw [fill=orange!20] (-0.1+3*2.7,3*0.8-0.35) rectangle (2.7+3*2.7-0.3,4*0.8+0.35);

\draw [fill=orange!20] (0.2+4*2.6,0*0.8-0.35) rectangle (2.7+4*2.6,1*0.8+0.35);
\draw [fill=orange!20] (0.2+4*2.6,2*0.8-0.35) rectangle (2.7+4*2.6,3*0.8+0.35);
\draw [fill=orange!20] (0.2+4*2.6,4*0.8-0.35) rectangle (2.7+4*2.6,5*0.8+0.35);

%

\foreach \x in {1,...,6}
\node at (-0.3,4.8-\x*0.8) {\scriptsize $\ket{+}_{\x}$};

\foreach \x in {0,...,5}
\draw (0.1,4.0-\x*0.8) -- (11+2*1.3,4.0-\x*0.8);

\foreach \x in {0,...,5}
\draw [dashed] (11+2*1.3,4.0-\x*0.8) -- (12+2*1.3,4.0-\x*0.8);

\foreach \x in {0,...,5}
\draw (12.7+2*1.3,4.0-\x*0.8) -- (12+2*1.3,4.0-\x*0.8);

\foreach \x in {0,...,5}
\draw [fill=white] (0.3,4.0-\x*0.8-0.3) rectangle (1,4.0-\x*0.8+0.3);

\foreach \x in {1,...,6}
\node at (0.65,4.8-\x*0.8) {\scriptsize $U_{\x,1}$};

\foreach \x in {1,...,6}
\draw [fill=black] (1.3,4.8-\x*0.8) circle [radius=0.07cm];

\draw (1.3,0.0) -- (1.3,0.8);
\draw (1.3,1.6) -- (1.3,2.4);
\draw (1.3,3.2) -- (1.3,4.0);

\foreach \x in {0,...,5}
\draw [fill=white] (1.6,4.0-\x*0.8-0.3) rectangle (2.3,4.0-\x*0.8+0.3);

\foreach \x in {1,...,6}
\node at (1.95,4.8-\x*0.8) {\scriptsize $U_{\x,2}$};

\foreach \x in {1,...,6}
\draw [fill=black] (1.3+1.3,4.8-\x*0.8) circle [radius=0.07cm];

\draw (1.3+1.3,0.0) -- (1.3+1.3,0.8);
\draw (1.3+1.3,1.6) -- (1.3+1.3,2.4);
\draw (1.3+1.3,3.2) -- (1.3+1.3,4.0);

\foreach \x in {0,...,5}
\draw [fill=white] (0.3+2*1.3,4.0-\x*0.8-0.3) rectangle (1+2*1.3,4.0-\x*0.8+0.3);

\foreach \x in {1,...,6}
\node at (0.65+2*1.3,4.8-\x*0.8) {\scriptsize $U_{\x,3}$};

\foreach \x in {2,...,5}
\draw [fill=black] (1.3+2*1.3,4.8-\x*0.8) circle [radius=0.07cm];

\draw (1.3+2*1.3,1.6) -- (1.3+2*1.3,0.8);
\draw (1.3+2*1.3,3.2) -- (1.3+2*1.3,2.4);

\foreach \x in {0,...,5}
\draw [fill=white] (0.3+3*1.3,4.0-\x*0.8-0.3) rectangle (1+3*1.3,4.0-\x*0.8+0.3);

\foreach \x in {1,...,6}
\node at (0.65+3*1.3,4.8-\x*0.8) {\scriptsize $U_{\x,4}$};

\foreach \x in {2,...,5}
\draw [fill=black] (1.3+3*1.3,4.8-\x*0.8) circle [radius=0.07cm];

\draw (1.3+3*1.3,1.6) -- (1.3+3*1.3,0.8);
\draw (1.3+3*1.3,3.2) -- (1.3+3*1.3,2.4);

\foreach \x in {0,...,5}
\draw [fill=white] (0.3+4*1.3,4.0-\x*0.8-0.3) rectangle (1+4*1.3,4.0-\x*0.8+0.3);

\foreach \x in {1,...,6}
\node at (0.65+4*1.3,4.8-\x*0.8) {\scriptsize $U_{\x,5}$};

\foreach \x in {1,...,6}
\draw [fill=black] (1.3+4*1.3,4.8-\x*0.8) circle [radius=0.07cm];

\draw (1.3+4*1.3,0.0) -- (1.3+4*1.3,0.8);
\draw (1.3+4*1.3,1.6) -- (1.3+4*1.3,2.4);
\draw (1.3+4*1.3,3.2) -- (1.3+4*1.3,4.0);

\foreach \x in {0,...,5}
\draw [fill=white] (0.3+5*1.3,4.0-\x*0.8-0.3) rectangle (1+5*1.3,4.0-\x*0.8+0.3);

\foreach \x in {1,...,6}
\node at (0.65+5*1.3,4.8-\x*0.8) {\scriptsize $U_{\x,6}$};

\foreach \x in {1,...,6}
\draw [fill=black] (1.3+5*1.3,4.8-\x*0.8) circle [radius=0.07cm];

\draw (1.3+5*1.3,0.0) -- (1.3+5*1.3,0.8);
\draw (1.3+5*1.3,1.6) -- (1.3+5*1.3,2.4);
\draw (1.3+5*1.3,3.2) -- (1.3+5*1.3,4.0);

\foreach \x in {0,...,5}
\draw [fill=white] (0.3+6*1.3,4.0-\x*0.8-0.3) rectangle (1+6*1.3,4.0-\x*0.8+0.3);

\foreach \x in {1,...,6}
\node at (0.65+6*1.3,4.8-\x*0.8) {\scriptsize $U_{\x,7}$};

\foreach \x in {2,...,5}
\draw [fill=black] (1.3+6*1.3,4.8-\x*0.8) circle [radius=0.07cm];

\draw (1.3+6*1.3,1.6) -- (1.3+6*1.3,0.8);
\draw (1.3+6*1.3,3.2) -- (1.3+6*1.3,2.4);

\foreach \x in {0,...,5}
\draw [fill=white] (0.3+7*1.3,4.0-\x*0.8-0.3) rectangle (1+7*1.3,4.0-\x*0.8+0.3);

\foreach \x in {1,...,6}
\node at (0.65+7*1.3,4.8-\x*0.8) {\scriptsize $U_{\x,8}$};

\foreach \x in {2,...,5}
\draw [fill=black] (1.3+7*1.3,4.8-\x*0.8) circle [radius=0.07cm];

\draw (1.3+7*1.3,1.6) -- (1.3+7*1.3,0.8);
\draw (1.3+7*1.3,3.2) -- (1.3+7*1.3,2.4);

\foreach \x in {0,...,5}
\draw [fill=white] (0.3+8*1.3,4.0-\x*0.8-0.3) rectangle (1+8*1.3,4.0-\x*0.8+0.3);

\foreach \x in {1,...,6}
\node at (0.65+8*1.3,4.8-\x*0.8) {\scriptsize $U_{\x,9}$};

\foreach \x in {1,...,6}
\draw [fill=black] (1.3+8*1.3,4.8-\x*0.8) circle [radius=0.07cm];

\draw (1.3+8*1.3,0.0) -- (1.3+8*1.3,0.8);
\draw (1.3+8*1.3,1.6) -- (1.3+8*1.3,2.4);
\draw (1.3+8*1.3,3.2) -- (1.3+8*1.3,4.0);

\foreach \x in {0,...,5}
\draw [fill=white] (0.25+9*1.3,4.0-\x*0.8-0.3) rectangle (1.05+9*1.3,4.0-\x*0.8+0.3);

\foreach \x in {1,...,6}
\node at (0.65+9*1.3,4.8-\x*0.8) {\scriptsize $U_{\x,10}$};

\foreach \x in {1,...,6}
\draw [fill=black] (1.3+9*1.3,4.8-\x*0.8) circle [radius=0.07cm];

\draw (1.3+9*1.3,0.0) -- (1.3+9*1.3,0.8);
\draw (1.3+9*1.3,1.6) -- (1.3+9*1.3,2.4);
\draw (1.3+9*1.3,3.2) -- (1.3+9*1.3,4.0);

\foreach \x in {1,...,6}
\draw (12.7+2*1.3,4.8-\x*0.8-0.05) -- (12.7+0.5+2*1.3,4.8-\x*0.8-0.05);
\foreach \x in {1,...,6}
\draw (12.7+2*1.3,4.8-\x*0.8+0.05) -- (12.7+0.5+2*1.3,4.8-\x*0.8+0.05);

\foreach \x in {1,...,6}
\draw [fill=white] (12.7+2*1.3,4.8-\x*0.8) circle [radius=0.3cm];
\foreach \x in {1,...,6}
\node at (12.7+2*1.3,4.8-\x*0.8) {\footnotesize $X$};

\end{tikzpicture}
\caption{}
\end{subfigure}
~
\begin{subfigure}{0.29\textwidth}
\centering
\begin{tikzpicture}[scale=0.65, every node/.style={scale=0.8}]

\draw (3.5,2.0) -- (7.8,2.0);
\draw (3.5,3.0) -- (7.8,3.0);

\draw [fill=white] (4,2.7) rectangle (5.,3.3);
\draw [fill=white] (4,1.7) rectangle (5.,2.3);
\node at (4.5,3) {\footnotesize $W_1$};
\node at (4.5,2) {\footnotesize $W_2$};

\draw [fill=black] (5.5,2.0) circle [radius=0.07];
\draw [fill=black] (5.5,3.0) circle [radius=0.07];
\draw [fill=black] (7.3,2.0) circle [radius=0.07];
\draw [fill=black] (7.3,3.0) circle [radius=0.07];

\draw (5.5,2.0) -- (5.5,3.0);
\draw (7.3,2.0) -- (7.3,3.0);

\draw [fill=white] (5.8,1.7) rectangle (7.0,2.3);
\node at (6.4,2) {\footnotesize $I$};
\draw [fill=white] (5.8,2.7) rectangle (7.0,3.3);
\node at (6.4,3) {\footnotesize $I$};

\node at (3.0,2.45) {=};

\draw (0.5,2.0) -- (2.5,2.0);
\draw (0.5,3.0) -- (2.5,3.0);

\draw [fill=white] (1.,2.7) rectangle (2.,3.3);
\node at (1.50,3) {\footnotesize $W_1$};
\draw [fill=white] (1.0,1.7) rectangle (2.0,2.3);
\node at (1.5,2) {\footnotesize $W_2$};

\end{tikzpicture}
\caption{}
\vspace{0.3cm}
\label{fig:gadget}
\begin{tikzpicture}[scale=0.65, every node/.style={scale=0.8}]

\draw (2.5,2.0) -- (7.8,2.0);
\draw (2.5,3.0) -- (7.8,3.0);

\draw [fill=white] (3,2.7) rectangle (5.,3.3);
\draw [fill=white] (3,1.7) rectangle (5.,2.3);
\node at (4.,3) {\footnotesize $R_Z(\frac{\pi}{2})W_1$};
\node at (4.,2) {\footnotesize $R_X(\frac{\pi}{2})W_2$};

\draw [fill=black] (5.5,2.0) circle [radius=0.07];
\draw [fill=black] (5.5,3.0) circle [radius=0.07];
\draw [fill=black] (7.3,2.0) circle [radius=0.07];
\draw [fill=black] (7.3,3.0) circle [radius=0.07];

\draw (5.5,2.0) -- (5.5,3.0);
\draw (7.3,2.0) -- (7.3,3.0);

\draw [fill=white] (5.8,1.7) rectangle (7.0,2.3);
\node at (6.4,2) {\footnotesize $R_X(\textrm{-}\frac{\pi}{2})$};
\draw [fill=white] (5.8,2.7) rectangle (7.0,3.3);
\node at (6.4,3) {\footnotesize $I$};

\node at (2.0,2.45) {=};

\draw (0.4-1.0,2.0) -- (2.5-1.0,2.0);
\draw (0.4-1.0,3.0) -- (2.5-1.0,3.0);

\draw [fill=white] (0.9-1.0,2.7) rectangle (1.7-1.0,3.3);
\node at (1.3-1.0,3) {\footnotesize $W_1$};
\draw [fill=white] (0.9-1.0,1.7) rectangle (1.7-1.0,2.3);
\node at (1.3-1.0,2) {\footnotesize $W_2$};

\draw (2.-1.0,1.85) -- (2.-1.0,3.0);
\draw [fill=black] (2.-1.0,3.0) circle [radius=0.07];
\draw [] (2.-1.0,2.0) circle [radius=0.15];

\end{tikzpicture}
\caption{}
\label{fig:ucx}
\end{subfigure}
\caption{\textbf{(a)} Six-qubit example of circuit in normal form. This circuit has the same repetitive structure as the Brickwork States \cite{BFK09}. Recompiling the target circuit into a normal form of this type can always be done using the circuit identities \textbf{(b)} and \textbf{(c)}.}
\label{fig:BwSType}
\end{figure}

\twocolumngrid
\noindent Theorem \ref{th:verif1} is of the form of Equation \ref{eq:stateend2} with $b\leq\kappa/(v+1)$|and $\widetilde{\rho}_{\textup{out}}$ is a quantum state containing the effects of noise in single-qubit gates. Expressing $\widetilde{\rho}_{\textup{out}}$ as
\begin{equation}~\label{eq:rhoprime}
\widetilde{\rho}_{\textup{out}}=h\textup{ }\widetilde{\tau}^{\textup{ tar}}_{1}\otimes|\textup{acc}\rangle\langle\textup{acc}|+(1-h)\textup{ }\widetilde{\tau}^{\textup{ tar}}_{2}\otimes|\textup{rej}\rangle\langle\textup{rej}|\textup{ ,}
\end{equation}  where $\widetilde{\tau}^{\textup{ tar}}_{1}$ and $\widetilde{\tau}^{\textup{ tar}}_{2}$ are arbitrary states for the target and $0\leq h\leq1$, we thus have
\begin{align}~\label{eq:end2}
&{\rho}_{\textup{out}}^\star\textup{=}g\bigg[b\textup{ }{\tau}_{\textup{out}}^{\prime\textup{ tar}}\otimes|\textup{acc}\rangle\langle\textup{acc}|+(1-b)\bigg(l\textup{}\textrm{ }{\sigma}_{\textup{out}}^{\textup{ tar}}\otimes|\textup{acc}\rangle\langle\textup{acc}|\cr
&+(1-l){\tau}_{\textup{out}}^{\textup{ tar}}\otimes|\textup{rej}\rangle\langle\textup{rej}|\bigg)\bigg]+(1-g)\bigg[h\textup{ }\widetilde{\tau}^{\textup{ tar}}_{1}\otimes|\textup{acc}\rangle\langle\textup{acc}|\cr
&+(1-h)\textup{ }\widetilde{\tau}^{\textup{ tar}}_{2}\otimes|\textup{rej}\rangle\langle\textup{rej}|\bigg]
\end{align}
As it can be seen, the probability that the target is in the wrong state and the flag bit is in the state $|\textup{acc}\rangle$ is $gb+(1-g)h\leq g\kappa/(v+1)+(1-g)h$, where we used that $b\leq\kappa/(v+1)$ from Theorem \ref{th:verif1}. This probability reaches its maximum for $h=1$, therefore we have $\varepsilon=g\kappa/(v+1)+1-g$. Note that if $r_{\textup{max},j}^{(k)}\leq r_0\ll1$, then
\begin{align}
g\geq\prod_{k=1}^{v+1}\prod_{j=1}^{m}(1-r_0)\approx 1-m(v+1)r_0+O(r_0^2).
\end{align}
Thus, if $r_0\ll1/m(v+1)$, then $g\approx1$ and $\varepsilon\approx1.7/(v+1)$.

It is worth noting that our Theorem \ref{th:verif1} also holds if single-qubit gates suffer unbounded noise, provided that this noise is gate-independent. Indeed, if $\E_j^{(k)}=\E$ does not depend on the parameters in $\U^{\prime\prime(k)}_j$ (cfr. Figure \ref{fig:noisyprotocol} for notation), using $\E_j^{(k)}\C\Z_j\F_j^{(k)}=\C\Z_j\F_j^{\textup{ }\prime\textup{ }(k)}$ (with $\F_j^{\textup{ }\prime\textup{ }(k)}=\C\Z_j^{-1}\E_j^{(k)}\C\Z_j\F_j^{(k)}$) we can factor this noise into that of $\C\Z_j$ and prove $\varepsilon=\kappa/(v+1)$ with the same arguments used in Theorem \ref{th:verif1}.  Similarly, we also expect our Theorem \ref{th:verif2} to hold if noise in single-qubit gates has a \textit{weak} gate-dependence, as is the case for some of the protocols centered around randomized benchmarking \cite{MGE11}. We leave the analysis of weakly gate-dependent noise to future works.\\

\noindent{\textbf{5. Mesothetic verification protocol: }} In Box {4} in the Methods we translate our accreditation protocol into a cryptographic protocol, obtaining what we call ``mesothetic'' verification protocol. To verify an $n$-qubit computation, in the mesothetic protocol the verifier (Alice) must possess a device that can receive $n$ qubits from the prover (Bob), implement single-qubit gates on all of them and send the qubits back to the Bob. In Appendix \ref{app:crypto} we present Theorems {D1} and {D2}, which are the counterparts of Theorems \ref{th:verif1} and \ref{th:verif2} for the cryptographic protocol. In the first two Theorems the number $\varepsilon$ is replaced by the soundness $\varepsilon_{\textup{cr}}$ (cfr Definition \ref{def:ver2} in Appendix \ref{app:crypto}), namely the probability that Alice accepts a wrong output for the target when Bob is cheating.

Our mesothetic verification protocol is different from prepare-and-send \cite{ABE08,FK12,BFKW13,B15,KD17,KW17,FKD17,ABEM17} or receive-and-measure \cite{HM15,MF16,HKSE17,MK18,TMMMF18} cryptographic protocols in that Alice encrypts the computation through the QOTP during the actual implementation of the circuits, and not before or after the implementation. To do this, she must possess an $n$-qubit quantum memory and be able to execute single-qubit gates. Despite being scalable, our protocol is more demanding that those in Ref. \cite{FK12,BFKW13,B15,KD17,KW17,FKD17,ABEM17,HM15,MF16,HKSE17,MK18,TMMMF18}, where Alice only requires a single-qubit memory. This suggests the interesting possibility that protocols optimized for experiments may translate into more demanding cryptographic protocols and vice-versa.

Similarly to post-hoc verification protocols \cite{FH15,MF16}, our protocol is not blind. Alice leaks crucial information to Bob regarding the target circuit, such as the position of two-qubit gates. This is not a concern for our goals, as verifiability in our protocol relies on Bob being incapable to distinguish between target and trap circuits, i.e. to retrieve the number $v_0$, see Lemma {D3} in Appendix \ref{app:crypto}. 

Blindness may be required to protect user's privacy in future scenarios of delegated quantum computing \cite{F16}. In Appendix \ref{app:crypto} we thus show how to make our protocol blind. This requires recompiling the target circuit into a circuit in normal form with fixed $cZ$ gates, such as the brickwork-type circuit in Figure \ref{fig:BwSType}. This yields an increase in circuit depth, hence the minimal overheads of our protocol must be traded for blindness.

\newpage
\onecolumngrid

\begin{figure}[H]
  \centering
 \begin{subfigure}{0.47\textwidth}
		\centering		
		\includegraphics[scale=0.6]{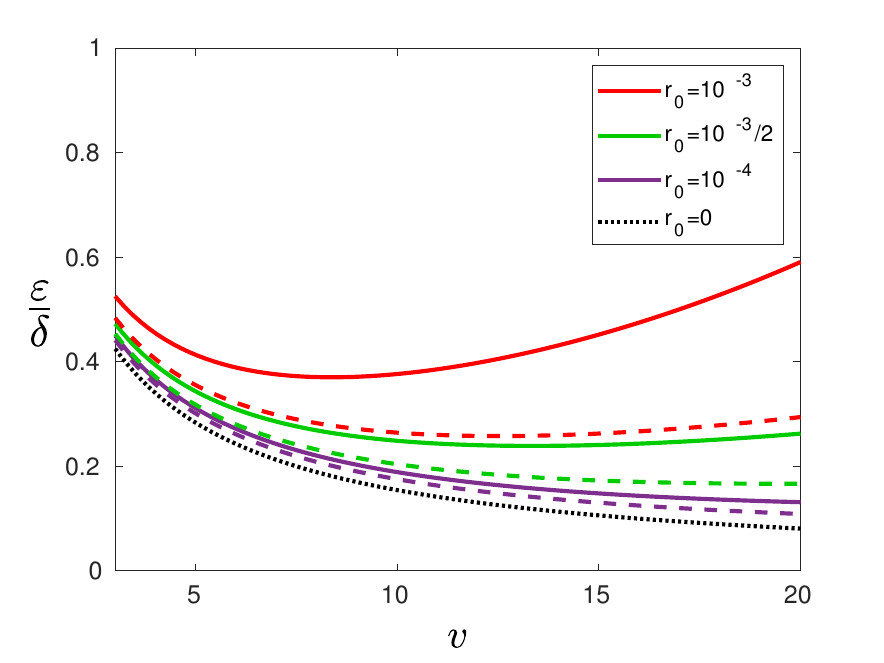}
        \caption{}
        \label{fig:GHZ}
 \end{subfigure}
 ~
 \begin{subfigure}{0.47\textwidth}
		\centering		
		\includegraphics[scale=0.6]{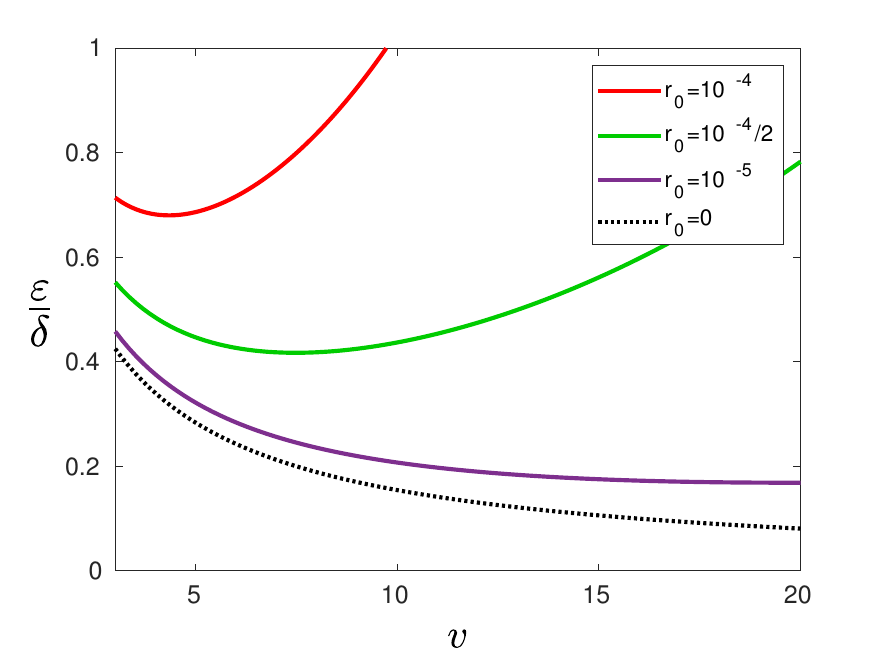}
        \caption{}
        \label{fig:RCS}
 \end{subfigure}
 \caption{RHS of Inequality \ref{eq:vardistboundednoise} for \textbf{(a)} target circuit preparing and measuring $n$-qubit GHZ states, with $n=m=7$ (dashed lines) and $n=m=10$ (solid lines) and \textbf{(b)} target circuit implementing a pseudo-random circuit for supremacy experiment \cite{B&al16} with $n=62$ qubits and circuit depth $m=34$. In these plots we assume that all operations are affected by bounded noise. We also assume that single-qubit state preparation, single-qubit measurements and $cZ$-gates have error rates $r_0$, and that all single-qubit gates have error rates $r_0/10$.}
 \label{fig:plot}
\end{figure}

\twocolumngrid

\noindent

\section*{Discussion}

We have presented a trap-based accreditation protocol for the outputs of NISQ devices. Our protocol is scalable and practical, and relies on minimal assumptions on the noise. Specifically, our protocol requires that single-qubit gates suffer bounded (but potentially gate-dependent and non-local in space and time) noise.

A single protocol run ends by either accepting or rejecting the output of the target circuit that we seek to accredit. We can then run our accreditation protocol multiple times (with the same target and with the same number of traps), each time keeping the output if the protocol accepts and discarding it if the protocol rejects. After multiple runs with i.i.d. noise, our protocol allows to bound the variation distance between noisy and noiseless probability distribution of the accepted outputs (Equation \ref{eq:eq1}).

Real-world devices can be accredited by running our protocol on them. The accreditation is provided by bounds on the variation distance that rely on $\varepsilon$, which
we obtained in Theorems \ref{th:verif1} and \ref{th:verif2}, and the acceptance probability prob(acc) of our accreditation protocol. The latter is estimated experimentally by running our protocol multiple times on the device being accredited.

Some noise models allow to lower-bound the acceptance probability analytically and consequently to upper-bound the variation distance. For instance, if all operations $\E_S^{(p)}$ in our protocol suffer bounded noise and have error rates $r_p$, we can write the state of the system at the end of the protocol as
\begin{equation}
\rho_{\textup{out}}=\delta\textup{ }\sigma_{\textup{out}}^{\textup{tar}}\otimes|\textup{acc}\rangle\langle\textup{acc}|+(1-\delta)\widetilde{\tau}\textup{ ,}
\end{equation}
where $\sigma_{\textup{out}}^{\textup{tar}}$ is the state of the target at the end of a noiseless protocol run, $\widetilde{\tau}$ is an arbitrary state for target and flag and $\delta=\prod_p(1-r_p)\in[0,1]$. This yields prob(acc)$\geq \delta$ and (cfr. Equation \ref{eq:bbbboundvardistt})
\begin{equation}
\label{eq:vardistboundednoise}
\frac{1}{2}\sum_{\overline{s}}\big|p_{\textup{noiseless}}(\overline{s})-p_{\textup{noisy}}(\overline{s})\big|\leq\frac{\varepsilon}{\textup{prob(acc)}}\leq \frac{\varepsilon}{\delta}
\end{equation}
In Figure \ref{fig:plot} we plot the RHS of the above equation. Plots of this type can be used to seek error rates that will provide the desired upper-bound on the variation distance.

Considering the error rates of present NISQ devices \cite{H&al14,H&al16,IBM,B&al14,SMCG16} we expect that our protocol may provide worthwhile upper-bounds for target circuits with up to $n\approx7$ qubits and $m\approx7$ bands (Figure \ref{fig:GHZ}). We also expect that larger target circuits may yield upper-bounds that are too large to be useful. In fact, for large target circuits, it is also possible that none of the protocol runs will accept the output of the target, and thus that our protocol will provide no upper-bound. Nevertheless, it is worth remarking that this does not indicate that our accreditation protocol is unable to accredit computations on NISQ devices. On the contrary, by providing large upper-bounds, our protocol reveals that the device being tested suffers high levels of noise and that its outputs are not credible.

Our work leaves several open questions. Our Theorem \ref{th:verif2} shows that our protocol requires reducing the error rates of single-qubit gates with the size of the target circuit. This requirement is similar to that found in other works \cite{GKW15,KD17,TMMMF18}, and is a known obstacle towards scalable quantum computing. A strategy that has been exploited in previous works is to incorporate fault-tolerance into the existing protocols \cite{GKW15,KD17}. Another has been to define verifiable fault-tolerance using notions such as acceptability and detectability \cite{FH17}. Interfacing fault tolerance with accreditation is an interesting challenge for the future.

Another open question regards the applicability of our accreditation protocol if single-qubit gates suffer unbounded noise. In its current state, the analysis of our protocol does not account for unbounded gate-dependent noise in single-qubit gates, including unitary errors such as over- or under-rotations. The reason is that the QOTP (which maps coherent errors into classically correlated Pauli errors) is applied at the level of single-qubit gates. Unbounded errors that depend on the gates used to randomize arbitrary noise processes to Pauli errors are an obstacle to other works including cryptographic protocols \cite{KD17} and protocols based on randomized benchmarking \cite{EAZ05,K&al07,DCEL09,Erhard&al19,CGFF17}.

Finally, with the {mesothetic} protocol we show how to adapt our protocol to the cryptographic setting. In the mesothetic protocol the verifier requires an $n$-qubit memory and the ability to execute single-qubit gates. {This protocol is }more demanding than several existing cryptographic protocols \cite{FK12,BFKW13,B15,KD17,KW17,FKD17,ABEM17,HM15,MF16,HKSE17,MK18,TMMMF18} requiring single-qubit memory for the verifier. An interesting question is whether {a mesothetic} protocol can be devised that only requires single-qubit gates and single-qubit memory for the verifier.

\section*{Acknowledgments}
This research was supported by the UK EPSRC
(EP/K04057X/2) and the UK Networked Quantum In-
formation Technologies (NQIT) Hub (EP/M013243/1).
We acknowledge helpful discussions with \textcolor{black}{Joel Wallman, Elham Kashefi, Dominic Branford, Zhang Jiang, Joseph Emerson}.

\section*{Methods}

\noindent\textbf{Overhead of our accreditation protocol: } Here we count the overhead of our protocol. Our protocol has no quantum overhead, as all circuits have the same size as the one being verified. The classical overhead consists in $O(nm)$ bits for each of the $v+1$ computations. Specifically, the target computation has an overhead of $2nm+n$ bits (the $2nm$ random bits $\alpha_{i,j},\alpha_{i,j}'$ and the $n$ random bits $\gamma_i$ in Routine 1), while the traps have an overhead of at most $2nm+n+nm$ bits (the $2nm+n$ random bits in Routine 1 and at most $nm$ random bits in Routine 2).\\

\vspace{0.8cm}

\newpage
\begin{footnotesize}
\noindent\fbox{
  \parbox{0.47\textwidth}{

\noindent\textbf{Box 1. {Accreditation} Protocol.}\\

\noindent\textbf{Input:}\\
$1.$ A target circuit that takes as input $n$ qubits in the state $\ket{+}$, contains only single-qubit gates and $cZ$ gates arranged in $m$ bands and ends with Pauli-$X$ measurements (Figure \ref{fig:circuit}).\\
$2.$ The number $v$ of trap circuits.\\

\noindent\textbf{Routine:}\\
\noindent 1. Choose a random number $v_0\in\{1,\dots,v+1\}$ and define $\{U_{i,j}^{(v_0)}\}=\{U_{i,j}\}$, where $\{U_{i,j}\}$ is the set of single-qubit gates in the target circuit.\\

\noindent 2. For $k=1,\dots,v+1$: If $k\neq v_0$ (trap circuit), run Routine 2 and obtain the set of single-qubit gates $\{U_{i,j}^{(k)}\}$ for the $k$-th trap circuit.\\

\noindent 3. For $k=1,\dots,v+1$: Run Routine 1 and obtain $\{U_{i,j}^{\prime\prime(k)}\}$, together with the bit-string $(\alpha^{(k)}_{1,m},\ldots\alpha^{(k)}_{n,m})$.\\

\noindent 4. For $k=1,\dots,v+1$: 
\begin{itemize}
\item[4.1] Create a state $\rho_{\textup{in}}=|+\rangle^{\otimes n}$.
\item[4.2] Implement circuit $k$ with single-qubit gates from the set $\{U_{i,j}^{\prime\prime(k)}\}$ and obtain output $\overline{s}^{(k)}=({s}^{(k)}_1,\ldots,{s}^{(k)}_n)$. Next, for all $i=1,\ldots,n$, recompute ${s}^{(k)}_i$ as $s_i^{(k)}\oplus\alpha^{(k)}_{i,m}$.
\end{itemize}
\noindent 5. Initialize a flag bit to $|\textup{acc}\rangle=|\textup{0}\rangle$. Then, for $k=1,\dots,v+1$: if $\overline{s}^{(k)}\neq\overline{0}$ and $k\neq v_0$ (trap circuit), set the flag bit to $|\textup{rej}\rangle=|\textup{acc}\oplus1\rangle$.\\

\noindent\textbf{Output:}\\
The output $\overline{s}^{(v_0)}$ of the target circuit and the flag bit.}}
\end{footnotesize}

\noindent\begin{center}
\begin{footnotesize}
\fbox{
  \parbox{0.47\textwidth}{
\textbf{Box 2. Routine {1}. [Quantum One-Time Pad].}\\
  
	\noindent\textbf{Input:}\\
	A set $\{U_{i,j}\}$ of single-qubit gates, for $j=1,\ldots,m$ and $i=1,\ldots,n$.\\
	
	\noindent\textbf{Routine:}\\
	\noindent {1.} For $j=1,\ldots,m$ and $i=1,\ldots,n$:\\
	\noindent Choose two random bits $\alpha_{i,j}$ and $\alpha^\prime_{i,j}$. Next, define $U'_{i,j}=X^{\alpha'_{i,j}}Z^{\alpha_{i,j}}U_{i,j}$.\\

	\noindent {2.} For $i=1,\ldots,n$:\\
	Choose a random bit $\gamma_i$ and define $U_{i,1}''=U_{i,1}'X^{\gamma_{i}}$.\\

	\noindent {3.} For $j=1,\ldots,m-1$:\\ 
	\noindent Using Equations \ref{eq:czconjj1} and \ref{eq:czconjj2} define $U_{i,j+1}^{\prime\prime}$ so that
	\begin{align}			&\big(\otimes_iU''_{i,j+1}\big) \widehat{cZ}_{j} \big(\otimes_iU'_{i,j}\big)=\big(\otimes_iU'_{i,j+1}\big) \widehat{cZ}_{j} \big(\otimes_iU_{i,j}\big)
	\end{align}	
	 where $\widehat{cZ}_{j}$ is the entangling operation in the $j$th band.\\

	\noindent\textbf{Output:}\\
	The set $\{U''_{i,j}\}$ and the $n$-bit string $(\alpha_{1,m},\ldots,\alpha_{n,m})$.
}}
\end{footnotesize}
\end{center}

\onecolumngrid

\begin{footnotesize}
\vspace{0.4cm}
\noindent\fbox{
  \parbox{1\textwidth}{
\noindent{\textbf{Box 3. Routine 2. [Single-qubit gates for trap circuits].}}\\

\noindent\textbf{Input:}\\
The target circuit.\\

\noindent\textbf{Routine:}\\
\noindent 1. Initialize the set $\{U_{i,j}=I_i\}$, for $i=1,\dots,n$ and $j=1,\dots,m$.\\

\noindent 2. For all $j=1,\dots,m-1$:
\begin{itemize}
\item[2.1] For all $i=1,\dots,n$: If in band $j$ of the target circuit qubits $i$ and $i'$ are connected by a $cZ$ gate, set\\
$\bullet$ $U_{i,j}=S_iU_{i,j}$ and $U_{i',j}=H_{i'}U_{i',j}$ with probability 1/2.\\
$\bullet$ $U_{i,j}=H_iU_{i,j}$ and $U_{i',j}=S_{i'}U_{i',j}$ with probability 1/2.\\
Otherwise, set $U_{i,j}=S_iU_{i,j}$ or $U_{i,j}=H_{i}U_{i,j}$ with probability 1/2.
\item[2.2] For all $i=1,\dots,n$: Set $U_{i,j+1}=U_{i,j}^{\dagger}$.
\end{itemize}
\noindent 3. For all $i=1,\dots,n$:\\
Choose a random bit $t\in\{0,1\}$. Next, set $U_{i,1}=U_{i,1}H^t$ and $U_{i,m}=H^tU_{i,m}$.\\

\noindent\textbf{Output:}\\
The set $\{U_{i,j}\}$.
}}
\end{footnotesize}

\vspace{0.3cm}

\begin{small}
\noindent\fbox{
  \parbox{\textwidth}{
\noindent\textbf{{Box 4}: Mesothetic Protocol} (further details in Appendix \ref{app:crypto})\\

\noindent\textbf{Input:} \\
A classical description of the target circuit and the number $v$ of traps. (The input is known to both Alice and Bob).\\
 
\noindent\textbf{Preliminary Operations:}\\
\noindent{1.} Alice randomly chooses which circuit $v_0\in\{1,\ldots,v+1\}$ will be used to implement the target. Next she defines $\{U_{i,j}^{(v_0)}\}=\{U_{i,j}\}$, where $\{U_{i,j}\}$ is the set of single-qubit gates in the target circuit.\\
\noindent{2.} For $k=1,\ldots,v+1$: If $k\neq v_0$ (trap circuit), Alice runs Routine 2 and obtains the set $\{U_{i,j}^{(k)}\}$ of single-qubit gates for the $k$-th circuit.\\
\noindent{3.} For $k=1,\ldots,v+1$: Alice runs Routine 1 and obtains the set of gates $\{U_{i,j}^{\prime\prime(k)}\}$, together with the random bits $\alpha_{i,m}^{(k)}$.\\

\noindent\textbf{Routine:}\\
\noindent{4.} For all $k=1,\ldots,v+1$, Alice and Bob interact as follows:
\begin{itemize}
\item[{4.1}] Bob creates $n$ qubits in state $\ket{+}$.
\item[{4.2}] For $j=1,\ldots,m$:
\begin{itemize}
\item[{4.2.1}] Bob sends all the qubits to Alice. For $i=1,\ldots,n$, Alice executes $U_{i,j}^{\prime\prime(k)}$ on qubit $i$. Finally, Alice sends all the qubits back to Bob.
\item[{4.2.2}] Bob applies the entangling gates $\widehat{cZ}_{j}$ contained in the $j$-th band of the target circuit.
\end{itemize}
\item[{4.3}] For $i=1,\ldots,n$:  Bob measures qubit $i$ in the Pauli-$X$ basis and communicates the output to Alice. Alice bit-flips the output if $\alpha_{i,m}^{(k)}=1$, otherwise she does nothing. Next, if $k\neq v_0$ and this output is $s_i=1$, Alice aborts the protocol, otherwise she does nothing.
\end{itemize}
\noindent{5.} Alice initializes a flag bit to the state $|\textup{acc}\rangle=|\textup{0}\rangle$. Next, for all $k=1,\ldots,v+1$: if $k\neq v_0$ (trap circuit) and $s_i^{(k)}\oplus\alpha_{i,m}^{(k)}\neq{0}$ for some $i\in\{1,\ldots,n\}$, Alice sets the flag bit to $|\textup{rej}\rangle=|\textup{acc}\oplus 1\rangle$.\\

\noindent\textbf{Output:}\\
The outputs of the target circuit and the flag bit.}}

\end{small}
\vspace{0.5cm}

\bibliographystyle{unsrtnat}
\bibliography{verification}

\onecolumngrid

\newpage
\vspace{\columnsep}

\appendix 
\noindent\textbf{Notation:}  In these Appendices we will indicate the action of the round of single-qubit gates in a band $j\in\{1,\ldots,m\}$ of a circuit $k\in\{1,\ldots,v+1\}$ as $\U_{j}^{\prime\prime(k)}(\rho_S)=\otimes_{i=1}^nU^{\prime\prime(k)}_{i,j}(\rho_S)U^{\dagger\prime\prime(k)}_{i,j}$, where $\rho_S$ is the state of the system. Similarly, we will indicate the action of a round of $cZ$ gates on the system as $\C\Z_{j}(\rho_S)=\widehat{cZ}_{j}(\rho_S)\widehat{cZ}_{j}$, where $\widehat{cZ}_{j}$ is the tensor product of all $cZ$ gates in band $j$ in the target circuit, and the action of $n$-qubit Pauli operators as $\p\in\{\I,\X,\Y,\Z\}^{\otimes n}$, where $\I(\rho)=\rho$, $\X(\rho)=X\rho X$, $\Y(\rho)=Y\rho Y$, $\Z(\rho)=Z\rho Z$ are single-qubit Pauli operators.

In Appendix \ref{subapp:lem1} we provide statement and proof of Lemma \ref{lem:pauliz}, In Appendix \ref{subapp:lem2} we provide statement and proof of Lemma \ref{lem:ct} , In Appendix \ref{app:th1} we prove Theorem \ref{th:verif1} and in Appendix \ref{app:crypto} we prove soundness of the mesothetic protocol.

\section{Statement and Proof of Lemma 1}~\label{subapp:lem1}
We now present and prove Lemma \ref{lem:pauliz}, which is as follows:
\begin{lemma}~\label{lem:pauliz}
Suppose that all single-qubit gates in all the circuits implemented in our protocol are noiseless, and that state preparation, measurements and two-qubit gates suffer noise of the type type {N1}. For a fixed choice of single-qubit gates ${\U}^{(1)}_{1},\ldots,{\U}^{(v+1)}_{m}$, summed over all the random bits $\alpha_{i,j},\alpha_{i,j}',\gamma_i$ (cfr. Routine 1), the joint state of the target circuit and of the traps after they have all been implemented is of the form
\begin{align}~\label{eq:rhoout}
{\rho}_{\textup{out}}\big({\U}^{(1)}_{1},\ldots&,{\U}^{(v+1)}_{m}\big)
=\sum_{\substack{\overline{s}^{(1)},\ldots,\overline{s}^{(v+1)}}}\textup{ }\sum_{\substack{{\p}^{(1)}_{0},\ldots,{\p}^{(v+1)}_{m}}}\frac{
\textup{prob}\big({\p}^{(1)}_{0},\ldots,{\p}^{(v+1)}_{m}\big)}{2^{n(v+1)}}\textup{ }\times\cr
&\bigotimes_{k=1}^{v+1}\textup{ }\langle+|^{\otimes n}\bigg[\Z^{\overline{s}^{(k)}}{\p}^{(k)}_{m}{\U}^{(k)}_{m}\circ_{j=1}^{m-1}\bigg({\C\Z}_{j} {\p}^{(k)}_{j}{\U}^{(k)}_{j}\bigg)\circ{\p}^{(k)}_{0}\big(\rho_{\textup{in}}\big)\bigg]|+\rangle^{\otimes n}\bigg(\otimes_{i} Z_i^{s^{(k)}_i}|+\rangle_{i}\langle+|Z_i^{s^{(k)}_i}\bigg)
\end{align}
where $\rho_{\textup{in}}=\otimes_{i}|+\rangle_{i}\langle+|$, $\overline{s}^{(k)}=({s}_1^{(k)},\ldots,s^{(k)}_n)$ is a binary string representing the output of the $k$-th circuit, $\Z^{\overline{s}^{(k)}}(\rho)=\otimes_{i} Z_i^{s^{(k)}_i}\rho Z_i^{s^{(k)}_i}$ and $\textup{prob}\big({\p}^{(1)}_{0},\ldots,{\p}^{(v+1)}_{m}\big)$ is the joint probability of a collection of Pauli errors ${\p}^{(1)}_{0},\ldots,{\p}^{(v+1)}_{m}$ affecting the system, with ${\p}^{(k)}_{1},\ldots,{\p}^{(k)}_{m-1}\in\{\I,\X,\Y,\Z\}^{\otimes n}$ and ${\p}^{(k)}_{0},{\p}^{(k)}_{m}\in\{I,\Z\}^{\otimes n}$ for all $k$.
\end{lemma}
This Lemma shows that if single-qubit gates are noiseless, the QOTP allows to reduce
noise of the type {N1} to classically correlated Pauli errors. These Pauli errors affect each circuit after state preparation ($\p_{0}^{(k)}$), before each entangling operation $\C\Z_{j}$ ($\p_{j}^{(k)}$, for $j\in\{1,\ldots,m-1\}$) and before the measurements ($\p_{m}^{(k)}$). Errors in the $cZ$ gates can be Pauli-$X$, $Y$ and $Z$, while those in state preparation and measurements are Pauli-$Z$ (this is because their Pauli-$X$ components stabilize $\rho_{\textup{in}}$ and Pauli-$X$ measurements respectively).

The main tool used in this section is the {``Pauli Twirl''~\cite{DCEL09}.}\\

\noindent\textup{\textbf{[{Pauli Twirl}]. }}\textit{Let $\rho$ be a $2^n\times2^n$ density matrix and let $P,P'$ be two $n$-fold tensor products of the set of Pauli operators $\{I,X,Y,Z\}$. Denoting by $\{Q_r\}$ the set of all $n$-fold tensor products of the set of Pauli operators $\{I,X,Y,Z\}$},
\begin{equation}~\label{eq:twirl}
\sum_{r=1}^{4^n}Q_rPQ_r\rho Q_rP'Q_r=0\textrm{ $\forall$ } {P\neq P'.}
\end{equation}
We will also use a restricted version of the Pauli Twirl, which is proven in Ref. \cite{KD17}\\

\noindent\textup{\textbf{[{Restricted Pauli Twirl}]. }}\textit{Let $\rho$ be a $2^n\times2^n$ density matrix and let $P,P'$ be two $n$-fold tensor products of the set of Pauli operators $\{I,Z\}$. Denoting by $\{Q_r\}$ the set of all $n$-fold tensor products of the set of Pauli operators $\{I,X\}$},
\begin{equation}~\label{eq:twirl}
\sum_{r=1}^{2^n}Q_rPQ_r\rho Q_rP'Q_r=0\textrm{ $\forall$ }{P\neq P'.}
\end{equation}
\textit{The same holds if $P$ and $P'$ are two $n$-fold tensor products of the set of Pauli operators $\{I,X\}$ and $\{Q_r\}$ is the set of all $n$-fold tensor products of the set of Pauli operators $\{I,Z\}$.}
\newpage
\begin{figure}[H]
\centering
\begin{tikzpicture}[scale=0.86, every node/.style={scale=0.99}]

\foreach \x in {1,...,6}
\node at (-0.4-1,4.8-\x*0.8) {\scriptsize $\ket{+}_{\x}$};
\node at (-0.4-1,-2*0.8) {$\ket{e_0}$};

\foreach \x in {0,...,5}
\draw (-0.0-1,4.0-\x*0.8) -- (5.7,4.0-\x*0.8);
\draw (-0.0-1,-2*0.8) -- (5.7,-2*0.8);

\foreach \x in {0,...,5}
\draw [dashed] (5.7,4.0-\x*0.8) -- (6.8,4.0-\x*0.8);
\draw [dashed] (5.7,-2*0.8) -- (6.8,-2*0.8);

\foreach \x in {0,...,5}
\draw (6.8,4.0-\x*0.8) -- (13.5,4.0-\x*0.8);
\draw (6.8,-2*0.8) -- (13.5,-2*0.8);

\draw [fill=white] (0.3-1,-1.6-0.3) rectangle (1.0-1,4.3);
\node at (0.65-1, 4.0-3.5*0.8) {$\mathbf{{R}}$};

\foreach \x in {0,...,5}
\draw [fill=white] (0.3,4.0-\x*0.8-0.3) rectangle (1.0,4.0-\x*0.8+0.3);

\foreach \x in {1,...,6}
\node at (0.65,4.8-\x*0.8) {\scriptsize $U^{\prime\prime}_{\x,1}$};

\draw [fill=white] (2.3-1,-1.6-0.3) rectangle (3.1-1,4.3);
\node at (2.7-1, 4.0-3.5*0.8) {$\mathbf{ {F}}_{1}$};

\foreach \x in {1,...,4}
\draw [fill=black] (2.3,4.8-\x*0.8) circle [radius=0.07cm];

\draw (2.3,2.4) -- (2.3,1.6);
\draw (2.3,3.2) -- (2.3,4.0);

\draw [thick,red,dashed] (2.6,4.5) -- (2.6,-1.6-0.4);

\foreach \x in {0,...,5}
\draw [fill=white] (1.3+1*1.6,4.0-\x*0.8-0.3) rectangle (2.0+1*1.6,4.0-\x*0.8+0.3);

\foreach \x in {1,...,6}
\node at (1.65+1*1.6,4.8-\x*0.8) {\scriptsize $U^{\prime\prime}_{\x,2}$};

\draw [fill=black] (3.7+1*1.6,4.8-1*0.8) circle [radius=0.07cm];
\draw [fill=black] (3.4+1*1.6,4.8-3*0.8) circle [radius=0.07cm];
\draw [fill=black] (3.7+1*1.6,4.8-4*0.8) circle [radius=0.07cm];
\draw [fill=black] (3.4+1*1.6,4.8-6*0.8) circle [radius=0.07cm];

\draw (3.4+1*1.6,0.0) -- (3.4+1*1.6,2.4);
\draw (3.7+1*1.6,4.0) -- (3.7+1*1.6,1.6);

\draw [fill=white] (3.9,-1.6-0.3) rectangle (4.7,4.3);
\node at (4.3, 4.0-3.5*0.8) {$\mathbf{ {F}}_{2}$};

\draw [thick,red,dashed] (5.5,4.5) -- (5.5,-1.6-0.4);

\draw [thick,red,dashed] (7,4.5) -- (7,-1.6-0.4);

\foreach \x in {0,...,5}
\draw [fill=white] (7.3,4.0-\x*0.8-0.3) rectangle (8.4,4.0-\x*0.8+0.3);

\foreach \x in {1,...,6}
\node at (7.85,4.8-\x*0.8) {\scriptsize $U^{\prime\prime}_{\x,m-1}$};

\draw [fill=white] (8.7,-1.6-0.3) rectangle (10,4.3);
\node at (9.35, 4.0-3.5*0.8) {$\mathbf{ {F}}_{m-1}$};

\draw [fill=black] (10.3,4.8-2*0.8) circle [radius=0.07cm];
\draw [fill=black] (10.3,4.8-5*0.8) circle [radius=0.07cm];

\draw (10.3,0.8) -- (10.3,3.2);

\draw [thick,red,dashed] (10.6,4.5) -- (10.6,-1.6-0.4);

\foreach \x in {0,...,5}
\draw [fill=white] (10.9,4.0-\x*0.8-0.3) rectangle (11.7,4.0-\x*0.8+0.3);

\foreach \x in {1,...,6}
\node at (11.3,4.8-\x*0.8) {\scriptsize $U^{\prime\prime}_{\x,m}$};

\draw [fill=white] (12,-1.6-0.3) rectangle (12.8,4.3);
\node at (12.4, 4.0-3.5*0.8) {$\mathbf{ {M}}$};

\foreach \x in {1,...,6}
\draw (13.5,4.8-\x*0.8-0.05) -- (14.,4.8-\x*0.8-0.05);
\foreach \x in {1,...,6}
\draw (13.5,4.8-\x*0.8+0.05) -- (14.,4.8-\x*0.8+0.05);

\foreach \x in {0,...,5}
\draw [fill=white] (13.5,\x*0.8) circle [radius=0.3cm];
\foreach \x in {1,...,6}
\node at (13.5,4.8-\x*0.8) {\scriptsize $X$};

\end{tikzpicture}
\caption{\small Noisy implementation of the 6-qubit target circuit in Figure \ref{fig:circuit}. The noise in state preparation is described by the unitary $ {\mathbf{R}}$, that in the measurements by $ {\mathbf{M}}$, that in the $cZ$-gates in a band $j=1,\ldots,m-1$ by $ {\mathbf{F}}_{j}$. All these unitaries act simultaneously on the system and on the environment (initially in the ground state $|e_0\rangle$).}
\label{fig:noise}
\end{figure}
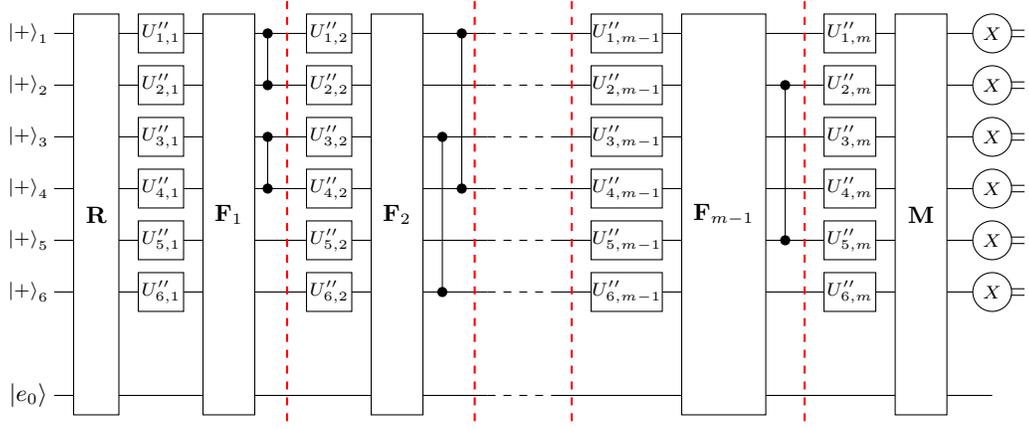

\begin{proof} \textit{(Lemma \ref{lem:pauliz})} We start proving the lemma for the case where we run a single circuit ($v=0$), and then we generalize to multiple circuits ($v>0$). Including all purifications in the environment, we can rewrite the noise as unitary matrices acting on system and environment (for clarity we write these unitaries in bold font). Doing this, for a fixed choice of gates ${\U}^{\prime\prime}_{j},\ldots,{\U}^{\prime\prime}_{m}$ (which depend on the choice of gates ${\U}^{}_{1},\ldots,{\U}^{}_{m}$ and on all the random bits $\alpha_{i,j},\alpha_{i,j}^\prime,\gamma_{i}$, cfr. Routine 1), the state of the system before the measurement becomes
\begin{align}
\label{eq:rhomultiplebands}
&\rho\big({\U}^{\prime\prime}_{1},\ldots,{\U}^{\prime\prime}_{m}\big)
=\cr
&\textup{Tr}_{\textup{E}}\bigg[\mathbf{ {M}}\textrm{ } {U}^{\prime\prime}_{m} \widehat{cZ}_{m-1}\mathbf{ {F}}_{m-1}  {U}^{\prime\prime}_{m-1}\ldots\widehat{cZ}_{1}\mathbf{ {F}}_{1}  {U}^{\prime\prime}_{1}\mathbf{ {R}}\bigg({\rho}_{\textrm{in}}\otimes\rho_E\bigg)\mathbf{ {R}}^\dagger\textrm{ } {U}^{\prime\prime\dagger}_{1} \mathbf{ {F}}_{1}^\dagger \widehat{cZ}_{1}\ldots {U}^{\prime\prime\dagger}_{m-1} \mathbf{ {F}}_{m-1}^\dagger \widehat{cZ}_{m-1} {U}^{\prime\prime\dagger}_{m}\mathbf{ {M}}^\dagger\bigg]\textup{ ,}
\end{align}
{where $\rho_{\textrm{in}}=\otimes_i|+\rangle_i\langle+|$, $\rho_E=|e_0\rangle\langle e_0|$ is the initial state of the environment, $ {U}^{\prime\prime}_{j}=\otimes_iU_{i,j}^{\prime\prime}$ are the gates output by Routine 1, the unitary matrix $ {\mathbf{R}}$ represents the noise in state preparation, the unitary matrix $ {\mathbf{M}}$ represents the noise in the measurements, $ \widehat{cZ}_{j}\mathbf{ {F}}_{j}$  is the noisy round of entangling gates in a band $j$ and Tr$_{\textup{E}}\big[\cdot\big]$ is the trace over the environment.}

For simplicity, we first prove our result for a circuit with $m=2$ bands and generalize to $m>2$ bands later. In this case, defining an orthonormal basis $\{|e_p\rangle\langle e_p|\}$ for the environment, the state in Equation \ref{eq:rhomultiplebands} is
\begin{align}
\rho{\big({\U}^{{\prime\prime}}_{1},{\U}^{\prime\prime}_{2}\big)}
=\sum_p\langle e_p|\mathbf{ {M}}\textrm{ } {U}^{\prime\prime}_{2} \widehat{cZ}_{1}\mathbf{ {F}}_{1}  {U}^{\prime\prime}_{1}\mathbf{ {R}}\big({\rho}_{\textrm{in}}\otimes|e_0\rangle\langle e_0|\big)\mathbf{ {R}}^\dagger\textrm{ } {U}^{\prime\prime\dagger}_{1} \mathbf{ {F}}_{1}^\dagger \widehat{cZ}_{1} {U}^{\prime\prime\dagger}_{2}\mathbf{ {M}}^\dagger|e_p\rangle\textup{ .}
\end{align}
Introducing resolutions of the identity on the environment before and after every noise operator, we have
\begin{align}
\rho{\big({\U}^{{\prime\prime}}_{1},{\U}^{\prime\prime}_{2}\big)}
=
&\sum_{\substack{p\\k_1,k_2,l_1,l_2}}\bigg[\langle e_p|\mathbf{ {M}}|e_{k_1}\rangle\textrm{ }  {U}^{\prime\prime}_{2} \widehat{cZ}_{1}\langle e_{k_1}|\mathbf{ {F}}_{1}|e_{k_2}\rangle  {U}^{\prime\prime}_{1}\langle e_{k_2}|\mathbf{ {R}}|e_0\rangle\bigg]\big({\rho}_{\textrm{in}}\big)\bigg[\langle e_0|\mathbf{ {R}}^\dagger|e_{l_2}\rangle\textrm{ } {U}^{\prime\prime\dagger}_{1} \langle e_{l_2}|\mathbf{ {F}}^\dagger_{1}|e_{l_1}\rangle \widehat{cZ}_{j} {U}^{\prime\prime\dagger}_{2} \langle e_{l_1}|\mathbf{ {M}^\dagger}|e_p\rangle\bigg]\textup{ ,}
\end{align}
since $\sum_k|e_k\rangle\langle e_k|=I_E$ and $\langle e_k|V_S|e_{k'}\rangle=V_S\delta_{k,k'}$ for every operator $V_S$ acting only on the system.
The operators $\langle e_p|\mathbf{ {M}}|e_{k_1}\rangle\textrm{, }\langle e_{k_1}|\mathbf{ {F}}_{1}|e_{k_2}\rangle \textrm{, }\langle e_{k_2}|\mathbf{ {R}}|e_0\rangle\textrm{, }\langle e_0|\mathbf{ {R}}^\dagger|e_{l_2}\rangle\textrm{, }\langle e_{l_2}|\mathbf{ {F}}^\dagger_{1}|e_{l_1}\rangle\textrm{, } \langle e_{l_1}|\mathbf{ {M}^\dagger}|e_p\rangle$ act only on the system, and can thus be written as in Table 2.
\begin{center}
\begin{eqnarray*}
\begin{tabular}{|l|l|l|}
\hline
&&\cr
$\langle e_{k_2}|\mathbf{ {R}}|e_{0}\rangle=\sum_{\mu_1}\eta^{(R)}_{k_2,0,\mu_1} {P}_{\mu_1}$
&$\langle e_{k_1}|\mathbf{ {F}}_{1}|e_{k_2}\rangle=\sum_{\mu_2}\eta^{(F)}_{k_1,k_2,\mu_2} {P}_{\mu_2}$&$\langle e_p|\mathbf{ {M}}|e_{k_1}\rangle=\sum_{\mu_3}\eta^{(M)}_{p,k_1,\mu_3} {P}_{\mu_3}$\cr
&&\cr
\hline
&&\cr
$\langle e_0|\mathbf{ {R}}^\dagger|e_{l_2}\rangle=\sum_{\nu_1}\eta^{*(R)}_{l_2,0,\nu_1} {P}_{\nu_1}$&
$\langle e_{l_2}|\mathbf{ {F}}_{1}^\dagger|e_{l_1}\rangle=\sum_{\nu_2}\eta^{*(F)}_{l_1,l_2,\nu_2} {P}_{\nu_2}$&$\langle e_{l_1}|\mathbf{ {M}}^\dagger|e_p\rangle=\sum_{\nu_3}\eta^{*(M)}_{p,l_1,\nu_3} {P}_{\nu_3}$\cr
&&\cr
\hline
\end{tabular}
\end{eqnarray*}
Table 2.
\vspace{0.3cm}
\end{center}
In Table 2, $ {P}_\mu\in\{I,X,Y,Z\}^{\otimes n}$ are $n$-fold tensor products of Pauli operators acting on the system and $\eta^{(R)}_{k_2,0,\mu_1}\textup{, }\eta^{(F)}_{k_1,k_2,\mu_2}\textup{, }\eta^{(M)}_{k_p,k_1,\mu_3}$ are complex numbers. We then obtain
\begin{align}
\rho{\big({\U}^{{\prime\prime}}_{1},{\U}^{\prime\prime}_{2}\big)}
=
&\sum_{\substack{p,k_1,k_2,l_1,l_2\\\mu_1,\mu_2,\mu_3\\\nu_1,\nu_2,\nu_3}}\bigg(\eta^{(R)}_{k_2,0,\mu_1}\eta^{(F)}_{k_1,k_2,\mu_2}\eta^{(M)}_{p,k_1,\mu_3}\eta^{*(R)}_{l_2,0,\nu_1}\eta^{*(F)}_{l_1,l_2,\nu_2}\eta^{*(M)}_{p,l_1,\nu_3}\bigg)\textrm{ }  {P}_{\mu_3} {U}^{\prime\prime}_{2} \widehat{cZ}_{1} {P}_{\mu_2}  {U}^{\prime\prime}_{1} {P}_{\mu_1}\big({\rho}_{\textrm{in}}\big)
 {P}_{\nu_1} {U}^{\prime\prime\dagger}_{1}  {P}_{\nu_2} \widehat{cZ}_{1} {U}^{\prime\prime\dagger}_{2} {P}_{\nu_3}\textup{ .}
\end{align}
We will now describe how to apply the Pauli twirl Lemmas iteratively, in the order the operations apply on the input. Therefore, we start by showing how to eliminate terms of the sum where $\mu_1 \neq \nu_1$. Since $X$ stabilizes $\ket{+}$ states, we can rewrite ${\rho}_{\textrm{in}}$ as $\big(\otimes_iX_i^{\gamma_{i}}\big){\rho}_{\textrm{in}}\big(\otimes_iX_i^{\gamma_{i}}\big)$. Moreover, using $ {U}^{\prime\prime}_{1}={U}^{\prime}_{1}\big(\otimes_iX_i^{\gamma_{i}}\big)$, cfr. Routine 1, the above state becomes
\begin{align}
\rho{\big({\U}^{{\prime\prime}}_{1},{\U}^{\prime\prime}_{2}\big)}
=&\sum_{\substack{p,k_1,k_2,l_1,l_2\\\mu_1,\mu_2,\mu_3\\\nu_1,\nu_2,\nu_3}}\bigg(\eta^{(R)}_{k_2,0,\mu_1}\eta^{(F)}_{k_1,k_2,\mu_2}\eta^{(M)}_{p,k_1,\mu_3}\eta^{*(R)}_{l_2,0,\nu_1}\eta^{*(F)}_{l_1,l_2,\nu_2}\eta^{*(M)}_{p,l_1,\nu_3}\bigg)\cr
&{P}_{\mu_3} {U}^{\prime\prime}_{2} \widehat{cZ}_{1} {P}_{\mu_2}  {U}^{\prime}_{1}\big(\otimes_iX_i^{\gamma_{i}}\big) {P}_{\mu_1}\big(\otimes_iX_i^{\gamma_{i}}\big)\big({\rho}_{\textrm{in}}\big)\big(\otimes_iX_i^{\gamma_{i}}\big)
 {P}_{\nu_1}\big(\otimes_iX_i^{\gamma_{i}}\big) {U}^{\prime\dagger}_{1}  {P}_{\nu_2}  \widehat{cZ}_{1} {U}^{\prime\prime\dagger}_{2} {P}_{\nu_3}\textup{ .}
\end{align}
Summing over all possible $\gamma_i$ and applying the {Restricted Pauli Twirl} (the Pauli-$X$ components of both $ {P}_{\mu_1}$ and $ {P}_{\nu_1}$ stabilize $\rho_{\textup{in}}$ and can thus be ignored), we obtain a factor $\delta_{\mu_1,\nu_1}$, and thus the above state becomes
\begin{align}
\rho{\big({\U}^{{\prime}}_{1},{\U}^{\prime\prime}_{2}\big)}&=
\frac{1}{2^n}\sum_{\{\gamma_i\}}\rho{\big({\U}^{{\prime\prime}}_{1},{\U}^{\prime\prime}_{2}\big)}\cr
&=\sum_{\substack{p,k_1,k_2,l_1,l_2\\\mu_1,\mu_2,\mu_3\\\ \nu_2,\nu_3}}\bigg(\eta^{(R)}_{k_2,0,\mu_1}\eta^{(F)}_{k_1,k_2,\mu_2}\eta^{(M)}_{p,k_1,\mu_3}\eta^{*(R)}_{l_2,0,\mu_1}\eta^{*(F)}_{l_1,l_2,\nu_2}\eta^{*(M)}_{p,l_1,\nu_3}\bigg)  {P}_{\mu_3} {U}^{\prime\prime}_{2} \widehat{cZ}_{1} {P}_{\mu_2}  {U}^{\prime}_{1} {P}_{\mu_1}\big({\rho}_{\textrm{in}}\big)
 {P}_{\mu_1} {U}^{\prime\dagger}_{1}  {P}_{\nu_2} \widehat{cZ}_{1} {U}^{\prime\prime\dagger}_{2} {P}_{\nu_3}\textup{ .}\cr
\end{align}
To operate a Pauli twirl on $ {P}_{\mu_2}$ and $ {P}_{\nu_2}$, we rewrite $ {U}^{\prime}_{1}$ as $\big(\otimes_iZ^{\alpha_{i,1}}X^{\alpha'_{1,i}}\big) {U}^{}_{1}$ and $ {U}^{\prime\prime}_{2} \widehat{cZ}_{1}$ as ${U}^{\prime}_{2} \widehat{cZ}_{1}\big(\otimes_iX^{\alpha'_{1,i}}Z^{\alpha_{i,1}}\big)$, cfr. Routine 1. Summing over $\alpha_{i,1}$ and $\alpha_{i,1}'$ and using the {Pauli Twirl}, we obtain $\delta_{\mu_2,\nu_2}$, and thus
\begin{align}
\rho{\big({\U}^{{}}_{1},{\U}^{\prime}_{2}\big)}=&\frac{1}{2^{2n}}\sum_{\{\alpha_{i,1}\},\{\alpha'_{i,1}\}}\rho{\big({\U}^{{\prime}}_{1},{\U}^{\prime\prime}_{2}\big)}\cr
=&\sum_{\substack{p,k_1,k_2,l_1,l_2\\\mu_1,\mu_2,\mu_3\\\ \nu_3}}\bigg(\eta^{(R)}_{k_2,0,\mu_1}\eta^{(F)}_{k_1,k_2,\mu_2}\eta^{(M)}_{p,k_1,\mu_3}\eta^{*(R)}_{l_2,0,\mu_1}\eta^{*(F)}_{l_1,l_2,\mu_2}\eta^{*(M)}_{p,l_1,\nu_3}\bigg)  {P}_{\mu_3} {U}^{\prime}_{2} \widehat{cZ}_{1} {P}_{\mu_2}  {U}^{}_{1} {P}_{\mu_1}\big({\rho}_{\textrm{in}}\big)
 {P}_{\mu_1} {U}^{\dagger}_{1}  {P}_{\mu_2} \widehat{cZ}_{1} {U}^{\prime\dagger}_{2} {P}_{\nu_3}\textup{ .}\cr
\end{align}
To operate a Pauli twirl on $ {P}_{\mu_3}$ and $ {P}_{\nu_3}$ we write the state of the system after the measurements:
\begin{align}
\rho_{\textup{meas}}\big({\U}^{{}}_{1},&{\U}^{\prime}_{2}\big)
=\frac{1}{2^{n}}\sum_{\{s_i\}}\otimes_i\bigg(\langle+|_iZ_i^{s_i}\textup{ }\rho{\big({\U}^{{}}_{1},{\U}^{\prime}_{2}\big)}\textup{ } Z_i^{s_i}|+\rangle_i\bigg)\textup{ }Z_{i}^{s_i}|+\rangle_i\langle+|Z_i^{s_i}\cr
=&\frac{1}{2^n}\sum_{\substack{p,k_1,k_2,l_1,l_2\\\mu_1,\mu_2,\mu_3,\nu_3\\ s_1,\ldots,s_n}}\bigg(\eta^{(R)}_{k_2,0,\mu_1}\eta^{(F)}_{k_1,k_2,\mu_2}\eta^{(M)}_{p,k_1,\mu_3}\eta^{*(R)}_{l_2,0,\mu_1}\eta^{*(F)}_{l_1,l_2,\mu_2}\eta^{*(M)}_{p,l_1,\nu_3}\bigg)\times\cr
&\langle+|^{\otimes n}\big(\otimes_iZ_i^{s_i}\big)  {P}_{\mu_3} {U}^{\prime}_{2} \widehat{cZ}_{1} {P}_{\mu_2}  {U}^{}_{1} {P}_{\mu_1}\big({\rho}_{\textrm{in}}\big)
 {P}_{\mu_1} {U}^{\dagger}_{1}  {P}_{\mu_2}\widehat{cZ}_{1} {U}^{\prime\dagger}_{2} {P}_{\nu_3}\big(\otimes_iZ_i^{s_i}\big)|+\rangle^{\otimes n}\times\bigg(\otimes_iZ_{i}^{s_i}|+\rangle_i\langle+|Z_i^{s_i}\bigg)\nonumber\\\cr
=&\frac{1}{2^n}\sum_{\substack{p,k_1,k_2,l_1,l_2\\\mu_1,\mu_2,\mu_3,\nu_3\\ s_1,\ldots,s_n}}\bigg(\eta^{(R)}_{k_2,0,\mu_1}\eta^{(F)}_{k_1,k_2,\mu_2}\eta^{(M)}_{p,k_1,\mu_3}\eta^{*(R)}_{l_2,0,\mu_1}\eta^{*(F)}_{l_1,l_2,\mu_2}\eta^{*(M)}_{p,l_1,\nu_3}\bigg)\cr
&\times\langle+|^{\otimes n}\big(\otimes_iZ_i^{s_i}\big)  {P}_{\mu_3}\big(\otimes_iX_i^{\alpha'_{i,2}}Z_i^{\alpha_{i,2}}\big) {U}^{}_{2}\widehat{cZ}_{1} {P}_{\mu_2}  {U}^{}_{1} {P}_{\mu_1}\big({\rho}_{\textrm{in}}\big)
 {P}_{\mu_1} {U}^{\dagger}_{1}  {P}_{\mu_2}\widehat{cZ}_{1} {U}^{\dagger}_{2}\big(\otimes_iX_i^{\alpha'_{i,2}}Z_i^{\alpha_{i,2}}\big) {P}_{\nu_3}\big(\otimes_iZ_i^{s_i}\big)|+\rangle^{\otimes n}\nonumber\\\cr
&\bigg(\otimes_iZ_{i}^{s_i}|+\rangle_i\langle+|Z_i^{s_i}\bigg)\textup{ ,}
\end{align}
where in the second equality we used $ {U}^{\prime}_{2}=\big(\otimes_iX_i^{\alpha'_{i,2}}Z_i^{\alpha_{i,2}}\big) {U}^{}_{2}$ . We can now rewrite $\ket{+}^{\otimes n}$ as $\otimes_iX_i^{\alpha_{i,2}^{\prime\prime}}\ket{+}^{\otimes n}$ and sum over $\{\alpha_{i,2}'\}$. Using the {Restricted Pauli Twirl} (the Pauli-$X$ components of both $ {P}_{\mu_3}$ and $ {P}_{\nu_3}$ stabilize $|+\rangle^{\otimes n}$ and can thus be ignored), we obtain $\delta_{\mu_3,\nu_3}$: 
\begin{align}
\rho_{\textup{meas}}\big({\U}^{{}}_{1},{\U}^{}_{2}\big)=&\frac{1}{2^n}\sum_{\alpha_{i,2}'}\rho_{\textup{meas}}\big({\U}^{{}}_{1},{\U}^{\prime}_{2}\big)\cr
&
=\frac{1}{2^n}\sum_{\substack{p,k_1,k_2,l_1,l_2\\\mu_1,\mu_2,\mu_3,\nu_3\\ s_1,\ldots,s_n}}\bigg(\eta^{(R)}_{k_2,0,\mu_1}\eta^{(F)}_{k_1,k_2,\mu_2}\eta^{(M)}_{p,k_1,\mu_3}\eta^{*(R)}_{l_2,0,\mu_1}\eta^{*(F)}_{l_1,l_2,\mu_2}\eta^{*(M)}_{p,l_1,\mu_3}\bigg)\cr
&\times\langle+|^{\otimes n}\big(\otimes_iZ_i^{s_i}\big)  {P}_{\mu_3}\big(\otimes_iZ_i^{\alpha_{i,2}}\big) {U}^{}_{2}\widehat{cZ}_{1} {P}_{\mu_2}  {U}^{}_{1} {P}_{\mu_1}\big({\rho}_{\textrm{in}}\big)
 {P}_{\mu_1} {U}^{\dagger}_{1}  {P}_{\mu_2}\widehat{cZ}_{1} {U}^{\dagger}_{2}\big(\otimes_iZ_i^{\alpha_{i,2}}\big) {P}_{\mu_3}\big(\otimes_iZ_i^{s_i}\big)|+\rangle^{\otimes n}\nonumber\\\cr
&\times\otimes_iZ_{i}^{s_i}|+\rangle_i\langle+|Z_i^{s_i}
\end{align}
Finally, after the classical post-processing (which replaces the outputs $s_i$ with $s_i\oplus\alpha_{i,2}$), average over $\{\alpha_{i,2}\}$ yields the outcome state
\begin{align}~\label{eq:final}
\rho_{\textup{out}}\big({\U}^{{}}_{1},&\textup{ }{\U}^{}_{2}\big)=\frac{1}{2^n}\sum_{\{\alpha_{i,2}\}}\frac{1}{2^n}\sum_{\substack{p,k_1,k_2,l_1,l_2\\\mu_1,\mu_2,\mu_3\\ s_1,\ldots,s_n}}\bigg(\eta^{(R)}_{k_2,0,\mu_1}\eta^{(F)}_{k_1,k_2,\mu_2}\eta^{(M)}_{p,k_1,\mu_3}\eta^{*(R)}_{l_2,0,\mu_1}\eta^{*(F)}_{l_1,l_2,\mu_2}\eta^{*(M)}_{p,l_1,\mu_3}\bigg)\cr
&\times\langle+|^{\otimes n}\big(\otimes_iZ_i^{s_i}\big)  {P}_{\mu_3}
\big(\otimes_iZ_i^{\alpha_{i,2}}\big) {U}^{}_{2}\widehat{cZ}_{1} {P}_{\mu_2}  {U}^{}_{1} {P}_{\mu_1}\big({\rho}_{\textrm{in}}\big)
 {P}_{\mu_1} {U}^{\dagger}_{1}  {P}_{\mu_2}\widehat{cZ}_{1} {U}^{\dagger}_{2}\big(\otimes_iZ_i^{\alpha_{i,2}}\big) {P}_{\mu_3}\big(\otimes_iZ_i^{s_i}\big)|+\rangle^{\otimes n}\nonumber\\\cr
&\times\otimes_iZ_{i}^{s_i\oplus\alpha_{i,2}}|+\rangle_i\langle+|Z_i^{s_i\oplus\alpha_{i,2}}\cr
&=\frac{1}{2^n}\sum_{\{\alpha_{i,2}\}}\frac{1}{2^n}\sum_{\substack{p,k_1,k_2,l_1,l_2\\\mu_1,\mu_2,\mu_3\\ s_1,\ldots,s_n}}\bigg(\eta^{(R)}_{k_2,0,\mu_1}\eta^{(F)}_{k_1,k_2,\mu_2}\eta^{(M)}_{p,k_1,\mu_3}\eta^{*(R)}_{l_2,0,\mu_1}\eta^{*(F)}_{l_1,l_2,\mu_2}\eta^{*(M)}_{p,l_1,\mu_3}\bigg)\cr
&\times\langle+|^{\otimes n}\big(\otimes_iZ_i^{s_i\oplus\alpha_{i,2}} \big) {P}_{\mu_3}\big(\otimes_iZ_i^{\alpha_{i,2}}\big) {U}^{}_{2}\widehat{cZ}_{1} {P}_{\mu_2}  {U}^{}_{1} {P}_{\mu_1}\big({\rho}_{\textrm{in}}\big)
 {P}_{\mu_1} {U}^{\dagger}_{1}  {P}_{\mu_2}\widehat{cZ}_{1} {U}^{\dagger}_{2}\big(\otimes_iZ_i^{\alpha_{i,2}}\big) {P}_{\mu_3}\big(\otimes_iZ_i^{s_i\oplus\alpha_{i,2}}\big)|+\rangle^{\otimes n}\nonumber\\\cr
&\times\otimes_iZ_{i}^{s_i}|+\rangle_i\langle+|Z_i^{s_i}\cr
&=\frac{1}{2^n}\sum_{\substack{\mu_1,\mu_2,\mu_3\\ s_1,\ldots,s_n}}\bigg(\sum_{\substack{p,k_1,k_2,l_1,l_2}}\eta^{(R)}_{k_2,0,\mu_1}\eta^{(F)}_{k_1,k_2,\mu_2}\eta^{(M)}_{p,k_1,\mu_3}\eta^{*(R)}_{l_2,0,\mu_1}\eta^{*(F)}_{l_1,l_2,\mu_2}\eta^{*(M)}_{p,l_1,\mu_3}\bigg)\cr
&\times\langle+|^{\otimes n}\big(\otimes_iZ_i^{s_i}\big)  {P}_{\mu_3} {U}^{}_{2}\widehat{cZ}_{1} {P}_{\mu_2}  {U}^{}_{1} {P}_{\mu_1}\big({\rho}_{\textrm{in}}\big)
 {P}_{\mu_1} {U}^{\dagger}_{1}  {P}_{\mu_2}\widehat{cZ}_{1} {U}^{\dagger}_{2} {P}_{\mu_3}\big(\otimes_iZ_i^{s_i}\big)|+\rangle^{\otimes n}\nonumber\\\cr
&\times\otimes_iZ_{i}^{s_i}|+\rangle_i\langle+|Z_i^{s_i}\cr
&=\frac{1}{2^n}\sum_{\substack{\mu_1,\mu_2,\mu_3\\ s_1,\ldots,s_n}}\phi_{\mu_1,\mu_2,\mu_3}\textup{ }\langle+|^{\otimes n}\big(\otimes_iZ_i^{s_i}\big)  {P}_{\mu_3} {U}^{}_{2}\widehat{cZ}_{1} {P}_{\mu_2}  {U}^{}_{1} {P}_{\mu_1}\big({\rho}_{\textrm{in}}\big)
 {P}_{\mu_1} {U}^{\dagger}_{1}  {P}_{\mu_2}\widehat{cZ}_{1} {U}^{\dagger}_{2} {P}_{\mu_3}\big(\otimes_iZ_i^{s_i}\big)|+\rangle^{\otimes n}\nonumber\\\cr
&\times\otimes_iZ_{i}^{s_i}|+\rangle_i\langle+|Z_i^{s_i}\textup{ ,}
\end{align}
where 
\begin{align}
\phi_{\mu_1,\mu_2,\mu_3}=\sum_{\substack{p,k_1,k_2,l_1,l_2}}\eta^{(R)}_{k_2,0,\mu_1}\eta^{(F)}_{k_1,k_2,\mu_2}\eta^{(M)}_{p,k_1,\mu_3}\eta^{*(R)}_{l_2,0,\mu_1}\eta^{*(F)}_{l_1,l_2,\mu_2}\eta^{*(M)}_{p,l_1,\mu_3}=\sum_p\bigg|\sum_{\substack{k_1,k_2}}\eta^{(R)}_{k_2,0,\mu_1}\eta^{(F)}_{k_1,k_2,\mu_2}\eta^{(M)}_{p,k_1,\mu_3}\bigg|^2\geq0\textup{}
\end{align}
and $\sum_{\mu_1,\mu_2,\mu_3}\phi_{\mu_1,\mu_2,\mu_3}=1$. $\rho_{\textup{out}}\big({\U}_{1},\textup{ }{\U}^{}_{2}\big)$ is therefore a convex combination of quantum states and $\phi_{\mu_1,\mu_2,\mu_3}$ can be seen as the joint probability of Pauli errors ${P}_{\mu_1},{P}_{\mu_2}$ and ${P}_{\mu_3}$. We can thus rewrite
\begin{align}
{\rho}_{\textup{out}}\big({\U}_{1}&,{\U}_{2}\big)
=\sum_{\substack{{\p}_{0},{\p}_{1},{\p}_{2}\\s_1,\ldots,s_n}}\frac{
\textup{prob}\big({\p}_{0},{\p}_{1},{\p}_{2}\big)}{2^{n}}\textup{ }\textup{ }\langle+|^{\otimes n}\bigg[\Z^{\overline{s}}{\p}_{2}\textup{ }{\U}_{2}\textup{ }{\C\Z}_{1}\textup{ }{\p}_{1}\textup{ }{\U}_{1}\textup{ }{\p}_{0}\big(\rho_{\textup{in}}\big)\bigg]|+\rangle^{\otimes n}\bigg(\otimes_{i} Z_i^{s_i}|+\rangle_{i}\langle+|Z_i^{s_i}\bigg)\textup{ ,}
\end{align}
where $\p_0,\p_2\in\{\I,\Z\}^{\otimes n}$, $\p_1\in\{\I,,\X,\Y,\Z\}^{\otimes n}$ and $\textup{prob}\big({\p}_{0},{\p}_{1},{\p}_{2}\big)$ is the joint probability of Pauli errors ${\p}_{0},{\p}_{1},{\p}_{2}$. This concludes the proof for the protocol with $v=0$ and $m=2$.

The generalization to a protocol with $v=0$ and $m>2$ is straightforward. Starting from the state in Equation \ref{eq:rhomultiplebands}, one can use the same arguments as for the two-band circuit. To generalize to multiple circuits ($v>0$), we start by noticing that the circuits are implemented in series, hence the noise can only affect one circuit at a time. By the principle of deferred measurements, we {can execute all} the measurements at the end of the protocol. Moreover, {we can prepare the input qubits} for all the circuits at the beginning of the protocol. Doing this, the state of the system after all circuits have been implemented becomes
\begin{small}
\begin{align}
&\rho{\big({\U}^{\prime\prime(1)}_{1},\ldots,{\U}^{\prime\prime(v+1)}_{m}\big)}
=\cr
\textup{Tr}_{\textup{E}}\bigg[&\bigg(\mathbf{ {M}}^{(v+1)}\textrm{ } {U}^{\prime\prime(v+1)}_{m}\widehat{cZ}_{m-1}\mathbf{ {F}}^{(v+1)}_{m-1}  {U}^{\prime\prime(v+1)}_{m-1}\ldots\widehat{cZ}_{1}\mathbf{ {F}}_{v+1} ^{(1)} {U}^{\prime\prime(v+1)}_{1}\mathbf{ {R}}^{(v+1)}\bigg)\ldots\bigg(\mathbf{ {M}}^{(1)}\textrm{ } {U}^{\prime\prime(1)}_{m}\widehat{cZ}_{m-1}\mathbf{ {F}}^{(1)}_{m-1}  {U}^{\prime\prime(1)}_{m-1}\ldots\widehat{cZ}_{1}\mathbf{ {F}}_{1} ^{(1)} {U}^{\prime\prime(1)}_{1}\mathbf{ {R}}^{(1)}\bigg)\cr
&\bigg({\rho}_{\textrm{in}}\otimes\rho_E\bigg)\bigg(\mathbf{ {R}}^{\dagger(1)} {U}^{\prime\prime\dagger(1)}_{1} \mathbf{ {F}}_{1}^{\dagger(1)} \widehat{cZ}_{1}\ldots {U}^{\prime\prime\dagger(1)}_{m-1} \mathbf{ {F}}_{m-1}^{\dagger(1)} \widehat{cZ}_{m-1} {U}^{\prime\prime\dagger(1)}_{m}\mathbf{ {M}}^{\dagger(1)}\bigg)\ldots\bigg(\mathbf{ {R}}^{\dagger(v+1)} {U}^{\prime\prime\dagger(v+1)}_{1} \mathbf{ {F}}_{1}^{\dagger(v+1)} \widehat{cZ}_{1}\ldots {U}^{\prime\prime\dagger(v+1)}_{m-1} \mathbf{ {F}}_{m-1}^{\dagger(v+1)}\cr
&\widehat{cZ}_{m-1} {U}^{\prime\prime\dagger(v+1)}_{m}\mathbf{ {M}}^{\dagger(v+1)}\bigg)\bigg]\cr
\end{align}
\end{small}
where $\rho_{\textup{in}}=\otimes_{k=1}^{v+1}\otimes_{i=1}^n|+\rangle^{(k)}_i\langle+|$, $\mathbf{ {R}}^{(k)},\mathbf{ {F}}^{(k)}_{j},\mathbf{ {M}}^{(k)}$ are unitary matrices that act only on the qubits in the $k$-th circuit and on the environment (which is the same for all the circuits) and ${U}^{\prime\prime(k)}_{j}=\otimes_i U_{i,j}^{\prime\prime(k)}$. Starting from here and using the same arguments as above, one can finally obtain Equation \ref{eq:rhoout}.
\end{proof}

\section{Statement and Proof of Lemma 2}~\label{subapp:lem2}
Lemma \ref{lem:pauliz} shows that in our accreditation protocol the noise of the form {N1} can be reduced to classically correlated collections of Pauli errors affecting the circuits. In this Appendix we prove Lemma \ref{lem:ct}, which states that all collections of Pauli errors can be detected with probability larger than 1/4. More formally:
\begin{lemma}~\label{lem:ct}
For any collection of Pauli errors affecting a trap circuit, summing over all possible single-qubit gates in the trap circuit (i.e. over all possible sets $\{U_{i,j}\}$ output by Routine 2), the probability that the trap circuit outputs $\overline{s}=\overline{0}$ is at most 3/4.  
\end{lemma}
\begin{proof}
For a given collection of Pauli errors $\{\p_{j}\}_{j=0}^m$ affecting a trap circuit, the state of the trap circuit after the measurements is of the form
\begin{equation}
\rho_{\textup{out}}^{\textup{trap}}\big(\{\p_{j}\}\big)=\frac{1}{M_1\times\cdots\times M_{m-1}}\sum_{\substack{\U_{1},\cdots,\U_{m}}}\p_{m}\U_m\circ\bigg(\circ_{j=1}^{m-1}
\C\Z_{j}\p_{j}\U_{j}\bigg)\circ
\p_{0}\big(\rho_{\textup{in}}^{\textup{trap}}\big)\textup{ ,}
\end{equation}
where $\rho_{\textup{in}}^{\textup{trap}}=\otimes_{i=1}^n|+\rangle_i\langle+|$, $\C\Z_{j}$ is the entangling operation in band $j$, $\p_{0},\p_{m}\in\{\I,\Z\}^{\otimes n}$, $\p_{j}\in\{\I,\X,\Y,\Z\}^{\otimes n}$ for all $j=1,\dots,m$ and $M_j$ is the number of choices of $\U_{j}$. Note that each number $M_j$ depends on the number of qubits connected by a $cZ$ in band $j$ of the trap circuit, cfr. Routine 2.

In a trap circuit the gate $\U_{1}$ in the first band is of the form $\U_{1}={\V}_{1}\h^t$, where ${\V}_{1}$ implements a gate from $\{H,S\}^{\otimes n}$ (cfr. step 2.1 of Routine 2) and $\h^t$ is the round of Hadamard gates activated at random (cfr. step 3 of Routine 2). Similarly, for all $j=2,\ldots,m-1$, $\U_{j}$ implements a gate belonging to the set $\{I,HS^{\dagger},SH\}^{\otimes n}$. These gates undo the gates in previous band and implement new ones (cfr. step 2.1 Routine 2 and Figure \ref{fig:trap}), thus we can write them as $\U_{j}=\V_{j}\V^{-1}_{j-1}$|with each $\V_{j}$ implementing a gate from the set $\{H,S\}^{\otimes n}$. Finally, the gate $\U_{m}$ in the last band is of the form $\U_{m}={\h}^t{\V}_{m-1}^{-1}$, where ${\V}_{m-1}^{-1}$ implements a gate from $\{H,S^\dagger\}^{\otimes n}$ and undoes the gate in band $m-1$ (cfr. step 2.2 of Routine 2). Using this, we obtain
\begin{equation}
\rho_{\textup{out}}^{\textup{trap}}\big(\{\p_{j}\}\big)=\frac{1}{2}\sum_{\substack{t=0,1}}\frac{1}{N_1\times\cdots\times N_{m-1}}\sum_{\substack{\V_{1},\cdots,\V_{m-1}}}\p_{m}\h^t\circ\bigg(\circ_{j=1}^{m-1}
\V_{j}^{{-1}}\C\Z_{j}\p_{j}\V_{j}\bigg)
\circ\h^t\p_{0}\big(\rho_{\textup{in}}^{\textup{trap}}\big)\textup{ ,}
\end{equation}
where $N_j$ is the number of possible choices of $\V_{j}$.

Using that $\V_{j}^{-1}\C\Z_{j}\V_{j}=\C\X_{j}$ is a tensor product of $cX$ gates, the above state can also be rewritten as
\begin{equation}
\rho_{\textup{out}}^{\textup{trap}}\big(\{\p_{j}\}\big)=\frac{1}{2}\sum_{\substack{t=0,1}}\frac{1}{N_1\times\cdots\times N_{m-1}}\sum_{\substack{\V_{1},\cdots,\V_{m-1}}}\p_{m}\h^t\circ\bigg(\circ_{j=1}^{m-1}
\C\X_{j}\V_{j}^{{-1}}\p_{j}\V_{j}\bigg)
\circ\h^t\p_{0}\big(\rho_{\textup{in}}^{\textup{trap}}\big)\textup{ .}
\end{equation}
Notice that each $\C\X_j$ carries an implicit dependency on $\V_j$ (the orientation of the $cX$ gates depends on $\V_j$, cfr. Figure \ref{fig:trap}). 

The probability that the trap outputs $\overline{s}=\overline{0}$ is
\begin{equation}
\textup{prob}\big(\overline{s}=\overline{0}\textrm{ $|$ }\{\p_{j}\}\big)={}\langle+|^{\otimes n}\rho_{\textup{out}}^{\textup{trap}}\big(\{\p_{j}\}\big)|+\rangle^{\otimes n}\textup{ .}
\end{equation}
To upper-bound this probability by 3/4, we first consider ``single-band'' collections of errors, namely collections $\{\p_{j}\}$ such that $\p_{j_0}\neq\I$ for some $j_0\in\{0,\ldots,m\}$ and $\p_{j}=\I$ for all other $j\neq j_0$. For these collections, we prove that the probability that the output of the trap is the correct one $\overline{s}=\overline{0}$ is smaller than 1/2:
\begin{equation}~\label{eq:h-1}
\textup{prob}(\overline{s}=\overline{0}\textup{ $|$  single-band collection})\leq \frac{1}{2}
\end{equation}
We prove this in {Statement 1}. 

Next, we consider ``two-band'' collections of errors. We obtain
\begin{equation}~\label{eq:h-2}
\textup{prob}(\overline{s}=\overline{0}\textup{ $|$  two-band collection})\leq\frac{3}{4}
\end{equation}   
We prove this in {Statement 2}. To obtain this bound, we \textit{move} the two errors towards each other (i.e. we commute them with all the gates in the middle) and subsequently \textit{merge} them, rewriting them as a single Pauli operator. The resulting Pauli operator is the identity $\I$ with probability $c$, or is a different operator with probability $1-c$. In the former case, the errors have canceled out with each other, while in the latter they have reduced to a single-band error. Importantly, in {Statement 2} we prove that $c\leq1/2$. This yields
\begin{align}
\textup{prob}(\overline{s}=\overline{0}\textup{ $|$  two-band collection})&=(1-c)\textup{prob}(\overline{s}=\overline{0}\textup{ $|$  single-band collection})+c\textup{ }\textup{prob}(\overline{s}=\overline{0}\textup{ $|$  no error})\cr
&\leq \frac{1-c}{2}+c=\frac{1}{2}+\frac{c}{2}\textup{ ,}
\end{align}
where we used $\textup{prob}(\overline{s}=\overline{0}\textup{ $|$  no error})=1$ and $\textup{prob}(\overline{s}=\overline{0}\textup{ $|$  single-band collection})\leq1/2$. Maximizing over $c\in[0,1/2]$, we find
\begin{equation}
\textup{prob}(\overline{s}=\overline{0}\textup{ $|$  two-band collection})\leq\max_{0\leq c\leq\frac{1}{2}}\bigg(\frac{1}{2}+\frac{c}{2}\bigg)=\frac{3}{4}
\end{equation}
Finally, we generalise to collections affecting more than two bands. For three-band collections, again we move two of these Pauli operators towards each other and merge them. Doing this, the three-band collection reduces to a single-band one with probability $c\leq1/2$ or to a two-band one with probability $1-c$. Thus, using the above results, we have
\begin{align}
\textup{prob}(\overline{s}=\overline{0}\textup{ $|$  three-band collection})&=(1-c)\textup{prob}(\overline{s}=\overline{0}\textup{ $|$  two-band collection})+c\textup{ }\textup{prob}(\overline{s}=\overline{0}\textup{ $|$  single-band collection})\cr
&\leq \frac{3(1-c)}{4}+\frac{c}{2}\leq\max_{0\leq c\leq\frac{1}{2}}\bigg(\frac{3(1-c)}{4}+\frac{c}{2}\bigg)=\frac{3}{4}
\end{align}
This argument can be iterated: at any fixed $h$, if $\textup{prob}(\overline{s}=\overline{0}\textup{ $|$  ($h-$2)-band collection})\leq3/4$ and $\textup{prob}(\overline{s}=\overline{0}\textup{ $|$  ($h-$1)-band collection})\leq3/4$, then it can be easily shown that $\textup{prob}(\overline{s}=\overline{0}\textup{ $|$  $h$-band collection})\leq3/4$.
We now complete the proof by proving {Statement 1} and {Statement 2}. \\

\noindent\textbf{Statement 1.} Single-band collections are defined as follows:
\begin{equation}
\p_{j}\neq \I\textup{ for }j=j_0\in\{0,\cdots,m\}\textup{ , }\p_{j}= \I\textup{ for all }j\neq j_0.
\end{equation}
If $j_0=0$, using $cX|++\rangle=|++\rangle$ and $cX|00\rangle=|00\rangle$, we have
\begin{align}
\textup{prob}\big(\overline{s}=\overline{0}\textrm{ $|$ }\{\p_{j}\}\big)=&\frac{1}{2}\sum_{t=0,1}{}\langle+|^{\otimes n}\h^t\circ\bigg(\circ_{j=1}^m\C\X_{j}\bigg)\circ \h^t\p_{0}|+\rangle^{\otimes n}\cr
=&\langle+|^{\otimes n}\p_{0}|+\rangle^{\otimes n}=0\textup{ ,}
\end{align}
since $\p_{0}\neq\I\in\{\I,\Z\}^{\otimes n}$, and the same happens for $\p_{m}\neq\I$. 

If $1\leq j_0\leq m-1$, we have
\begin{equation}
\textup{prob}\big(\overline{s}=\overline{0}\textrm{ $|$ }\{\p_{j}\}\big)=\frac{1}{N_{j_0}}\sum_{\V_{j_0}}\frac{1}{2}\sum_{t=0,1}{}\langle+|^{\otimes n}\h^t\V_{j_0}^{-1}\p_{j_0}\V_{j_0}\h^t|+\rangle^{\otimes n}\textup{ ,}
\end{equation}
where we used again that $cX|++\rangle=|++\rangle$ and $cX|00\rangle=|00\rangle$. Notice that $\langle+|^{\otimes n}\p|+\rangle^{\otimes n}=0$ for all Pauli operators $P$ whose Pauli-$Z$ component is non-trivial, therefore $\sum_t\langle+|^{\otimes n}\h^t\p\h^t|+\rangle^{\otimes n}/2\leq1/2$ for all $\p\in\{\I,\X,\Y,\Z\}^{\otimes n}/\I$. This yields
\begin{equation}
\textup{prob}\big(\overline{s}=\overline{0}\textrm{ $|$ }\{\p_{j}\}\big)\leq\frac{1}{N_{j_0}}\sum_{\V_{j_0}}\frac{1}{2}\leq\frac{1}{2}\textrm{ ,}
\end{equation}
where we used that $\V_{j_0}^{-1}\p_{j_0}\V_{j_0}$ is a Pauli operator for any $\V_{j_0}$, as this $\V_{j_0}$ implements a gate from the set $\{H,S\}^{\otimes n}$.\\

\noindent\textbf{Statement 2.} Two-band collections are defined as follows:
\begin{equation}
\p_{j}\neq \I\textup{ for }j=j_1,j_2\in\{0,\cdots,m\} \textup{ (with }j_1< j_2\textup{)}\textup{ , }\p_{j}= \I\textup{ for all }j\neq j_1,j_2.
\end{equation}
We can distinguish four classes of two-band collections:\\

\indent 1) Errors in state preparation and entangling gates, i.e. $j_1=0$ and $1\leq j_2\leq m-1$.

\indent 2) Errors in entangling gates and measurements, i.e. $1\leq j_1\leq m-1$ and $j_2=m$.

\indent 3) Errors in two different entangling gates, i.e. $1\leq j_1<j_2\leq m-1$.
 
\indent 4) Errors in state preparation and measurements, i.e. $j_1=0$ and $j_2=m$.\\

\noindent Errors in class 1 yield $\overline{s}=\overline{0}$ with probability at most 3/4. To prove this, we start by rewriting this probability as
\begin{align}
\textup{prob}\big(\overline{s}=\overline{0}\textrm{ $|$ }\{{P}_{j}\}\big)=&\frac{1}{2}\sum_{\substack{t=0,1}}\frac{1}{N_1\times\cdots\times N_{m-1}}\sum_{\substack{\V_{1},\cdots,\V_{m-1}}}\langle+|^{\otimes n}\h^t\circ\bigg(\circ_{j=j_2}^{m-1}
\C\X_{j}\bigg)\circ\V_{j_2}^{{-1}}\p_{j_2}
\V_{j_2}\cr
&\circ\bigg(\circ_{j=1}^{j_2-1}
\C\X_{j}\bigg)
\circ\h^t\p_{0}\big(\rho_{\textup{in}}^{\textup{trap}}
\big)|+\rangle^{\otimes n}\textup{}\cr
=&\frac{1}{2}\sum_{\substack{t=0,1}}\frac{1}{N_1\times\cdots\times N_{m-1}}\sum_{\substack{\V_{1},\cdots,\V_{j_2}}}\langle+|^{\otimes n}\h^t\V_{j_2}^{{-1}}\p_{j_2}
\V_{j_2}\circ\bigg(\circ_{j=1}^{j_2-1}
\C\X_{j}\bigg)
\circ\h^t\p_{0}\big(\rho_{\textup{in}}^{\textup{trap}}
\big)|+\rangle^{\otimes n}\textup{ ,}\cr
\end{align}
where we used $cX|++\rangle=|++\rangle$ and $cX|00\rangle=|00\rangle$. We now start from the case $j_2=1$. We then note that (i) if $n=1$ (single-qubit circuit), all $\p_1\in\{\X,\Y,\Z\}$ and all $\p_0=\Z$ lead to
\begin{align}
\label{ineq:multi1}
\textup{prob}\big(\overline{s}=\overline{0}\textrm{ $|$ }\{{P}_{j}\}\textup{, 1 qubit}\big)=&\frac{1}{2}\sum_{\substack{t=0,1}}\frac{1}{N_1}\sum_{\substack{\V_{1}\in\{\h,\s\}}}\langle+|^{}\h^t\V_1^{-1}\p_1\V_1
\h^t\p_{0}\big(\rho_{\textup{in}}^{\textup{trap}}\big)|+\rangle^{}\leq\frac{3}{4}
\textup{ ,}
\end{align}
and (ii) if $n=2$ (two-qubit circuit) and in band $1$ the two qubits are connected by $cZ$, all $\p_1\neq\I\in\{\I,\X,\Y,\Z\}^{\otimes 2}$ and all $\p_0\neq\I\in\{\I,\Z\}^{\otimes 2}$ lead to
\begin{align}
\label{ineq:multi2}
\textup{prob}\big(\overline{s}=\overline{0}\textrm{ $|$ }\{{P}_{j}\}\textup{, 2 qubits}\big)=&\frac{1}{2}\sum_{\substack{t=0,1}}\frac{1}{N_1}\sum_{\substack{\V_{1}\in\{\h\otimes\s,\s\otimes\h\}}}\langle+|^{\otimes 2}\h^t\V_1^{-1}\p_1\V_1
\h^t\p_{0}\big(\rho_{\textup{in}}^{\textup{trap}}\big)|+\rangle^{\otimes 2}\leq\frac{3}{4}
\textup{ .}
\end{align}
The above inequalities for $n=1$ and $n=2$ can be proven using that $H$ maps $\{X,Y,Z\}$ into $\{Z,Y,X\}$ under conjugation and $S$ maps $\{X,Y,Z\}$ into $\{Y,X,Z\}$ under conjugation (apart from unimportant global phases). Extension to more than two qubits is as follows: If $\textup{prob}\big(\overline{s}=\overline{0}\textrm{ $|$ }\{{P}_{j}\}\textup{ , $n_0$ qubits}\big)\leq3/4$, then tensoring one more qubit yields
\begin{align}
\textup{prob}\big(\overline{s}=\overline{0}&\textrm{ $|$ }\{{P}_{j}\}\textup{, $n_0+1$ qubits}\big)=\cr
&\frac{1}{2}\sum_{\substack{t=0,1}}\frac{1}{N_1}\sum_{\substack{\V_{1}}}\bigg(\langle+|^{\otimes n_0}\h^t_{1,\ldots,n_0}\V_{1|1,\ldots,n_0}^{-1}\p_{1|1,\ldots,n_0}\V_{1|1,\ldots,n_0}
\h^t_{1,\ldots,n_0}\p_{0|1,\ldots,n_0}\big(\rho_{\textup{in}}^{\textup{trap}}\big)|+\rangle^{\otimes n_0}\times\cr
&\langle+|^{}\h^t_{n_0+1}\V_{1|n_0+1}^{-1}\p_{1|n_0+1}\V_{1|n_0+1}
\h^t_{n_0+1}\p_{0|n_0+1}\big(\rho_{\textup{in}}^{\textup{trap}}\big)|+\rangle^{}\bigg)\textup{ ,}
\end{align}
where $\h^t_{1,\ldots,n_0}, \V_{1|1,\ldots,n_0},\p_{1|1,\ldots,n_0}$ and $\p_{0|1,\ldots,n_0}$ are the components of $\h^t, \V_{1},\p_{1}$ and $\p_{0}$ acting on qubits $\{1,\ldots,n_0\}$ and $\h^t_{n_0+1}, \V_{1|n_0+1},\p_{1|n_0+1}$ and $\p_{0|n_0+1}$ the components acting on qubit $n_0+1$. Using that if $A_h,B_h\geq0\textup{ $\forall$ }h$, then $\sum_{h}A_hB_h\leq\sum_{h}A_h\sum_{h}B_h$, we obtain
\begin{align}
\textup{prob}\big(\overline{s}=\overline{0}&\textrm{ $|$ }\{{P}_{j}\}\textup{, $n_0+1$ qubits}\big)\leq\cr
&\frac{1}{2}\sum_{\substack{t=0,1}}\frac{1}{N_1}\sum_{\substack{\V_{1}}}\bigg(\langle+|^{\otimes n_0}\h^t_{1,\ldots,n_0}\V_{1|1,\ldots,n_0}^{-1}\p_{1|1,\ldots,n_0}\V_{1|1,\ldots,n_0}
\h^t_{1,\ldots,n_0}\p_{0|1,\ldots,n_0}\big(\rho_{\textup{in}}^{\textup{trap}}\big)|+\rangle^{\otimes n_0}\bigg)\times\cr
&\frac{1}{2}\sum_{\substack{t=0,1}}\frac{1}{N_1}\sum_{\substack{\V_{1}}}\bigg(\langle+|^{}\h^t_{n_0+1}\V_{1|n_0+1}^{-1}\p_{1|n_0+1}\V_{1|n_0+1}
\h^t_{n_0+1}\p_{0|n_0+1}\big(\rho_{\textup{in}}^{\textup{trap}}\big)|+\rangle^{}\bigg)\cr
&\leq\frac{3}{4}\times\frac{3}{4}\leq\frac{3}{4}\textup{ ,}
\end{align}
Tensoring two qubits connected by $cZ$ yields the same bound, and this concludes the proof by induction for $j_2=1$. If $j_2\in\{1,\ldots,m-1\}$ the proof is similar, but the Pauli operator $\V_{j_2}^{-1}\p_{j_2}\V_{j_2}$ must be commuted with $\C\X_1,\ldots,\C\X_{j_2-1}$ (where we remember that each $\C\X_j$ depends on $\V_j$). At fixed $\V_1,\ldots,\V_{j_2-1}$ it can be shown (with the same arguments as used for $j_2=1$, i.e. considering first the cases $n=1$ and $n=2$ and then generalizing to $n>2$) that summations over $\V_{j_2}$ and $t$ yield an upper-bound of 3/4. The upper-bound on $\textup{prob}\big(\overline{s}=\overline{0}\textrm{ $|$ }\{{P}_{j}\}\textup{, $n_0+1$ qubits}\big)$ follows by summing over all possible values of $\V_1,\ldots,\V_{j_2-1}$.

Errors in class 2 yield $\overline{s}=\overline{0}$ with probability at most 3/4. This can be proven with the same arguments as for errors in class 1.

Errors in class 3 yield $\overline{s}=\overline{0}$ with probability at most 3/4. To see this, consider first the case where the errors affect neighboring bands ($j_2=j_1+1$), which yields
\begin{equation}
\textup{prob}\big(\overline{s}=\overline{0}\textrm{ $|$ }\{\p_{j}\}\big)=\frac{1}{2}\sum_{t=0,1}\frac{1}{N_{j_1}N_{j_1+1}}\sum_{\V_{j_1},\V_{j_1+1}}\langle+|^{\otimes n}\h^t\V_{j_1+1}^{-1} \p_{j_1+1}\V_{j_1+1}\V_{j_1}^{-1} \p_{j_1}\V_{j_1}\h^t|+\rangle^{\otimes n}\leq\frac{3}{4}
\end{equation}
As for errors in class 1, this can be shown by proving that the bound holds for the single-qubit case and the two-qubit one, and then using induction. If the errors affect two non-neighboring bands ($j_2\neq j_1+1$), we have
\begin{equation}
\textup{prob}\big(\overline{s}=\overline{0}\textrm{ $|$ }\{\p_{j}\}\big)=\frac{1}{2}\sum_{t=0,1}\frac{1}{N_{j_1}N_{j_2}}\sum_{\V_{j_1},\V_{j_2}}\langle+|^{\otimes n}\h^t\V_{j_2}^{-1} \p_{j_2}\V_{j_2}\bigg(\circ_{j=j_1+1}^{j_2-1}\C\X_{j}\bigg)\V_{j_1}^{-1} \p_{j_1}\V_{j_1}\h^t|+\rangle^{\otimes n}\leq\frac{3}{4}
\end{equation}
To prove the inequality, one can commute $\V_{j_1}^{-1} \p_{j_1}\V_{j_1}$ (which is a Pauli operator) with the entangling operation and use the same arguments as for $j_2=j_1+1$.

Finally, errors in class 4 yield 
\begin{align}
\textup{prob}\big(\overline{s}=\overline{0}\textrm{ $|$ }\{{P}_{j}\}\big)=&\frac{1}{2}\sum_{\substack{t=0,1}}\frac{1}{N_1\times\cdots\times N_{m-1}}\sum_{\substack{\V_{1},\cdots,\V_{m-1}}}\langle+|^{\otimes n}\p_m\h^t\circ\bigg(\circ_{j=1}^{m-1}
\C\X_{j}\bigg)\circ\h^t\p_{0}\big(\rho_{\textup{in}}^{\textup{trap}}
\big)|+\rangle^{\otimes n}\leq\frac{1}{2}
\end{align}
To see this, consider first the case $t=0$, and consider commuting $\p_{0}\in\{\I,\Z\}^{\otimes n}/\I$ with all the gates in the circuit. Since all these gates are $cX$ gates with random orientation, the identities
\begin{align}
\label{eq:zerrors}
cX\textrm{}(Z_1\otimes I_2)&=(Z_1\otimes I_2) cX\cr
cX(I_1\otimes Z_2)&=(Z_1\otimes Z_2) cX \cr
cX(Z_1\otimes Z_2)&=(I_1\otimes Z_2) cX
\end{align}
ensure that every time time that a Pauli-$Z$ error is commuted with a $cX$, this error becomes another error, chosen at random from two possible ones|Figure \ref{fig:single-band}. This can be used to prove that if $t=0$, errors in class 4 are detected with probability larger than 1/2. The same considerations apply to the case $t=1$, where the identities 
\begin{align}
\label{eq:xerrors}
cX\textrm{}(X_1\otimes I_2)&=(X_1\otimes X_2) cX \cr
cX(I_1\otimes X_2)&=(I_1\otimes X_2) cX\cr
cX(X_1\otimes X_2)&=(X_1\otimes I_2) cX
\end{align}
must be used instead of identities \ref{eq:zerrors}. 
\end{proof}

\begin{figure}[H]
	\centering
	\begin{tikzpicture}[scale=0.8, every node/.style={scale=0.75}]		
	
	\foreach \x in {1,2} 
	\node at (-0.6,3.0-1*\x) {\large $\ket{+}_{\x}$};

	\foreach \x in {1,...,2} 
	\draw (0.0-0.1,0.0+1*\x) -- (3.5,0.0+1*\x);
	
	\draw [fill=red!30] (0.2,1.7) rectangle (0.8,2.3);
	\node at (0.5,2.0) {$Z$};
	
	\draw [fill=ufogreen!30] (1.0,0.7) rectangle (2.2,2.3);
	
	\draw (1.3,0.8) -- (1.3,2.0);
	\draw (1.1,1) -- (1.5,1);
	\draw [fill=black] (1.3,2) circle [radius=0.07cm];
	\draw [] (1.3,1) circle [radius=0.2cm];
	
	\node at (1.6,1.5) {\large or};	
	
	\draw (1.9,1) -- (1.9,2.2);
	\draw (1.7,2) -- (2.1,2);
	\draw [fill=black] (1.9,1) circle [radius=0.07cm];
	\draw [] (1.9,2) circle [radius=0.2cm];

	\draw [fill=red!30] (1.5+0.9,1.7) rectangle (2.1+0.9,2.3);
	\node at (1.8+0.9,2.0) {$Z$};
	
	\foreach \x in {1,...,2} 
	\draw (2.7+0.9,0.0+1*\x-0.05) -- (3.2+0.9,0.0+1*\x-0.05);
	\foreach \x in {1,...,2} 
	\draw (2.7+0.9,0.0+1*\x+0.05) -- (3.2+0.9,0.0+1*\x+0.05);
	\foreach \x in {1,...,2} 
	\draw [fill=white] (2.7+0.9,0.0+1*\x) circle [radius=0.33cm];
	\foreach \x in {1,...,2} 
	\node at (2.7+0.9,0.0+1*\x) {$X$};
	
	\node at (3.5+0.9,1.0) {\large 0};
	\node at (3.5+0.9,2.0) {\large 0};
	\node at (4.0+0.9,1.5) {\large or};
	\node at (4.5+0.9,1.0) {\large 0};
	\node at (4.5+0.9,2.0) {\large 1};
	\end{tikzpicture}
	\caption{\footnotesize In this example, $\p_{0}=\p_{1}=Z_1$ (red gates) and $t=0$. Due to identities \ref{eq:zerrors}, commuting $\p_{1}$ with the entangling gate (green box, $cX$ gate with random orientation) make the two errors cancel out if qubit 1 is the control qubit. On the contrary, if qubit 1 is the target qubit, the errors do not cancel and cause a bit-flip of output $s_1$. Thus, for $t=0$ these errors are detected with probability 1/2. The same can be proven for $t=1$ using identities \ref{eq:xerrors}, as well as for all other errors $\p_{0},\p_{1}$.}
\label{fig:single-band}
\end{figure}
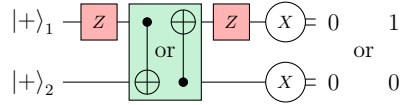

\section{Proof of Theorem 1}~\label{app:th1}
We start by using Lemma \ref{lem:pauliz} (Appendix \ref{lem:pauliz}), which allows to reduce noise of the type {N1} to classically correlated single-qubit Pauli errors and to write the joint state of target and trap circuits after all circuits have been implemented (for a fixed choice of single-qubit gates $\U_j^{(k)}$) as
\begin{align}
{\rho}_{\textup{out}}\big({\U}^{(1)}_{1},&\ldots,{\U}^{(v+1)}_{m}\big)
=\sum_{\substack{\overline{s}^{(1)},\ldots,\overline{s}^{(v+1)}}}\textup{ }\sum_{\substack{{\p}^{(1)}_{1},\ldots,{\p}^{(v+1)}_{m}}}\frac{
\textup{prob}\big({\p}^{(1)}_{1},\ldots,{\p}^{(v+1)}_{m}\big)}{2^{n(v+1)}}\textup{ }\times\cr
&\bigotimes_{k=1}^{v+1}\textup{ }\langle+|^{\otimes n}\bigg[\Z^{\overline{s}^{(k)}}{\p}^{(k)}_{m}{\U}^{(k)}_{m}\circ_{j=1}^{m-1}\bigg({\C\Z}_{j} {\p}^{(k)}_{j}{\U}^{(k)}_{j}\bigg)\circ{\p}^{(k)}_{0}\big(\rho_{\textup{in}}\big)\bigg]|+\rangle^{\otimes n}\bigg(\otimes_{i} Z_i^{s^{(k)}_i}|+\rangle_{i}\langle+|Z_i^{s^{(k)}_i}\bigg)
\end{align}
where $\rho_{\textup{in}}=\otimes_{i}|+\rangle_{i}\langle+|$, $\overline{s}^{(k)}=({s}_1^{(k)},\ldots,s^{(k)}_n)$ is the output of the $k$-th circuit, $\Z^{\overline{s}^{(k)}}(\rho)=\otimes_{i} Z_i^{s^{(k)}_i}\rho Z_i^{s^{(k)}_i}$, and $\textup{prob}\big({\p}^{(1)}_{1},\ldots,{\p}^{(v+1)}_{m}\big)$ is the joint probability of a collection of Pauli errors $\{{\p}^{(1)}_{1},\ldots,{\p}^{(v+1)}_{m}\}$ affecting the system, with ${\p}^{(k)}_{1},\ldots,{\p}^{(k)}_{m-1}\in\{\I,\X,\Y,\Z\}^{\otimes n}$ and ${\p}^{(k)}_{0},{\p}^{(k)}_{m}\in\{I,\Z\}^{\otimes n}$ for all $k$. Summing over all possible choices of single-qubit gates we thus obtain the state of target and traps at the end of the protocol:
\begin{align}
{\rho}_{\textup{out}}=&\textup{ }\sum_{{\U}^{(1)}_{1},\ldots,{\U}^{(v+1)}_{m}}
{\textup{prob}\big({\U}^{(1)}_{1},\ldots,{\U}^{(v+1)}_{m}\big)}\textup{ }\textup{ }{\rho}_{\textup{out}}{\big({\U}^{(1)}_{1},\ldots,{\U}^{(v+1)}_{m}\big)}\cr
=&\textup{ }\sum_{{\U}^{(1)}_{1},\ldots,{\U}^{(v+1)}_{m}}
{\textup{prob}\big({\U}^{(1)}_{1},\ldots,{\U}^{(v+1)}_{m}\big)}\sum_{\substack{\overline{s}^{(1)},\ldots,\overline{s}^{(v+1)}}}\textup{ }\sum_{\substack{{\p}^{(1)}_{1},\ldots,{\p}^{(v+1)}_{m}}}\frac{
\textup{prob}\big({\p}^{(1)}_{1},\ldots,{\p}^{(v+1)}_{m}\big)}{2^{n(v+1)}}\textup{ }\times\cr
&\bigotimes_{k=1}^{v+1}\textup{ }\langle+|^{\otimes n}\bigg[\Z^{\overline{s}^{(k)}}{\p}^{(k)}_{m}{\U}^{(k)}_{m}\circ_{j=1}^{m-1}\bigg({\C\Z}_{j} {\p}^{(k)}_{j}{\U}^{(k)}_{j}\bigg)\circ{\p}^{(k)}_{0}\big(\rho_{\textup{in}}\big)\bigg]|+\rangle^{\otimes n}\bigg(\otimes_{i} Z_i^{s^{(k)}_i}|+\rangle_{i}\langle+|Z_i^{s^{(k)}_i}\bigg)\textup{ ,}
\end{align}
where $\textup{prob}\big({\U}^{(1)}_{1},\ldots,{\U}^{(v+1)}_{m}\big)$ is the probability of single-qubit gates ${\U}^{(1)}_{1},\ldots,{\U}^{(v+1)}_{m}$ being chosen. Crucially, notice that the probability associated to each collection of Pauli errors does not depend on the specific choice of single-qubit gates $\U_j^{(k)}$. We can thus rewrite the above state as
\begin{align}
{\rho}_{\textup{out}}=&\sum_{\substack{\overline{s}^{(1)},\ldots,\overline{s}^{(v+1)}}}\textup{ }\sum_{\substack{{\p}^{(1)}_{1},\ldots,{\p}^{(v+1)}_{m}}}\frac{
\textup{prob}\big({\p}^{(1)}_{1},\ldots,{\p}^{(v+1)}_{m}\big)}{2^{n(v+1)}}\textup{ }\sum_{{\U}^{(1)}_{1},\ldots,{\U}^{(v+1)}_{m}}
{\textup{prob}\big({\U}^{(1)}_{1},\ldots,{\U}^{(v+1)}_{m}\big)}\textup{ }\times\cr
&\bigotimes_{k=1}^{v+1}\textup{ }\langle+|^{\otimes n}\bigg[\Z^{\overline{s}^{(k)}}{\p}^{(k)}_{m}{\U}^{(k)}_{m}\circ_{j=1}^{m-1}\bigg({\C\Z}_{j} {\p}^{(k)}_{j}{\U}^{(k)}_{j}\bigg)\circ{\p}^{(k)}_{0}\big(\rho_{\textup{in}}\big)\bigg]|+\rangle^{\otimes n}\bigg(\otimes_{i} Z_i^{s^{(k)}_i}|+\rangle_{i}\langle+|Z_i^{s^{(k)}_i}\bigg)
\end{align}
Consider now the state 
\begin{align}
{\rho}_{\textup{out}}\big({\p}^{(1)}_{1},&\ldots,{\p}^{(v+1)}_{m}\big)
=\sum_{\substack{\overline{s}^{(1)},\ldots,\overline{s}^{(v+1)}}}\textup{ }\frac{1}{2^{n(v+1)}}\textup{ }\sum_{{\U}^{(1)}_{1},\ldots,{\U}^{(v+1)}_{m}}
{\textup{prob}\big({\U}^{(1)}_{1},\ldots,{\U}^{(v+1)}_{m}\big)}\times\cr
&\bigotimes_{k=1}^{v+1}\textup{ }\langle+|^{\otimes n}\bigg[\Z^{\overline{s}^{(k)}}{\p}^{(k)}_{m}{\U}^{(k)}_{m}\circ_{j=1}^{m-1}\bigg({\C\Z}_{j} {\p}^{(k)}_{j}{\U}^{(k)}_{j}\bigg)\circ{\p}^{(k)}_{0}\big(\rho_{\textup{in}}\big)\bigg]|+\rangle^{\otimes n}\bigg(\otimes_{i} Z_i^{s^{(k)}_i}|+\rangle_{i}\langle+|Z_i^{s^{(k)}_i}\bigg)\textup{ }\textup{ }
\end{align}
corresponding to a fixed collection of Pauli errors $\{{\p}^{(1)}_{1},\ldots,{\p}^{(v+1)}_{m}\}$ and assume that the Pauli errors only affect $\widehat{v}\in\{1,\ldots,v+1\}$ circuits. In this case, as the Pauli errors do not depend on the single-qubit gates, they do not depend on the random number $v_0$ (which labels the position of the target circuit) nor on the random parameters in the trap circuits. Therefore, summing over all choices of single-qubit gates, the probability that the target is among the $\widehat{v}$ circuits affected by noise is $\widehat{v}/(v+1)$. 

Next using Lemma \ref{lem:ct} (which states that summed over all possible choices of single-qubit gates, trap circuits output $\overline{s}=\overline{0}$ with probability at most 3/4) we have that if $\widehat{v}-1$ trap circuits are affected by noise, they all output $\overline{s}=\overline{0}$ with probability at most $(3/4)^{\widehat{v}-1}$. We thus obtain
\begin{align}
{\rho}_{\textup{out}}{\big({\p}^{(1)}_{1},\ldots,{\p}^{(v+1)}_{m}\big)}
=&\widetilde{b}\textup{  }\widetilde{\tau}_{\textup{out}}^{\prime\textup{ tar}}\otimes\sigma_{\overline{s}=\overline{0}}^{\textup{trap}}+(1-\widetilde{b})\bigg(
\widetilde{l}\textup{ }\sigma_{\textup{out}}^{\textup{tar}}\otimes\sigma_{\overline{s}=\overline{0}}^{\textup{trap}}+(1-\widetilde{l})\textup{ }\widetilde{\tau}_{\textup{out}}^{\textup{ tar}}\otimes\sigma_{\overline{s}\neq\overline{0}}^{\textup{trap}}\bigg)
\end{align}
where $\widetilde{\tau}_{\textup{out}}^{\prime\textup{ tar}}$ ($\sigma_{\textup{out}}^{\textup{tar}}$) is the state of the target circuit if the target computation is (is not) among the $\widehat{v}$ traps affected by noise,  $\widetilde{\tau}_{\textup{out}}^{\textup{tar}}$ is an arbitrary state for the target, $\sigma_{\overline{s}=\overline{0}}^{\textup{trap}}$ is the state of the traps when they all output $\overline{s}=\overline{0}$,  $\sigma_{\overline{s}\neq\overline{0}}^{\textup{trap}}$ is an arbitrary state for the traps orthogonal to  $\sigma_{\overline{s}=\overline{0}}^{\textup{trap}}$, $0\leq \widetilde{l}\leq1$ and
\begin{equation}
\widetilde{b}\leq\frac{\widehat{v}}{v+1}\bigg(\frac{3}{4}\bigg)^{\widehat{v}-1}
\end{equation}
For all $v\geq3$ the RHS of the above upper-bound on $\widetilde{b}$ is maximized by $\widehat{v}=3$, which yields
\begin{align}
\widetilde{b}\leq\max_{\widehat{v}}\frac{\widehat{v}}{v+1}\bigg(\frac{3}{4}\bigg)^{\widehat{v}-1}=\frac{3}{v+1}\bigg(\frac{3}{4}\bigg)^2=\frac{\kappa}{v+1}\textrm{ , with }\kappa\approx\frac{1.7}{v+1}
\end{align}
As this is true for all collections of Pauli errors $\{{\p}^{(1)}_{1},\ldots,{\p}^{(v+1)}_{m}\}$, we can sum over all possible collections (weighted by their respective probabilities) and obtain
\begin{align}
\rho_{\textup{out}}=&\sum_{\substack{{\p}^{(1)}_{1},\ldots,{\p}^{(v+1)}_{m}}}{
\textup{prob}\big({\p}^{(1)}_{1},\ldots,{\p}^{(v+1)}_{m}\big)}\textup{ }\textup{ }{\rho}_{\textup{out}}{\big({\p}^{(1)}_{1},\ldots,{\p}^{(v+1)}_{m}\big)}\cr
=&\textup{ }b\textup{ }\tau_{\textup{out}}^{\prime\textup{ tar}}\otimes\sigma_{\overline{s}=\overline{0}}^{\textup{trap}}+(1-b)\bigg(l\textup{ }\sigma_{\textup{out}}^{\textup{tar}}\otimes\sigma_{\overline{s}=\overline{0}}^{\textup{trap}}+(1-l)\textup{ }\tau_{\textup{out}}^{\textup{tar}}\otimes\sigma_{\overline{s}\neq\overline{0}}^{\textup{trap}}\bigg)
\end{align}
with $b\leq\kappa/(v+1)$ and $0\leq l\leq1$. The Theorem is finally proven by replacing the state of the traps with that of the flag bit, namely by replacing $\sigma_{\overline{s}=\overline{0}}^{\textup{trap}}$ with $|\textrm{acc}\rangle$
and  $\sigma_{\overline{s}\neq\overline{0}}^{\textup{trap}}$ with $|\textrm{rej}\rangle$.

\section{Mesothetic Verification Protocol}
\noindent\textbf{Background: }
We will now define the notions of protocol and verifiability in the cryptographic setting. We start by defining quantum states as states belonging to the Hilbert space $\mathbb{H}_{ABC}=\mathbb{H}_A\otimes\mathbb{H}_B\otimes\mathbb{H}_C$, where $\mathbb{H}_A$ and $\mathbb{H}_B$ are Alice's and Bob's registers and $\mathbb{H}_C$ is a common register used to exchange qubits. The definition of protocol is as follows:
\begin{mydef}~\label{def:delegated2}
\textsc{\textbf{[Protocol]}} 
We define a $q$-step protocol on input $\rho_{\textup{in}}\in\pazocal{L}(\mathbb{H}_{ABC})$ as a series of CPTP maps $\{\E_{ABC}^{(p)}\}_{p=1}^q=\{\E_{AC}^{(p)}\circ\E_{BC}^{(p)}\}_{p=1}^q$ yielding the state $\rho_{\textup{out}}=\circ_{p=1}^q\E_{ABC}^{(p)}(\rho_{\mathrm{in}})$.
\end{mydef}
In our protocol Alice verifies the correct implementation of the target circuit through the trap circuits. If Bob is honest, all these circuits output a default output and the flag bit is set to $|\textup{acc}\rangle\in\{\ket{0},\ket{1}\}$, otherwise it is likely that they return an incorrect output and the flag bit is set to $|\textup{rej}\rangle=|\textup{acc}\oplus1\rangle$. We thus define verifiability as follows:
\begin{mydef}~\label{def:ver2}
{\textsc{\textbf{[Verifiability]}}
Consider a protocol $\{\E_{ABC}^{(p)}\}_{p=1}^q$ on input $\rho_{\textup{in}}$. The protocol is ``$\delta_{\textup{cr}}$-complete'' if
\begin{align}
D\bigg(\textup{Tr}_{BC}\bigg[\circ_{p=1}^q\E_{ABC}^{(p)}({\rho}_{\mathrm{in}})\bigg]\textup{, }{\sigma}_{\mathrm{out}}^{\mathrm{tar}}\otimes|\textup{acc}\rangle\langle\textup{acc}|\bigg)\leq1-\delta_{\textup{cr}}\textrm{ ,}
\end{align}
\noindent  where ${\sigma}_{\mathrm{out}}^{\mathrm{tar}}$ is the correct state of the target at the end of the protocol, $|\textup{acc}\rangle$ is the state of the state of the traps indicating acceptance and $0\leq\delta_{\textup{cr}}\leq1$ is called ``completeness''.

The protocol is ``$\varepsilon_{\textup{cr}}$-sound'' if for any set of CPTP maps $\{\widetilde{\E}_{BC}^{(p)}\}_{p=1}^q$ acting on Bob's register and on the channel, the state in Alice's register at the end of the protocol is}
\be
~\label{eq:trd2}
D\bigg(\textup{Tr}_{BC}\bigg[\circ_{p=1}^q
\big(\E_{ABC}^{(p)}\circ\widetilde{\E}_{BC}^{(p)}\big)({\rho}_{\mathrm{in}})\bigg]\textup{ , } l\textrm{ }{\sigma}_{\mathrm{out}}^{\mathrm{tar}}\otimes|\textup{acc}\rangle\langle\textup{acc}|   + (1-l)\tau_{\mathrm{out}}^{\mathrm{tar}}\otimes|\textup{rej}\rangle\langle\textup{rej}|\bigg)\leq\varepsilon_{\textup{cr}}
\ee
\noindent{where $0\leq\varepsilon_{\textup{cr}}\leq1$ is called ``soundness'', $0\leq l \leq1$, $\tau_{\mathrm{out}}^{\mathrm{tar}}$ is an arbitrary state and $|\textup{rej}\rangle$ is orthogonal to $|\textup{acc}\rangle$.}
\end{mydef}
We thus model Bob's deviations as arbitrary CPTP maps acting on the system being processed and on Bob's private register.

We can now explain how to translate our protocol in the cryptographic setting.\\

\noindent\textbf{Mesothetic protocol: }\label{app:crypto}
We assume that Alice, the verifier, and Bob, the prover, want to implement interactively a circuit of the type of Figure \ref{fig:circuit}, with $n$ qubits, depth $m$ and single-qubit gates $\{U_{i,j}\}$. We start by defining the resources required by Alice:\\

\noindent$\bullet$ Alice's resources: a device that can receive $n$ qubits from Bob, apply a single-qubit gate to each of them and send the qubits back to Bob. Alice's device must be able to execute all single-qubit gates contained in the target circuit, together with $H,S,S^{\dagger}$ (used in trap circuits) and $X,Y,Z$ (used for the QOTP). We assume that the overall number of these gates is $G$.\\

\noindent For the protocol to always accept, Bob must possess a quantum computer that can prepare qubits in the $\ket{+}$ state, execute $cZ$ gates, measure in the Pauli-$X$ basis and store qubits. 

We can now explain our protocol, which is formally presented in Box 4 in the Methods. Our protocol takes as input a description of the target circuit and the number $v$ of trap circuits, both inputs also being known to Bob. Steps \textbf{1, 2} and \textbf{3} are classical and solely involve Alice's register, which at the end of step \textbf{3} (for a fixed choice of single-qubit gates ${\U}^{\prime\prime(1)}_{1},\ldots,{\U}^{\prime\prime(v+1)}_{m}$) is in the state
\begin{align}
\rho_{A,\textup{in}}\big({\U}^{\prime\prime(1)}_{1},\ldots,{\U}^{\prime\prime(v+1)}_{m}\big)=&\bigg[\otimes_{k=1}^{v+1}\otimes_{j=1}^m\otimes_{i=1}^n|U_{i,j}^{\prime\prime(k)}\rangle\langle U_{i,j}^{\prime\prime(k)}|\bigg]\otimes\bigg[\otimes_{k=1}^{v+1}\otimes_{i=1}^n|\alpha_{i,m}^{(k)}\rangle\langle \alpha_{i,m}^{(k)}|\bigg]
\end{align}
Here, $|U_{i,j}^{\prime\prime(k)}\rangle$ is a classical description of the gate $U_{i,j}^{\prime\prime(k)}$. Since Alice's device can execute $G$ different gates, this classical description requires at most log${}_2(G)$ bits. Moreover, $|\alpha_{i,m}^{(k)}\rangle$ is a classical description of the bit $\alpha_{i,m}^{(k)}$.

In step \textbf{4} Alice and Bob interact as follows. For all circuits $k=1,\ldots,v+1$, in step \textbf{4.1} Bob creates $n$ qubits in the state $\ket{+}$. Then, for each band $j=1,\ldots,m$, Bob sends all the qubits to Alice, Alice acts on each qubit $i$ with $U_{i,j}^{(k)}$, Alice sends all the qubits back to Bob and Bob implements $\widehat{cZ}_{j}$. Finally, in step \textbf{4.3} Bob measures all the qubits and returns the outputs to Alice. The random but $\alpha_{i,m}^{(k)}$ coming from the QOPT in the last round of single-qubit gates bit-flips the outputs, therefore Bob transmits to Alice the bits $s_{i}^{(k)}\oplus\alpha_{i,m}^{(k)}$. Alice post-processes those bits and obtains $s_{i}^{(k)}$. At the end of step \textbf{4}, the states in Alice and Bob's registers (if Bob is honest) are of the form
\begin{align}
\rho_{A,\textup{out}}&\big({\U}^{\prime\prime(1)}_{1},\ldots,{\U}^{\prime\prime(v+1)}_{m}\big)=\bigg[\otimes_{k=1}^{v+1}\otimes_{j=1}^m\otimes_{i=1}^n|U_{i,j}^{\prime\prime(k)}\rangle\langle U_{i,j}^{\prime\prime(k)}|\bigg]\otimes\bigg[\otimes_{k=1}^{v+1}\otimes_{i=1}^n|s_{i}^{(k)}\rangle\langle s_{i}^{(k)}|\bigg]\cr
\rho_{B,\textup{out}}&=\otimes_{k=1}^{v+1}\otimes_{i=1}^n Z_i^{s^{(k)}_i\oplus\alpha_{i,m}^{(k)}}|+\rangle^{(k)}_{i}\langle+|Z_i^{s^{(k)}_i\oplus\alpha_{i,m}^{(k)}}
\end{align}
The protocol ends with Alice checking whether the traps yield the correct output and setting the flag bit to $|\textup{acc}\rangle$ or $|\textup{rej}\rangle$.

We will now prove completeness and soundness for our mesothetic protocol (we refer the reader to the beginning of the Appendix for notation). We begin with the assumption that Alice's device is noiseless. Soundness requires the following Lemmas:\\

\noindent\textbf{Lemma {D1}. }\textit{Suppose that Alice's device is noiseless. For a fixed choice of single-qubit gates ${\U}^{(1)}_{1},\ldots,{\U}^{(v+1)}_{m}$ by Alice, tracing out Bob's private register and summing over all the random bits $\alpha^{(k)}_{i,j},\alpha^{\prime(k)}_{i,j},\gamma^{(k)}_i$ except $\alpha^{(k)}_{i,m}$ (cfr. Routine 1), the state in Bob's register at the end of the protocol is of the form}
\begin{align}~\label{eq:rhooutbob}
{\rho}_{B,\textup{out}}\big({\U}^{(1)}_{1},&\ldots,{\U}^{(v+1)}_{m}\big)
=\sum_{\substack{\overline{s}^{(1)},\ldots,\overline{s}^{(v+1)}\\{\p}^{(1)}_{1},\ldots,{\p}^{(v+1)}_{m}}}\frac{
\textup{prob}\big({\p}^{(1)}_{1},\ldots,{\p}^{(v+1)}_{m}\big)}{2^{n(v+1)}}\textup{ }\cr
&\bigotimes_{k=1}^{v+1}\textup{ }\langle+|^{\otimes n}\bigg[\Z^{\overline{s}^{(k)}}{\p}^{(k)}_{m}{\U}^{(k)}_{m}\circ_{j=1}^{m-1}\bigg({\C\Z}_{j} {\p}^{(k)}_{j}{\U}^{(k)}_{j}\bigg)\circ{\p}^{(k)}_{0}\big(\rho_{\textup{in}}\big)\bigg]|+\rangle^{\otimes n}\times\bigg(\otimes_{i} Z_i^{s^{(k)}_i\oplus\alpha^{(k)}_{i,m}}|+\rangle_{i}\langle+|Z_i^{s^{(k)}_i\oplus\alpha^{(k)}_{i,m}}\bigg)\textrm{ ,}\cr
\end{align}
\textit{where $\rho_{\textup{in}}=\otimes_{i}|+\rangle_{i}\langle+|$, $\overline{s}^{(k)}=({s}_1^{(k)},\ldots,s^{(k)}_n)$ is the output of the $k$-th circuit, $\Z^{\overline{s}^{(k)}}(\rho)=\otimes_{i} Z_i^{s^{(k)}_i}\rho Z_i^{s^{(k)}_i}$ and $\textup{prob}\big({\p}^{(1)}_{1},\ldots,{\p}^{(v+1)}_{m}\big)$ is the joint probability of a collection of Pauli errors ${\p}^{(1)}_{1},\ldots,{\p}^{(v+1)}_{m}$ affecting the system, with $\p^{(k)}_{1},\ldots,\p^{(k)}_{m-1}\in\{\I,\X,\Y,\Z\}^{\otimes n}$ and $\p^{(k)}_{0},\p^{(k)}_{m}\in\{\I,\Z\}^{\otimes n}$ for all $k$.}\\

\begin{proof}\textit{(Sketch). }
We start proving the Lemma for a protocol with a single circuit ($v=1$), next we generalize to $v>1$. For a fixed choice of gates $U_{i,j}^{\prime\prime}$ by Alice, the joint state of Alice's register $A$ and Bob's register $B$ before the final measurements is of the form 
\begin{align}
\rho_{AB}\big({\U}^{\prime\prime(1)}_{1},\ldots,{\U}^{\prime\prime(v+1)}_{m}\big)=\textup{ }&\bigg[\otimes_{j=1}^m\otimes_{i=1}^n|U_{i,j}^{\prime\prime}\rangle\langle U_{i,j}^{\prime\prime}|\bigg]\otimes\bigg[\otimes_{k=1}^{v+1}\otimes_{i=1}^n|\alpha_{i,m}\rangle\langle \alpha_{i,m}|\bigg]\otimes&\textup{(Alice's register)}\cr
&\bigg[\widetilde{\E}_{m}\U^{\prime\prime}_{m}\circ
\bigg(\circ_{j=1}^{m-1}\widetilde{\E}_{j}\U^{\prime\prime}_{j}\bigg)\circ\E_{0}
\big(\rho^{}_{\textup{in}}\big)\bigg]\textup{, }&\textup{(Bob's register)}
\end{align}
where $\widetilde{\E}_{j}$ are Bob's deviations. Specifically, $\widetilde{\E}_{0}$ is Bob's deviation when he prepares $\rho_{\textup{in}}$ (step \textbf{4.1}), for all $j=1,\ldots,m-1$, $\widetilde{\E}_{j}$ is Bob's deviations when in step \textbf{4.2} he should execute $\widehat{cZ}_{j}$, $\widetilde{\E}_{m}$ is Bob's deviation before he measures the qubits (step \textbf{4.3}). Without loss of generality, suppose now that Bob holds another register, which we call ``ancillary register'' and denote with $B_{\textup{anc}}$. Tensoring $\rho_{\textup{in}}$ with the state $\sigma_{\textup{anc}}$ in the register $B_{\textup{anc}}$, for $j\in\{1,\ldots,m-1\}$ we can rewrite Bob's deviations $\widetilde{\E}_{j}$ as unitaries ${\mathbf{F}}_{j}\widehat{cZ}_{j}$, where ${\mathbf{F}}_{j}$ is a unitary matrix that acts on Bob's register $B$ and on $B_{\textup{anc}}$|for convenience we indicate unitary matrices acting both on Bob's register $B$ and on $B_{\textup{anc}}$ in bold font. Similarly, we can replace $\widetilde{\E}_{0}$ with the unitary $\mathbf{R}$ and $\widetilde{\E}_{m}$ with the unitary $\mathbf{M}$. We thus obtain
\begin{align}
\rho_{ABB_{\textup{anc}}}\big({\U}^{\prime\prime(1)}_{1},&\ldots,{\U}^{\prime\prime(v+1)}_{m}\big)=\textup{ }\bigg[\otimes_{j=1}^m\otimes_{i=1}^n|U_{i,j}^{\prime\prime}\rangle\langle U_{i,j}^{\prime\prime}|\bigg]\otimes\bigg[\otimes_{k=1}^{v+1}\otimes_{i=1}^n|\alpha_{i,m}\rangle\langle \alpha_{i,m}|\bigg]\otimes\cr
&\bigg[\mathbf{M}\textup{ }{U}^{\prime\prime}_{m}
{\mathbf{F}}_{m-1}\widehat{cZ}_{m-1}^{}{U}^{\prime\prime}_{m-1}\ldots{\mathbf{F}}_{1}\widehat{cZ}_{1}^{}{U}^{\prime\prime}_{1}{\mathbf{R}}
\bigg(\rho_{\textup{in}}\otimes\sigma_{\textup{anc}}\bigg)
\mathbf{R}^\dagger{U}^{\prime\prime\dagger}_{1}\widehat{cZ}_{1}{\mathbf{F}}_{1}^\dagger\ldots{U}^{\prime\prime\dagger}_{m-1}\widehat{cZ}_{m-1}{\mathbf{F}}_{m-1}^\dagger{U}^{\prime\prime\dagger}_{m}\mathbf{M}^\dagger
\bigg]\textup{ }\textup{ }&
\end{align}
Tracing out Alice's register and $B_{\textup{anc}}$ yields 
\begin{align}
&\textup{Tr}_{AB_{\textup{anc}}}\big[\rho_{ABB_{\textup{anc}}}\big({\U}^{\prime\prime(1)}_{1},\ldots,{\U}^{\prime\prime(v+1)}_{m}\big)\big]\cr
&=\textup{Tr}_{B_{\textup{anc}}}\bigg[\mathbf{M}\textup{ }{U}^{\prime\prime}_{m}
{\mathbf{F}}_{m-1}\widehat{cZ}_{m-1}^{}{U}^{\prime\prime}_{m-1}\ldots{\mathbf{F}}_{1}\widehat{cZ}_{1}^{}{U}^{\prime\prime}_{1}{\mathbf{R}}
\bigg(\rho_{\textup{in}}\otimes\sigma_{\textup{anc}}\bigg)
\mathbf{R}^\dagger{U}^{\prime\prime\dagger}_{1}\widehat{cZ}_{1}{\mathbf{F}}_{1}^\dagger\ldots{U}^{\prime\prime\dagger}_{m-1}\widehat{cZ}_{m-1}{\mathbf{F}}_{m-1}^\dagger{U}^{\prime\prime\dagger}_{m}\mathbf{M}^\dagger\bigg]
\end{align}
This state is equal to the state $\rho$ in Equation \ref{eq:rhomultiplebands}, provided that we replace $\sigma_{\textup{anc}}$ with $\rho_E$ and $\textup{Tr}_{B_{\textup{anc}}}$ with Tr$_{E}[\textup{ }\cdot\textup{ }]$, hence the Lemma can be proven repeating the same calculations. The reader can verify that the same applies to the case $v>1$.
\end{proof}

\noindent\textbf{Lemma {D2}. }
\textit{For any collection of Pauli errors affecting a trap circuit, summing over all possible single-qubit gates in the trap circuit (i.e. over all possible sets $\{U_{i,j}\}$ output by Routine 2), the probability that the trap circuit outputs $\overline{s}=\overline{0}$ is at most 3/4.}\\

As Alice choses all the gates in a trap circuits with Routine 2, the proof is the same as that of Lemma \ref{lem:ct}. Using these two Lemmas we now compute the classical state in Alice's register after all the circuits have been implemented:\\

\noindent\textbf{Lemma {D3}. }
\textit{Suppose that Alice's device is noiseless. Then, for any number $v\geq3$ of trap computations, the state in Alice's register at the end of the protocol is of the form (see Definition \ref{def:ver2} for notation)}
\begin{align}~\label{eq:end2crypto}
{\rho}_{\textup{out}}\textup{= }&b\textup{ }{\tau}_{\textup{out}}^{\prime\textup{ tar}}\otimes|\textup{acc}\rangle\langle\textup{acc}|+(1-b)\bigg(l\textup{ }\textrm{ }{\sigma}_{\textup{out}}^{\textup{ tar}}\otimes|\textup{acc}\rangle\langle\textup{acc}|\textup{+ }(1-l){\tau}_{\textup{out}}^{\textup{ tar}}\otimes|\textup{rej}\rangle\langle\textup{rej}|\bigg)\textrm{ ,}
\end{align}
\textit{where ${\tau}_{\textup{out}}^{\prime\textup{ tar}}$ is an incorrect state for the target, $0\leq b\leq \kappa/(v+1)$ and $\kappa = 3(3/4)^2 \approx 1.7$.}\\
\begin{proof}
After Bob has implements all the circuits and communicates all the outputs to Alice, Alice holds the same classical state as Bob (Equation \ref{eq:rhooutbob}). After Alice post-processes the outputs this state becomes (cfr. Lemma {D1} for notation)
\begin{align}
{\rho}_{A,\textup{out}}\big({\U}^{(1)}_{1},&\ldots,{\U}^{(v+1)}_{m}\big)
=\sum_{\substack{\overline{s}^{(1)},\ldots,\overline{s}^{(v+1)}\\{\p}^{(1)}_{1},\ldots,{\p}^{(v+1)}_{m}}}\frac{
\textup{prob}\big({\p}^{(1)}_{1},\ldots,{\p}^{(v+1)}_{m}\big)}{2^{n(v+1)}}\textup{ }\cr
&\bigotimes_{k=1}^{v+1}\textup{ }\langle+|^{\otimes n}\bigg[\Z^{\overline{s}^{(k)}}{\p}^{(k)}_{m}{\U}^{(k)}_{m}\circ_{j=1}^{m-1}\bigg({\C\Z}_{j} {\p}^{(k)}_{j}{\U}^{(k)}_{j}\bigg)\circ{\p}^{(k)}_{0}\big(\rho_{\textup{in}}\big)\bigg]|+\rangle^{\otimes n}\times\bigg(\otimes_{i} Z_i^{s^{(k)}_i}|+\rangle_{i}\langle+|Z_i^{s^{(k)}_i}\bigg)\textrm{ ,}\cr
\end{align}
Crucially, notice that the probability associated to each deviation ${\p}^{(1)}_{1},\ldots,{\p}^{(v+1)}_{m}$ does not depend on the choice of single-qubit gates made by Alice. Thus, summing over all choices of single-qubit gates, this state becomes
\begin{align}
{\rho}_{A,\textup{out}}
=&\sum_{\substack{\overline{s}^{(1)},\ldots,\overline{s}^{(v+1)}\\{\p}^{(1)}_{1},\ldots,{\p}^{(v+1)}_{m}}}\frac{
\textup{prob}\big({\p}^{(1)}_{1},\ldots,{\p}^{(v+1)}_{m}\big)}{2^{n(v+1)}}\textup{ }\sum_{{\U}^{(1)}_{1},\ldots,{\U}^{(v+1)}_{m}}\textup{prob}\big({\U}^{(1)}_{1},\ldots,{\U}^{(v+1)}_{m}\big)\cr
&\bigotimes_{k=1}^{v+1}\textup{ }\langle+|^{\otimes n}\bigg[\Z^{\overline{s}^{(k)}}{\p}^{(k)}_{m}{\U}^{(k)}_{m}\circ_{j=1}^{m-1}\bigg({\C\Z}_{j} {\p}^{(k)}_{j}{\U}^{(k)}_{j}\bigg)\circ{\p}^{(k)}_{0}\big(\rho_{\textup{in}}\big)\bigg]|+\rangle^{\otimes n}\times\bigg(\otimes_{i} Z_i^{s^{(k)}_i}|+\rangle_{i}\langle+|Z_i^{s^{(k)}_i}\bigg)\textrm{ ,}\cr
\end{align}
The Lemma can now be proven following the same steps as in the proof of Theorem \ref{th:verif1}.
\end{proof}
Using these three Lemmas we can now calculate completeness and soundness for the mesothetic protocol and obtain Theorem {D1}:\\

\noindent{\textbf{Theorem D1. }}\textit{Suppose that Alice's device is noiseless. Then, for any number $v\geq3$ of trap circuits, the mesothetic protocol is verifiable with}
\begin{equation}
\delta_{\textup{cr}}=1\textup{ }\textup{ and }\textup{ }\varepsilon_{\textup{cr}}=\frac{\kappa}{v+1}\textup{ ,}
\label{eq:eps}
\end{equation}
\textit{where $\kappa = 3(3/4)^2 \approx 1.7.$}
\begin{proof}
Completeness can be proven trivially showing that for all circuits, the random Pauli gates used for the QOTP cancel between them, and that the effects of the QOTP in the last band are recovered by the classical post-processing made by Alice. Therefore, all the circuits implement the desired computation, and a correctly implemented protocol yields the correct output for the target and the state of the traps indicating acceptance.

To prove soundness we need to compute the trace distance in Equation \ref{eq:trd2}. Using Lemma {D3} this trace distance becomes
\be
D\bigg(b\textup{ }{\tau}_{\textup{out}}^{\prime\textup{ tar}}\otimes|\textup{acc}\rangle\langle\textup{acc}|+(1-b)\big[l\textup{ }\textrm{ }{\sigma}_{\textup{out}}^{\textup{ tar}}\otimes|\textup{acc}\rangle\langle\textup{acc}|\textup{+ }(1-l){\tau}_{\textup{out}}^{\textup{ tar}}\otimes|\textup{rej}\rangle\langle\textup{rej}|\big]\textup{ , } l\textrm{ }{\sigma}_{\mathrm{out}}^{\mathrm{tar}}\otimes|\textup{acc}\rangle\langle\textup{acc}|   + (1-l)\tau_{\mathrm{out}}^{\mathrm{tar}}\otimes|\textup{rej}\rangle\langle\textup{rej}|\bigg)\leq b
\ee
and since $b\leq\kappa/(v+1)$ by Lemma {D3}, we obtain soundness $\varepsilon_{\textup{cr}}=\kappa/(v+1)$.
\end{proof}

We can also compute completeness and soundness in the case where Alice's device suffers bounded and potentially gate-dependent noise:\\

\noindent{\textbf{Theorem D2. }}\textit{Suppose that Alice's device is affected by bounded noise, i.e. that for all circuits $k\in\{1,\ldots,v+1\}$ and bands $j\in\{1,\ldots,m\}$ she applies $\E_{j}^{(k)}{U}_{j}^{(k)}$, where $\E_{j}^{(k)}=(1-r_{j}^{(k)})\I+r_{j}^{(k)}\E_{j}^{\prime \textup{ }(k)}$ for some arbitrary CPTP-map $\E_{j}^{\prime\textup{ }(k)}$ and number $0\leq r_{j}^{(k)}<1$. Then, for any number $v\geq3$ of trap computations, the mesothetic protocol is verifiable with}
\begin{align}
\delta_{\textup{cr}}=1\textup{ }\textup{ and }\textup{ }\varepsilon_{\textup{cr}}=g\frac{\kappa}{v+1}+\textup{ }1\textup{ }-g\textrm{ ,}
\end{align}
\textit{where $\kappa = 3(3/4)^2 \approx 1.7$, $g=\prod_{j,k}1-r_{\textup{max, }j}^{(k)}$ and $r_{\textup{max, }j}^{(k)}$ is the maximum error rate of the round of gates in band $j$ of circuit $k$, where this maximum is taken over all choices gates for this round.}\\
\begin{proof}
The proof of completeness is the same as for Theorem {D1}. To compute soundness, we denote as ${\rho}_{\textup{out}}^\star$ the state in Alice's register at the end of a protocol run when Alice's device is noisy, and as ${\rho}_{\textup{out}}$ the state in Alice's register at the end of a protocol run when Alice's device is noiseless (Lemma {D3}). Indicating as $r_{\textup{max, }j}^{(k)}$ the maximum error rate for gates in band $j$, we rewrite this noisy map as $\E_{j}^{(k)}=(1-r_{\textup{max, }j}^{(k)})I+r_{\textup{max, }j}^{(k)}\Q_{j}^{(k)}$ for some other CPTP map $\Q_{j}^{(k)}$. This allows to rewrite the classical state in Alice's register at the end of the protocol as $g{\rho}_{\textup{out}}+(1-g){\rho}_{\textup{out}}^\star$ and to obtain the upper-bound the trace distance.
\end{proof}

We conclude this Appendix by showing how our protocol can be made blind. Blindness is a property exhibited by many cryptographic protocols \cite{GKK17} defined as follows:
\begin{mydef}\textup{\textbf{[Blindness.] }}A protocol $\{\E^{(p)}_{ABC}\}=\{\E^{(p)}_{AB}\otimes\E^{(p)}_{BC}\}$ on input $\rho_{\textup{in}}$ is blind if for any set of maps $\{\widetilde{\E}^{(p)}_{BC}\}$ acting on Bob's register and on the channel, the state in Bob's register at the end of the protocol leaks at most a constant function of the input.
\end{mydef}
\noindent Typically blind cryptographic protocols leak an upper-bound on the size of the target circuit.

As introduced above, our mesothetic protocol is not blind. Indeed, Bob has access to non-trivial information about the target circuit, such as the position of $cZ$ gates. This is not a problem for our scopes, as verifiability only relies on Bob's ignorance of the number $v_0$, which is kept secret by Alice with the QOTP. Nevertheless, blindness may be required for maintaining the privacy of the users in future scenarios of delegated computations, thus it is important to understand whether our protocol can be turned into a blind protocol. 

To obtain blindness, Alice (endowed with an ideal device) must recompile the target circuit into a normal form, for instance by recompiling the target circuit into a circuit of the type illustrated in Figure \ref{fig:BwSType} (which inspired to Brickwork States \cite{BFK09}). Then, instead of giving as input to the mesothetic protocol a classical description of the target circuit, she only needs to input the desired size of the circuit in normal form ($n$ qubits, depth $m$). 

Implementing target and traps on a circuit in normal form makes Alice's instructions to Bob independent from the specific target circuit that Alice wishes to verify. This allows to prove blindness:\\

\noindent{\textbf{Theorem D3.  [Blindness with Circuit in Normal Form].}} \textit{Suppose that Alice can apply noiseless single-qubit gates. If Alice rewrites the target circuit in normal form, the mesothetic protocol (with input the desired size of the circuit in normal form) only leaks an upper-bound on the size of the target circuit.}
\begin{proof}
To prove blindness we notice that during the protocol run Bob cannot retrieve any information about the the target circuit. Indeed, Bob's tasks are the same for all circuits (prepare the same input state, execute the same entangling gates and measure in the same basis) and these tasks do not depend on the target circuit, since this target is implemented on a circuit in normal form. Moreover, the only type of information that Bob receives from Alice during the implementation of the circuits are the qubits at step \textbf{4.2.1}, but the QOTP prevents Bob from retrieving useful information: at any $k=1,\ldots,v+1$ and $j=1,\ldots,m$, if Bob sends to Alice a state $\rho_{j}^{(k)}$, Alice returns to him the state
\begin{equation}
\label{eq:blind}
{U}_{j}^{\prime\prime(k)}\rho_{j}^{(k)}{U}_{j}^{\dagger\prime\prime(k)}=
\otimes_{i=1}^nZ^{\alpha^{(k)}_{i,j}}X^{\alpha^{\prime(k)}_{i,j}}\bigg[{U}_{j}^{(k)}{P}_{j-1}
\rho_{j}^{(k)}{P}_{j-1}{U}_{j}^{\dagger(k)}\bigg]X^{\alpha^{\prime(k)}_{i,j}}Z^{\alpha^{(k)}_{i,j}}\textup{ ,}
\end{equation}
where ${P}_{j-1}$ is the Pauli operator that undoes previous QOTP. Summing over all possible $\alpha_{i,j}^{(k)}$ and $\alpha_{i,j}^{\prime(k)}$ yields the completely mixed state.
\end{proof}

\end{document}